\documentclass[prb,twocolumn,superscriptaddress,nofootinbib,preprintnumbers]{revtex4-2}

\usepackage{amsmath,amssymb,amsfonts,amsthm}
\usepackage{graphicx}
\usepackage{hyperref}
\usepackage{braket}
\usepackage{xcolor}
\usepackage{booktabs}
\usepackage{multirow}
\usepackage{makecell}
\usepackage{algorithm,algpseudocode}
\usepackage{tikz}
\usetikzlibrary{decorations.pathreplacing, positioning, fit,arrows.meta, patterns, shapes.geometric}

\newcommand{\eg}{\textit{e.g.}}
\newcommand{\vs}{\textit{vs.}}

\newcommand{\kry}{\mathcal{K}}

\newcommand{\Kcl}{\kry_{\rm cl}}
\newcommand{\Kq}{\kry_q}
\newcommand{\KEFS}{\kry_{\rm EFS}}
\newcommand{\Fprod}{\mathcal{F}_{\rm prod}}
\newcommand{\Fent}{\mathcal{F}_{\rm ent}}
\newcommand{\Ftot}{\mathcal{F}_{\rm frz}^{\rm tot}}

\newtheorem{theorem}{Theorem}
\newtheorem{lemma}{Lemma}

\begin{document}

\title{Quantum Hilbert Space Fragmentation and Entangled Frozen States}

\author{Zihan Zhou}
\email{zihanz@princeton.edu}
\thanks{These authors contributed equally to this work.}
\affiliation{Department of Physics, Princeton University, NJ 08544, USA}

\author{Tian-Hua Yang}
\email{yangth@princeton.edu}
\thanks{These authors contributed equally to this work.}
\affiliation{Department of Physics, Princeton University, NJ 08544, USA}

\author{Bo-Ting Chen}
\email{bc2490@princeton.edu}
\affiliation{Department of Physics, Princeton University, NJ 08544, USA}

\date{\today}

\begin{abstract}
We find that rank deficiency of the local Hamiltonian in a classically fragmented model is the key mechanism leading to quantum Hilbert space fragmentation. The rank deficiency produces local null directions that can generate entangled frozen states (EFS): entangled states embedded in mobile classical Krylov sectors that do not evolve under Hamiltonian dynamics. When the entangled frozen subspace is non-empty, the mobile classical sector splits into a mobile quantum Krylov subspace and an entangled frozen subspace, and the model exhibits quantum fragmentation. We establish this mechanism in four models of increasing symmetry structure: an asymmetric qubit projector with no symmetry, the $\mathbb{Z}_2$-symmetric GHZ projector, a $\mathbb{Z}_3$-symmetric cyclic qutrit projector, and the Temperley-Lieb model. For the asymmetric and GHZ projector models, we obtain closed-form expressions for irreducible Krylov dimensions, degeneracies, and sector multiplicities. The all-mobile-sector EFS in these two models exhibits a sub-volume-law bipartite entanglement entropy scaling as $S \sim \sqrt{L}$. Further, we introduce the notion of weak and strong quantum fragmentation, the quantum counterpart of the weak-strong distinction in classical fragmentation. After removing the EFS, the mobile quantum Krylov subspace decomposes into irreducible blocks. In the weak case, the number of irreducible blocks remains $\mathcal{O}(1)$, each is individually ergodic with Gaussian Orthogonal Ensemble (GOE) level statistics, and the unresolved spectrum follows an $m$GOE distribution. In the strong case, the number of irreducible blocks grows with system size, and the gap-ratio distribution approaches Poisson as $L\to\infty$.
\end{abstract}

\maketitle

\emph{Introduction. --- }
Hilbert space fragmentation (HSF)~\cite{sala2020ergodicity,Khemani:2019vor,Moudgalya:2021xlu,Moudgalya:2019vlp,Moudgalya:2021ixk,Regnault:2022ocy,Iadecola:2025fyg,Brighi:2022kca,Chen:2024fkk,Chen:2026aqj,Rakovszky:2020apb,Yang:2019mft,Morningstar:2020ror,Balasubramanian:2023wgn,Kwan:2023kjp,hartExactMazurBounds2024,Brighi:2022kca,ganguli2025aspects,PhysRevB.108.144308,wang2025exponentiallyslowthermalization1d,PhysRevB.109.064302,PhysRevB.106.L220301} is a mechanism for ergodicity breaking in which the Hilbert space decomposes into exponentially many dynamically disconnected Krylov subspaces, beyond what conventional symmetries dictate. Unlike many-body localization~\cite{pal2010many,Nandkishore:2014kca,vosk2015theory,nandkishore2015many,Abanin:2018yrt}, which relies on disorder, and quantum many-body scars~\cite{shiraishi2017systematic,Bernien:2017ubn,Turner:2017fxc,Turner:2018gwa,Moudgalya:2018fua,moudgalya2020quantum,ren2021quasisymmetry}, which arise from fine-tuned eigenstates, HSF is generated by local algebraic constraints and persists at all energy scales. HSF comes in two varieties: classical and quantum fragmentation~\cite{Moudgalya:2021ixk,Regnault:2022ocy,Chen:2026aqj,PhysRevX.15.011068,PhysRevResearch.5.043239,wbzt-scvs}. In classical fragmentation, Krylov sectors are spanned by product states. A systematic understanding of this case has been achieved through semigroup word problems~\cite{Balasubramanian:2023wgn}, where the glassy dynamics of local rewriting rules generates a combinatorial structure of frozen states and mobile sectors. Classical fragmentation is classified as strong if the largest irreducible Krylov subspace $\mathcal{K}_{{\rm max},L}$ has measure zero compared relative to the full Hilbert space in the thermodynamic limit, i.e.\ ${\dim} \mathcal{K}_{{\rm max},L}/ {\dim}\, \mathcal{H}\to 0$ as $L\to \infty$ and weak if ${\dim} \mathcal{K}_{{\rm max},L}/ {\dim}\, \mathcal{H}\sim \mathcal{O}(1)$.

Quantum fragmentation goes further: it decomposes the classical Krylov sectors themselves into sub-sectors spanned by entangled states. This phenomenon is far less understood. The only well-studied example is the Temperley-Lieb (TL) model~\cite{Moudgalya:2021ixk,Moudgalya:2022gtj,Moudgalya:2022nll}, where the TL algebra ${\rm TL}_{L-1}(\delta)$~\cite{Temperley:1971iq,Jones:1983kv} and its $U_q(\mathfrak{sl}_2)$ quantum group commutant provide closed-form expressions for Krylov dimensions and degeneracies. Quantum fragmentation has also been observed in the quantum breakdown model~\cite{Lian:2022nqj,Chen:2024fkk,Hu_2024,Hu_2025,Liu2025QBM,hu2025glass} and the quantum East model~\cite{Brighi:2022kca,PhysRevB.108.144308,ganguli2025aspects}, though without a comparably developed analytical framework. Since the TL model remains the only fully understood example, two fundamental questions are open. First, what is the minimal ingredient for quantum fragmentation? The TL model carries so much algebraic structure (the Jones relation, a quantum group commutant) that it is unclear which part is essential and which is incidental. Second, what role do symmetries play? Symmetries act locally and carry distinct physical consequences from the general (non-local) commutant algebra, yet whether they are a prerequisite for quantum fragmentation, or merely an optional add-on, has not been established.

In this Letter, we answer both questions. We find that the key mechanism leading to quantum fragmentation is the rank deficiency of local coupling terms in a classically fragmented model.
No symmetry is required. The rank deficiency creates \emph{entangled frozen states} (EFS): entangled states within mobile classical Krylov sectors that do not evolve under the Hamiltonian dynamics. Each mobile classical sector then splits into an entangled-frozen Krylov subspace $\KEFS$ and a complementary mobile quantum Krylov subspace $\Kq$. Symmetry does not create quantum fragmentation, but when present, it can further decompose the mobile quantum Krylov subspaces if they are closed under the symmetry transformation.

We establish this picture in four models of increasing structure. The asymmetric triplet-flip projector provides the minimal example, showing that quantum fragmentation exists even without symmetry. The GHZ and cyclic qutrit projectors then show how $\mathbb{Z}_2$ and $\mathbb{Z}_3$ symmetries organize the mobile quantum space $\Kq$ through symmetry-related degeneracies and charge sectors, without changing the underlying mechanism. In both the asymmetric and GHZ projector models, the all-mobile-sector EFS displays an unusual analytically tractable sub-volume entanglement entropy $S \sim \sqrt{L}$. Finally, the TL model shows how additional algebraic constraints, encoded in the Jones relation, qualitatively change the structure of $\Kq$ by further decomposing it into many irreducible components. These increasing structural constraints lead to a natural distinction between \textit{weak} and \textit{strong} quantum fragmentation. After removing the EFS, the weak case retains at least one irreducible block occupying a finite fraction of $\Kq$, whereas in the strong case every irreducible block occupies a vanishing fraction. In the models studied here, the weak-strong distinction is clearly reflected in the gap-ratio statistics.


\textit{Mechanism. ---} 
Consider a one-dimensional chain of $L$ sites with local Hilbert space dimension $d$, and let $S = \{0, 1, \ldots, d-1\}$ be the local alphabet. A local word of length $r$ is an element $w = (s_1, \dots,s_r) \in S^{r}$, and we associate to it the basis state $\ket{w}=\ket{s_1} \otimes \cdots\otimes\ket{s_r}$. A semigroup dynamics is specified by a presentation $\tilde{G} = \langle S \,|\, R \rangle$, where $R$ is a set of relations between words over $S$. We assume that $R$ only relates words of the same length $\ell$. This semigroup corresponds to a family of Hamiltonians, with local terms
\begin{equation}\label{eq:Hgeneral}
h_i = \sum_\alpha\sum_{\substack{w', w \in R_\alpha }} g^{w',w} \ket{w'}\bra{w}_{[i,i+\ell-1]}\,,
\end{equation}
where $R_\alpha$ denotes the $\alpha$-th equivalence class in $R$, $g^{w,w'}$ are coupling matrix between states in $R_\alpha$. The total Hamiltonian is $H=\sum_{i=1}^{L-\ell+1} J_i h_i$ with $J_i$ a random variable capturing the local strength of the interaction.
The classical Krylov sectors are the connected components of the graph whose vertices are computational basis states and whose edges connect states related by a single application of a relation in $R$. A basis state forming a dimension-1 sector is a \emph{product frozen state}; collectively these states span the product-frozen subspace $\Fprod$. Classical fragmentation is strong if the dimension of the largest Krylov subspace compared with the full Hilbert space dimension goes to zero in the thermodynamic limit, i.e. ${\rm dim} \mathcal{K}_{{\rm  max},L} / {\rm dim}\mathcal{H} \rightarrow 0 , L\rightarrow \infty$. If ${\rm dim} \mathcal{K}_{{\rm  max},L}/{\rm dim}\mathcal{H} \rightarrow \mathcal{O}(1), L\rightarrow \infty$, the fragmentation is weak.

Quantum fragmentation can arise when the coupling matrix $g^{w,w'}$ within an equivalence class $R_\alpha$ is rank-deficient. If this $|R_\alpha|\times|R_\alpha|$ matrix has full rank, the local term explores all directions in $R_\alpha$ and no quantum fragmentation occurs. If instead it has rank $r < |R_\alpha|$, each local term $h_i$ acquires $(|R_\alpha| - r)$ null directions. For a mobile classical Krylov sector $\Kcl^{(\lambda)}$, the intersection of these local kernels defines the entangled frozen subspace
\begin{equation}
\KEFS^{(\lambda)} = \bigcap_i \ker\!\big(h_i\big|_{\Kcl^{(\lambda)}}\big)\,.
\end{equation}
For generic couplings $J_i$, there are no accidental cancellations between different local terms, so $\KEFS^{(\lambda)} = \ker\left(H|_{\Kcl^{(\lambda)}}\right)$. States in $\KEFS^{(\lambda)}$ are necessarily entangled because there is always one $h_i$ to evolve the product state in a mobile sector. So no product state can lie in $\bigcap_i \ker(h_i)$. We call the states in this set entangled frozen states (EFS).

We say the model exhibits \emph{quantum fragmentation} if $\KEFS^{(\lambda)}$ is non-empty for at least one mobile sector $\lambda$. Note that rank deficiency is necessary but not \emph{a priori} sufficient: each $h_i$ individually has a kernel, but the intersection over all windows could still be empty. In all models studied in this work, however, rank deficiency does produce non-empty $\KEFS^{(\lambda)}$. The orthogonal complement
\begin{equation}\label{eq:decomp}
\Kcl^{(\lambda)} = \Kq^{(\lambda)} \oplus \KEFS^{(\lambda)}
\end{equation}
defines the mobile quantum Krylov subspace. Together with the product-frozen subspace $\Fprod$, the entangled-frozen subspaces assemble into the total frozen space
\begin{equation}
\Ftot = \Fprod \oplus \Fent\,, \qquad \Fent = \bigoplus_{\lambda \in {\rm mobile}} \KEFS^{(\lambda)} ~.
\end{equation}
Hereafter we suppress the sector label $(\lambda)$ when no confusion arises. We call $\Kq$ irreducible if it contains no proper subspace that is invariant under all $\{h_i\}$; when $\Kq$ is reducible, it further decomposes into irreducible quantum Krylov subspaces. As we shall see, whether $\Kq$ is reducible or not, and how many irreducible sub-sectors it contains, sharply distinguishes different types of quantum fragmentation. We emphasize that no symmetry is needed for quantum fragmentation to occur. Symmetry, when present, organizes sectors into degenerate orbits and can enable charge decomposition within $\Kq$, but it does not create the fragmentation itself.


Having established the general framework, we now demonstrate the mechanism in four models of increasing algebraic structure. In this Letter, we focus on the simplest case, where the local term is a rank-1 projector onto an entangled state within the equivalence class, $h_i = J_i \ket{\psi}\bra{\psi}_i$. We further show that other known models exhibiting quantum fragmentation, including the quantum breakdown model and the quantum East model, also fit within this framework. The details are provided in the Supplemental Material~\cite{SM}. We also note an independent and complementary work, Ref.~\cite{han2026quantumfragmentation}, which develops a Rokhsar-Kivelson-type protocol for constructing quantum fragmented models.


\emph{Asymmetric triplet-flip projector. ---} We begin with the simplest case, proving that symmetry is unnecessary. Consider the semigroup $\tilde{G}_1 = \langle 0,1 \,|\, 000 = 111\rangle$ on a qubit chain. The equivalence class has $|R_\alpha| = 2$. We choose the rank-1 projector onto an asymmetric state:
\begin{equation}\label{eq:asym}
H = \sum_i J_i \ket{\psi}\bra{\psi}_{i,i+1,i+2}\,, \quad \ket{\psi} = a\ket{000} + b\ket{111}\,,
\end{equation}
with $a \neq b$ (for concreteness, $a = 1/\sqrt{5}$, $b = 2/\sqrt{5}$). This explicitly breaks the $\mathbb{Z}_2$ spin-flip symmetry: $[\hat{X}, H] \neq 0$ where $\hat{X} = \prod_j X_j$. The local null direction is $|\psi^-\rangle=(b\ket{000} -a\ket{111})/\sqrt{a^2+b^2}$, which is entangled for any $a,b \neq 0$. 
 
As a classical model, there are $N_{\rm frozen}^{\rm cl}(L) = 2 F_{L+1}$ frozen product states, where $F_n$ is the $n$-th Fibonacci number. These are precisely the strings without three consecutive identical qubits, which we denote $|{\rm N3C}\rangle$. The remaining $N_{\rm mobile}^{\rm cl}(L)=F_{L}-1$ non-frozen classical Krylov sectors have dimension greater than one. We call them \emph{mobile} sectors: within each, the semigroup rewriting rule $000 \leftrightarrow 111$ acts like a particle hopping through a frozen background, reshuffling the product-state basis while the N3C portion of the string remains inert. We illustrate the motion under the semigroup equivalence relation in Fig.~\ref{fig:hopping}. 
\begin{figure}[t]
\centering
\begin{tikzpicture}[
  site/.style={minimum size=0.55cm, draw, rounded corners=1.5pt, inner sep=0pt, font=\small\bfseries},
  frozen/.style={site, fill=gray!12, draw=gray!50},
  mobile/.style={site, fill=red!12, draw=red!40},
  mobileb/.style={site, fill=blue!12, draw=blue!40},
  lbl/.style={font=\footnotesize, text=gray!70!black},
  arr/.style={->, >=stealth, thick, gray!60}
]
\node[lbl, anchor=east] at (-0.4, 0) {(a)};
\foreach \x/\v in {0/0, 1/0, 2/0} {
  \node[mobile] (a\x) at (\x*0.75, 0) {\v};
}
\node[frozen] (a3) at (3*0.75, 0) {1};
\draw[decorate, decoration={brace, mirror, amplitude=4pt}, red!50]
  ([yshift=-2pt]a0.south west) to ([yshift=-2pt]a2.south east);
 
\draw[arr] (1.1, -0.55) to (1.1, -0.95) ;
 
\node[lbl, anchor=east] at (-0.4, -1.4) {(b)};
\foreach \x/\v in {0/1, 1/1, 2/1, 3/1} {
  \node[mobileb] (b\x) at (\x*0.75, -1.4) {\v};
}
 
\draw[arr] (1.1, -1.95) to (1.1, -2.35);
 
\node[lbl, anchor=east] at (-0.4, -2.8) {(c)};
\node[frozen] (c0) at (0*0.75, -2.8) {1};
\foreach \x/\v in {1/0, 2/0, 3/0} {
  \node[mobile] (c\x) at (\x*0.75, -2.8) {\v};
}
\draw[decorate, decoration={brace, mirror, amplitude=4pt}, red!50]
  ([yshift=-2pt]c1.south west) to ([yshift=-2pt]c3.south east);
 
\node[lbl, anchor=west] at (2.9, 0) {triplet at sites 1 to 3};
\node[lbl, anchor=west] at (2.9, -1.4) {both windows active};
\node[lbl, anchor=west] at (2.9, -2.8) {triplet at sites 2 to 4};

\end{tikzpicture}
\caption{Mobile triplet hopping at $L=4$ in the semigroup $000=111$. (a)~The $000$ triplet occupies sites 1 to 3. Applying the rewriting rule $000\to 111$ yields state~(b), where both three-site windows contain $111$. A second application $111\to 000$ on sites 2 to 4 gives state~(c). The net effect is that the $000$ triplet has hopped one site to the right. These three product states form a single classical Krylov sector.}
\label{fig:hopping}
\end{figure}
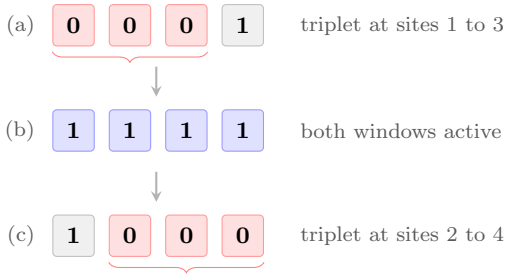
Each mobile classical Krylov subspace can be labeled by the following integrals of motion (IOMs) \cite{hartExactMazurBounds2024}: 
\begin{equation}
    \hat{N}^{\alpha_1 \alpha_2 \cdots \alpha_k}=\sum_{j_1<j_2<\cdots<j_k} \exp \left(\frac{2 \pi i}{p} \sum_{n=1}^k j_n\right) \prod_{n=1}^k \hat{n}_{j_n}^{\alpha_n},
\end{equation}
where the local particle number $\hat{n}_{j}^\alpha=|\alpha\rangle \langle \alpha|_j$ counts the local particle number. The root state of each classical Krylov subspace can be chosen as $|000\rangle^{\otimes k} \otimes |{\rm N3C}\rangle$, where $k$ counts the number of mobile triplets.
 
As a quantum fragmentation model in Eq.~\eqref{eq:asym}, there are additional frozen states that are intrinsically entangled. They can be constructed explicitly from $|\psi^-\rangle$ using the induction operation $\mathrm{Ind}$. The condition $H|\phi\rangle = 0$ requires that the wavefunction coefficients satisfy $c_{\eta,111,\sigma} = -\gamma\, c_{\eta,000,\sigma}$ for every three-site window, where $\gamma = a/b$. Given a seed state, $\mathrm{Ind}$ propagates this constraint across all overlapping windows. For example, at $L=4$, we show one entangled state in Fig.~\ref{fig:indmap}.
\begin{figure}[!h]
\centering
\begin{tikzpicture}[
    site/.style={minimum size=0.55cm, draw, rounded corners=1.5pt, inner sep=0pt, font=\small\bfseries},
    frozen/.style={site, fill=gray!12, draw=gray!50},
    mobile/.style={site, fill=red!12, draw=red!40},
    mobileb/.style={site, fill=blue!12, draw=blue!40},
    lbl/.style={font=\footnotesize, text=gray!70!black},
    math/.style={font=\large},
    bluebox/.style={draw=blue!40, thick, rounded corners=1pt},
    arr/.style={->, >=stealth, thick, gray!60}
]

\def\xstep{0.75}
\def\ystep{-0.8}

\node[math, anchor=south east] (main_eq) at (-1.8*\xstep, \ystep*2.6) {Ind($|\psi^-\rangle \otimes |1\rangle$) =};

\node[math, anchor=east] (minus_ket1) at (-0.3, \ystep*1) {$+ \, |$};
\node[math, anchor=west] (rket1) at (3.4*\xstep, \ystep*1) {$\rangle$};

\foreach \x/\v in {0/0, 1/0, 2/0, 3/1} {
    \node[site] (ket1_\x) at (\x*\xstep, \ystep*1) {\v};
}

\node[math, anchor=east] (plus_ket2) at (-0.3, \ystep*2.2) {$-\gamma \, |$};
\node[math, anchor=west] (rket2) at (3.4*\xstep, \ystep*2.2) {$\rangle$};

\foreach \x/\v in {0/1, 1/1, 2/1, 3/1} {
    \node[site] (ket2_\x) at (\x*\xstep, \ystep*2.2) {\v};
}

\node[math, anchor=east] (minus_ket3) at (-0.3, \ystep*3.4) {$+ \, |$};
\node[math, anchor=west] (rket3) at (3.4*\xstep, \ystep*3.4) {$\rangle$};

\foreach \x/\v in {0/1, 1/0, 2/0, 3/0} {
    \node[site] (ket3_\x) at (\x*\xstep, \ystep*3.4) {\v};
}


\def\boxpad{2.5pt}
\node[bluebox, fit=(ket1_0) (ket1_2) (ket2_0), inner sep=3pt] (box1) {};

\node[blue!60, font=\large, anchor=south] (label1) at (1.1*\xstep, \ystep*1 + 0.4) {$|\psi^-\rangle$};

\node[bluebox, fit=(ket2_1) (ket2_2) (ket3_1) (ket3_3), inner sep=3pt] (box2) {};

\node[blue!60, font=\large, anchor=north] (label2) at (1.85*\xstep, \ystep*3.4 - 0.4) {$|\psi^-\rangle$};


\end{tikzpicture}
\caption{Illustration of an EFS constructed from the Ind map. First, construct the classical Krylov subspace spanned by the state $|000\rangle \otimes |1\rangle$; then, insert coefficients in front of the computational basis states such that on any three consecutive sites, the state is either locally frozen or looks like $|\psi^-\rangle$.}
\label{fig:indmap}
\end{figure}

Each $000\to 111$ flip introduces a factor $-\gamma$, while each $111\to 000$ flip introduces $-1/\gamma$, so a round trip $000\to 111\to 000$ returns with coefficient $+1$. In this model, the product-frozen and entangled-frozen subspaces at general $L$ are
\begin{equation}
\begin{aligned}
\Fprod & = \mathrm{span}\Big\{\ket{{\rm N3C}}\Big\} ~, \\
\Fent & = \mathrm{span}\Big\{\mathrm{Ind}\big(|{\psi}^-\rangle^{\otimes k} \otimes |{\rm N3C}\rangle\big) \Big\}\,,
\end{aligned}
\end{equation}
with $k=1,\ldots,\lfloor L/3\rfloor$ and the total frozen space is $\Ftot = \Fprod \oplus \Fent$. Their dimensions are
\begin{equation}
d_L^{\rm prod} = 2F_{L+1}\,, \qquad d_L^{\rm ent} = F_L - 1\,, \qquad d_L^{\rm tot} = F_{L+3}-1 ~,
\label{eq:dim-frozen}
\end{equation}

The connection between the classical and quantum models can be stated precisely. In the classical model, the Krylov dimension of the sector with $k$ mobile triplets is $D_k^{\rm cl}(L) = \binom{L}{k} - \sum_{j=0}^{k-1}\binom{L}{j}$. We give a proof of this formula in \cite{SM}. In the quantum version of this model in Eq.~\eqref{eq:asym}, each classical mobile sector acquires exactly one EFS, giving the quantum Krylov dimension
\begin{equation}\label{eq:Qk}
D^q_{k}(L) = D_k^{\rm cl}(L) - 1 = \binom{L}{k} - \sum_{j=0}^{k-1}\binom{L}{j} - 1\,.
\end{equation}
For each mobile sector ($k\ge1$), the ``$-1$'' is precisely the EFS $\mathrm{Ind}(|\psi^-\rangle^{\otimes k}\otimes|{\rm N3C}\rangle)$. For this sector,
\begin{equation}
\KEFS^k = \mathrm{span}\Big\{\mathrm{Ind}(|\psi^-\rangle^{\otimes k}\otimes|{\rm N3C}\rangle)\Big\}\,,
\end{equation}
so the decomposition takes the form $\Kcl = \Kq \oplus \KEFS$.

\begin{table*}[t]
\centering
\caption{Krylov subspaces for the asymmetric projector model $H = \sum_i J_i(a\ket{000}+b\ket{111})(a\bra{000}+b\bra{111})_i$ with $a \neq b$. $F_n$: $n$-th Fibonacci number. The first two rows distinguish product-frozen sectors from entangled-frozen sectors. The root state labels the mobile quantum Krylov subspace; $\ket{\rm N3C}$ denotes a frozen string with no three consecutive identical qubits. For $1 \le k < L/3$, the number of non-degenerate sectors is $2F_{L+1-3k}$; the all-mobile sector $k=L/3$ (when $3\mid L$) is unique. Without symmetry, there is no Krylov degeneracy (Deg $= 1$) and no charge decomposition.}
\label{tab:krylov_asym}
\begin{tabular}{lccc}
\toprule
Root state & $D^q$ & Deg. & $\#$ of sectors \\
\midrule
$\ket{{\rm Frozen}_{\rm prod}}$ & $1$ & $1$ & $2F_{L+1}$ \\[4pt]
$\ket{{\rm Frozen}_{\rm ent}}$ & $1$ & $1$ & $F_L-1$ \\[4pt]
$\ket{\psi}\!\otimes\!\ket{\rm N3C}$ & $L-2$ & $1$ & $2F_{L-2}$ \\[4pt]
$\ket{\psi}^{\otimes 2}\!\otimes\!\ket{\rm N3C}$ & $\dfrac{L^2-3L-4}{2}$ & $1$ & $2F_{L-5}$ \\[6pt]
$\ket{\psi}^{\otimes 3}\!\otimes\!\ket{\rm N3C}$ & $\dfrac{L^3-6L^2-L-12}{6}$ & $1$ & $2F_{L-8}$ \\[6pt]
$\ket{\psi}^{\otimes k}\!\otimes\!\ket{\rm N3C}$, $1\le k<L/3$ & $D_k^q(L)$ & $1$ & $2F_{L+1-3k}$ \\[4pt]
($3\mid L$) $\ket{\psi}^{\otimes L/3}$ & $D_{L/3}^q(L)$ & $1$ & $1$ \\[4pt]
\bottomrule
\end{tabular}
\end{table*}


\emph{GHZ projector. ---} Setting $a = b = 1/\sqrt{2}$ in Eq.~\eqref{eq:asym} restores the $\mathbb{Z}_2$ spin-flip symmetry $[\hat{X}, H] = 0$, giving the GHZ projector:
\begin{equation}\label{eq:GHZ_proj}
H_{\rm GHZ} = \sum_i J_i \ket{\rm GHZ}\bra{\rm GHZ}_{i,i+1,i+2}\,.
\end{equation}
For each mobile classical sector, the quantum Krylov dimension $D^q$ before resolving the $\mathbb{Z}_2$ symmetry is the same as in the asymmetric model. The difference is that the $\mathbb{Z}_2$ symmetry enriches the sector structure in two ways. For a generic mobile sector with $k<L/3$, labeled by a nonempty N3C remainder string $w$ of length $L-3k$, the sector rooted at $\ket{\rm GHZ}^{\otimes k}\otimes\ket{w}$ is mapped by $\hat{X}$ to the distinct sector rooted at $\ket{\rm GHZ}^{\otimes k}\otimes\ket{\bar w}$, where $\bar w$ is the bit-complement of $w$. These two sectors are distinct but isospectral, and therefore form a degenerate pair. For the all-mobile sector when $3|L$, i.e.\ $k = L/3$ with empty frozen string, $\hat{X}$ maps the sector to itself and therefore acts as a symmetry within $\Kq$. In that case, the quantum Krylov subspace of dimension $D_{L/3}^{q}$ further splits into two irreducible charge sectors. The dimensions of these sectors can be derived as follows. The classical all-mobile sector has dimension $D_{L/3}^{\rm cl} = D_{L/3}^{q} + 1$, which is always even since every binary string differs from its bit-complement. Each $\hat{X}$ eigenspace of the classical sector therefore has dimension $(D_{L/3}^{q}+1)/2$. The unique EFS in this sector, $\mathrm{Ind}(\ket{\rm GHZ^-}^{\otimes L/3})$, has definite $\hat{X}$-parity $(-1)^{L/3}$. Removing it from the corresponding eigenspace gives
\begin{equation}
\dim \kry_{\hat{X}=\pm 1} = \frac{D_{L/3}^q \mp (-1)^{L/3}}{2}\,.
\end{equation}
The two charge sectors differ in dimension by exactly one because the EFS carries a definite parity.
 
This example illustrates our general principle: symmetry does not create quantum fragmentation, since it is already present in the asymmetric model, but rather organizes the Krylov subspaces. When the symmetry maps a sector to a different one, it creates degeneracies; when it maps a sector to itself, it enables further decomposition within $\Kq$.

The EFS in the asymmetric and GHZ triplet-flip models also have a unique entanglement structure. Within each classical Krylov subspace, the EFS is the unique frustration-free zero mode annihilated by every local term of the Hamiltonian. Despite being dynamically frozen, it is highly entangled: the half-chain entanglement entropy in the largest mobile sector scales as $S_{\rm EFS}(L/2) \sim \sqrt{L}$ \cite{SM}. This sub-volume scaling is reminiscent of the frustration-free Motzkin spin chain~\cite{motzkin} and is sharply distinct from area-law entangled quantum many-body scars.

\emph{Cyclic qutrit projector. ---} We now generalize to the cyclic qutrit projector model with $\mathbb{Z}_3$ symmetry. The semigroup $\tilde{G}_3 = \langle 0,1,2 \,|\, 012 = 120 = 201,\; 102 = 021 = 210\rangle$ with even and odd cyclic projectors
\begin{equation}\label{eq:cyclic}
H_3 = \sum_i J_i(\alpha\ket{\psi_+}\bra{\psi_+}_i + \beta\ket{\psi_-}\bra{\psi_-}_i)\,, \quad \alpha \neq \beta ~,
\end{equation}
where $\ket{\psi_\pm}$ are the even/odd cyclic singlets, possesses $\mathbb{Z}_3$ digit-shift symmetry generated by
$X = \left(\begin{smallmatrix} 0&0&1\\1&0&0\\0&1&0\end{smallmatrix}\right)$, $X\ket{q} = \ket{q{+}1\bmod 3}$.

Each equivalence class has $|R_\alpha| = 3$, and the rank-1 singlet projector creates a 2-dimensional local null space per window. The EFS satisfy $c_{\ldots012\ldots}+c_{\ldots120\ldots}+c_{\ldots201\ldots} = 0$, and similarly for the odd class at every window. The frozen state count satisfies $d_{L+1}^{\rm frozen} = 2d_L^{\rm frozen}+2d_{L-1}^{\rm frozen}+1$.

Unlike the GHZ model, where $\dim \KEFS = 1$ in each mobile sector, the cyclic model has $\dim \KEFS$ growing with system size. The $\mathbb{Z}_3$ symmetry organizes generic mobile sectors into triplets related by the digit shift $\hat{X}^{\otimes L}$. When $3|L$, the all-mobile sectors are invariant under this symmetry, so $\hat{X}^{\otimes L}$ acts within $\Kq$ and resolves it into three charge sectors. This shows that the growth of $\KEFS$ is logically distinct from the further decomposition of $\Kq$. When $J_i,\alpha,\beta$ are real, the system additionally possesses complex conjugation $K$ as an anti-unitary symmetry, which enforces $\Kq^{(1)}$ to be isospectral to $\Kq^{(2)}$. Together with the $\mathbb{Z}_3$ generator, $K$ generates the magnetic group $\mathbb{Z}_3 \sqcup K \mathbb{Z}_3$, and the Hilbert space decomposition follows Wigner’s corepresentation theory \cite{Wigner1959,RevModPhys.40.359,rumynin2021realrepresentationsc2gradedgroups}. A complete sector table at $L=9$ is given in \cite{SM}.

\emph{Temperley-Lieb model. ---} The TL model contains the same two ingredients as the previous examples, namely a classically fragmented semigroup and rank-deficient local projectors, but it adds a new one: the local projectors satisfy the Jones relation. As a result, the remaining quantum Krylov subspace is not merely organized by symmetry charges. Instead, the projected bond algebra acts within $\Kq$ itself and forces a further decomposition into irreducible TL modules.

The model starts from the semigroup dynamics $\tilde{G}_{\rm TL} = \langle 0,1,2\,|\, 00 = 11 = 22\rangle$ and the rank-1 singlet projector $\ket{\Phi^+} = (\ket{00}+\ket{11}+\ket{22})/\sqrt{3}$ gives the Temperley-Lieb Hamiltonian:
\begin{equation}\label{eq:TL}
H_{\rm TL} = \sum_i J_i \ket{\Phi^+}\bra{\Phi^+}_{i,i+1}\,.
\end{equation}
The projectors $e_i = \ket{\Phi^+}\bra{\Phi^+}_{i,i+1}$ satisfy the Jones relation,
\begin{equation}\label{eq:Jones}
e_i^2 = e_i\,,\quad e_ie_j = e_je_i\;(|i{-}j|\geq 2)\,,\quad e_ie_{i\pm 1}e_i = \tfrac{1}{3}\,e_i\,.
\end{equation}
so the projected bond algebra in each mobile sector is a representation of ${\rm TL}_{L-1}(3)$. At the same time, the rank deficiency produces EFS exactly as in the previous models. If we parametrize the state as $|\Psi\rangle = \sum c_s|s\rangle$ with $s$ ranging over strings in $\{0,1,2\}^L$, the EFS condition requires
\begin{equation}
\begin{aligned}
    c_{\ldots 00 \ldots}+c_{\ldots 11 \ldots}+c_{\ldots 22 \ldots} =0 ~.
\end{aligned}
\end{equation}
For example, in the mobile sector $\Kcl(0000)$ at $L=4$, the entangled-frozen subspace has dimension $\dim \KEFS = 7$; a convenient basis is listed in the End Matter. Thus the same decomposition $\Kcl = \Kq \oplus \KEFS$ holds here as well. The new feature is that $\Kq$ is itself reducible as a ${\rm TL}_{L-1}(3)$ module. We therefore decompose it into standard modules $\Delta_j$, whose dimensions are
\begin{equation}
    {\rm dim} \Delta_j = \binom{L}{(L-j) / 2}-\binom{L}{(L-j) / 2-1} ~.
\end{equation}
When $L$ is even, $j=0,2,4,\ldots$ takes even values; when $L$ is odd, $j=1,3,5,\ldots$ takes odd values. In the $L=4$ example, $\dim \Delta_0=2,\dim \Delta_2 =3,\dim \Delta_4=1$, and the quantum Krylov subspace decomposes as
\begin{equation}
    \Kq = \Delta_0 \otimes  \mathbb{C}^1 \oplus \Delta_2 \otimes \mathbb{C}^{2} ~.
\end{equation}
In Table~\ref{tab:TL}, we show the decomposition of $\Kq$ in the all-mobile sectors for $L=5,6,7,8,9$. The number of irreducible TL blocks grows rapidly with system size, so the Jones relation fragments $\Kq$ itself rather than merely organizing sectors by symmetry.
\begin{table}[t]
\caption{Decomposition of the all-mobile quantum Krylov subspace into irreducible representations of the TL algebra. Here $N_{\rm irr}$ counts irreducible TL blocks, i.e.\ the total number of standard-module copies in the decomposition.}
\label{tab:TL}
\centering
\resizebox{\columnwidth}{!}{%
\begin{tabular}{cccc}
\toprule
$L$ & $D_q$ & $N_{\rm irr}$ & Krylov decomposition  \\
\midrule
5 & 17 & 4 & $\Delta_1 \otimes \mathbb{C}^{1} \oplus \Delta_3 \otimes \mathbb{C}^3$  \\
6 & 58 & 10 & $\Delta_0 \otimes \mathbb{C}^1 \oplus \Delta_{2} \otimes \mathbb{C}^2 \oplus \Delta_{4}\otimes \mathbb{C}^7$  \\
7 & 128 & 16 & $\Delta_{1} \otimes \mathbb{C}^1 \oplus \Delta_{3} \otimes \mathbb{C}^{3} \oplus \Delta_{5}\otimes \mathbb{C}^{12}$ \\
8 & 413 & 39 & $\Delta_{0}\otimes \mathbb{C}^{1} \oplus \Delta_{2}\otimes \mathbb{C}^{2} \oplus \Delta_{4} \otimes \mathbb{C}^{7} \oplus \Delta_{6} \otimes \mathbb{C}^{29}$ \\
9 & 934 & 69 & $\Delta_{1}\otimes \mathbb{C}^1 \oplus \Delta_{3} \otimes \mathbb{C}^{3} \oplus \Delta_{5} \otimes \mathbb{C}^{12} \oplus \Delta_{7} \otimes \mathbb{C}^{53}$ \\
\bottomrule
\end{tabular}
}
\end{table}

\textit{Weak \vs\ Strong Quantum Fragmentation. ---}
All four models discussed above share the same structure. The classical frozen sectors are isolated product states and together span $\Fprod$, while each mobile classical Krylov subspace splits as in Eq.~\eqref{eq:decomp} into a mobile quantum Krylov subspace and an entangled-frozen subspace. The difference between models lies in whether the remaining mobile space $\Kq$ is itself irreducible. In the asymmetric projector model, $\Kq$ cannot be further decomposed. In the GHZ projector model, the all-mobile $\Kq$ sector further decomposes according to the $\mathbb{Z}_2$ spin-flip symmetry. Similarly, in the cyclic model, the all-mobile $\Kq$ decomposes with respect to the $\mathbb{Z}_3$ symmetry. In the TL model, $\Kq$ decomposes into TL standard modules. This motivates the following definition of weak and strong quantum fragmentation.

\begin{figure*}[t]
    \centering
    \includegraphics[width=\linewidth]{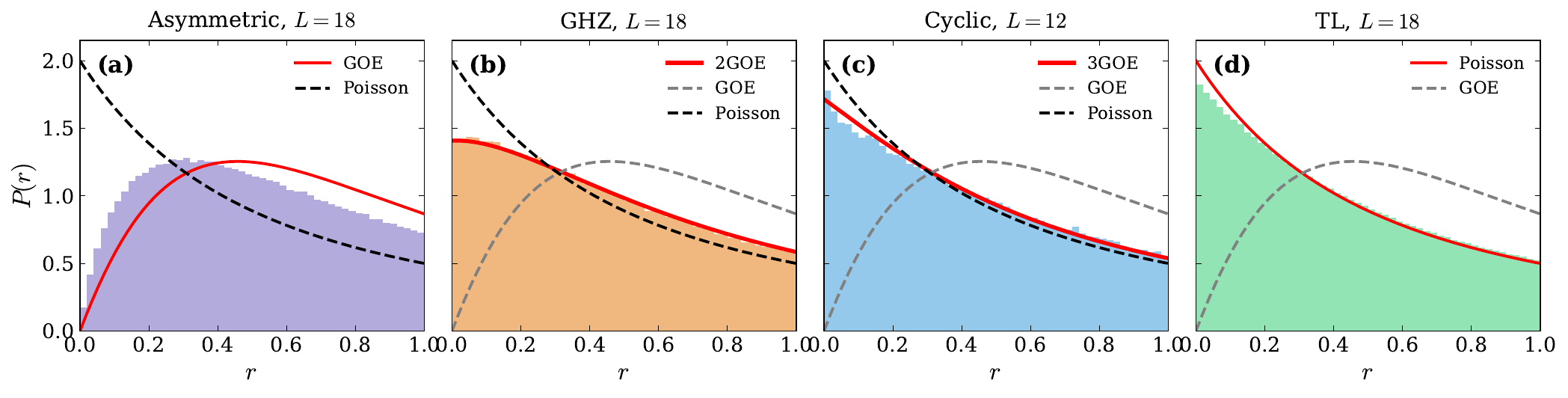}
    \caption{Gap ratio distribution $P(r)$ for the distinct eigenvalues in $\Kq$ of the largest mobile sector, compared to GOE and Poisson. (a)~Asymmetric projector ($L=18$): finite-size data drifting toward GOE. (b)~GHZ projector ($L=18$): well captured by $2$GOE, reflecting unresolved $\mathbb{Z}_2$ sectors. (c)~Cyclic qutrit ($L=12$): well captured by $3$GOE, reflecting unresolved $\mathbb{Z}_3$ sectors. (d)~TL model ($L=18$): many TL blocks, approaching Poisson distribution.}
    \label{fig:level_stats}
\end{figure*}

For a given mobile classical Krylov sector, we first remove the EFS. We then decompose the remaining reducible quantum Krylov subspace into irreducible invariant subspaces of the projected bond algebra generated by the local terms,
\begin{equation}
    \Kq = \bigoplus_{\mu=1}^{N_{\rm irr}(L)} \kry_{q,\mu} ~,
\end{equation}
where each $\kry_{q,\mu}$ admits no further decomposition under all $\{h_i\}$. Let
\begin{equation}
    D_q = \dim \Kq\,, \qquad D_{\rm max} = \max_\mu \dim \kry_{q,\mu} ~.
\end{equation}
We call the quantum fragmentation \emph{weak} if $D_{\rm max}/D_q \sim \mathcal{O}(1)$ in the thermodynamically large mobile sectors, namely if at least one irreducible quantum block continues to occupy a finite fraction of the mobile space $\mathcal{K}_q$. In all models studied here, this is equivalently characterized by $N_{\rm irr}(L)=\mathcal{O}(1)$. We call the quantum fragmentation \emph{strong} if
\begin{equation}
    \frac{D_{\rm max}}{D_q} \xrightarrow[L\to\infty]{} 0 ~,
\end{equation}
so that no irreducible quantum block retains a finite fraction of $\Kq$. In our examples, this occurs because $N_{\rm irr}(L)$ grows with system size.

This definition is the direct quantum analog of the usual weak/strong criterion for classical fragmentation based on the largest Krylov subspace fraction. With this definition, the asymmetric projector model is weakly quantum fragmented because $\Kq$ is irreducible. The GHZ and cyclic projector models are also weakly quantum fragmented: after resolving the $\mathbb{Z}_2$ or $\mathbb{Z}_3$ charges in the symmetry-stable all-mobile sectors, only $\mathcal{O}(1)$ irreducible blocks remain. By contrast, the TL model is strongly quantum fragmented. There the Jones relation decomposes $\Kq$ into an increasing number of TL standard modules with growing multiplicities, so $D_{\rm max}/D_q$ decreases with $L$.

The weak/strong distinction leaves a clear imprint on the gap-ratio statistics~\cite{Oganesyan2007,Atas2013,Bohigas:1983er,Berry:1977jf}
\begin{equation}
r_n = \frac{\min(\delta_n, \delta_{n+1})}{\max(\delta_n, \delta_{n+1})}, \quad \delta_n = E_{n+1} - E_n~,
\end{equation}
computed within $\Kq$ of the largest mobile sector with random couplings $J_i$ averaged over disorder realizations. We show the results in Fig.~\ref{fig:level_stats}. In weakly fragmented models, $\Kq$ decomposes into $\mathcal{O}(1)$ irreducible blocks, each of which individually thermalizes with GOE level repulsion. When these blocks are not resolved by symmetry quantum numbers, the observed spectrum is a superposition of $m$ independent GOE distributions, described by the $m$GOE gap-ratio distribution $P_m(r)$~\cite{Giraud2022}. The asymmetric model shown in panel~(a) has $m=1$ irreducible block and drifts toward GOE with increasing $L$. The GHZ model has $m=2$ blocks, corresponding to $\mathbb{Z}_2$ charge, and the cyclic model exhibits $m=3$ from $\mathbb{Z}_3$ charge. Both are well captured by the corresponding $m$GOE predictions as shown in panels~(b,c). In the strongly fragmented TL model shown in panel~(d), $\Kq$ splits into an exponentially growing number of irreducible TL standard modules as $L$ increases. Since $N_{\rm irr}(L) \to \infty$, the $m$GOE distribution with $m = N_{\rm irr}$ approaches Poisson in this limit, and the gap-ratio distribution is expected to converge to $P_{\rm Poisson}(r)$ as $L \to \infty$.

\textit{Discussion. ---}
We have identified entangled frozen states as the universal mechanism underlying quantum Hilbert space fragmentation. The mechanism contains two essential ingredients, classically fragmented models and rank-deficient local terms. There is no necessary requirement for symmetries or specific bond algebras. However, symmetries are additional inputs that can generate degenerate Krylov subspaces and decompose reducible Krylov subspaces into symmetry charge sectors. In the asymmetric and GHZ projector models, the half-chain entanglement entropy of the EFS in the largest mobile sector scales as $\sqrt{L}$. We introduce the notation weak/strong quantum fragmentation to capture the reducibility of the quantum Krylov subspace under these additional structures. In the four models discussed in this paper, we can distinguish strong and weak quantum fragmentation in the gap-ratio statistics.

The EFSs are isolated from classical mobile Krylov subspaces, yet they span a subspace of dimension growing with $L$. This dynamical protection, arising from entanglement structure rather than symmetry, suggests connections to quantum error correction. The GHZ projector model admits efficient Trotterization for quantum simulation on near-term hardware, providing a concrete experimental platform for probing quantum fragmentation and the weak/strong transition.

\textit{Notes added. --- } During the preparation of this work, we are aware of a similar work~\cite{han2026quantumfragmentation}.

\textit{Acknowledgements. --- } We thank Yiqiu Han, Oliver Hart, Yahui Li, Biao Lian, Shuo Liu, Yukai Lu, Yu-Ping Wang, Nicholas O'Dea for helpful discussions. Z.Z. would like to especially thank Matias Zaldarriaga for organizing the AI term at the IAS and for his support in facilitating the use of Claude Code.

\begin{center}

{\large\textbf{END MATTER}}

\vspace{0.5cm}

\textbf{A: Basis of \texorpdfstring{$\KEFS$}{KEFS} in \texorpdfstring{$\Kcl(0000)$}{Kcl(0000)} of the TL model}


\end{center}

For the $L=4$ mobile sector $\Kcl(0000)$ of the TL model, a convenient basis of the seven-dimensional entangled-frozen subspace $\KEFS$ is
\begin{equation*}
\begin{aligned}
    |{\rm EFS}_1\rangle & = |0110\rangle - |0220\rangle ~, \quad |{\rm EFS}_2 \rangle = |1001\rangle - |1221\rangle ~, \\
    |{\rm EFS}_3\rangle & = |2002\rangle - |2112\rangle ~, \\
    |{\rm EFS}_4\rangle & = |0000\rangle-|0011\rangle-|0110\rangle-|1001\rangle  \\
    & \quad -|1100\rangle+|1111\rangle ~, \\
        |{\rm EFS}_5\rangle & = |0000\rangle-|0022\rangle-|0220\rangle-|2002\rangle \\
    & \quad -|2200\rangle+|2222\rangle ~, \\
    \end{aligned}
\end{equation*}
\begin{equation}
\begin{aligned}
    |{\rm EFS}_6 \rangle & = |1111\rangle-|1122\rangle-|1221\rangle-|2112\rangle \\
    & \quad -|2211\rangle+|2222\rangle ~, \\
    |{\rm EFS}_7\rangle & =|0000\rangle-|0011\rangle-|0220\rangle -|2200\rangle+|2211\rangle ~.
\end{aligned}
\end{equation}

\vspace{0.5cm}

\textbf{B: Sub-volume $\sqrt{L}$ entanglement scaling}

\vspace{0.5cm}
Fig.~\ref{fig:efs-scaling} demonstrates the $\sqrt{L}$ entanglement scaling for the EFS in the GHZ projector and the asymmetric projector models with weights $\gamma=a/b$. Each mobile classical Krylov sector is labeled by an N3C frozen string $w_{\rm frozen}$ of length $L - 3k$, which represents a group element $c_f \in \mathbb{Z}_3 * \mathbb{Z}_3$ ~\cite{SM}. The all-mobile sector corresponds to $c_f = e$, while the smallest non-empty frozen tails for $L = 3k+1$ and $L = 3k+2$ are $c_f = 0$ and $c_f = 01$ respectively. The bipartite entanglement entropy $S(L/2)$ at the symmetric bisection grows linearly in $\sqrt{L}$ in all three sector families. Across the full range of cut positions $L_A$, the entropy $S(L_A)$ traces the characteristic Brownian-bridge arch $\sim \sqrt{L_A(L-L_A)/L}$ in Fig.~\ref{fig:efs-vs-LA}. The details of the bridge-walk picture is given in \cite{SM}.

\begin{widetext}

\begin{figure}[!h]
\centering
\includegraphics[width=0.95\linewidth]{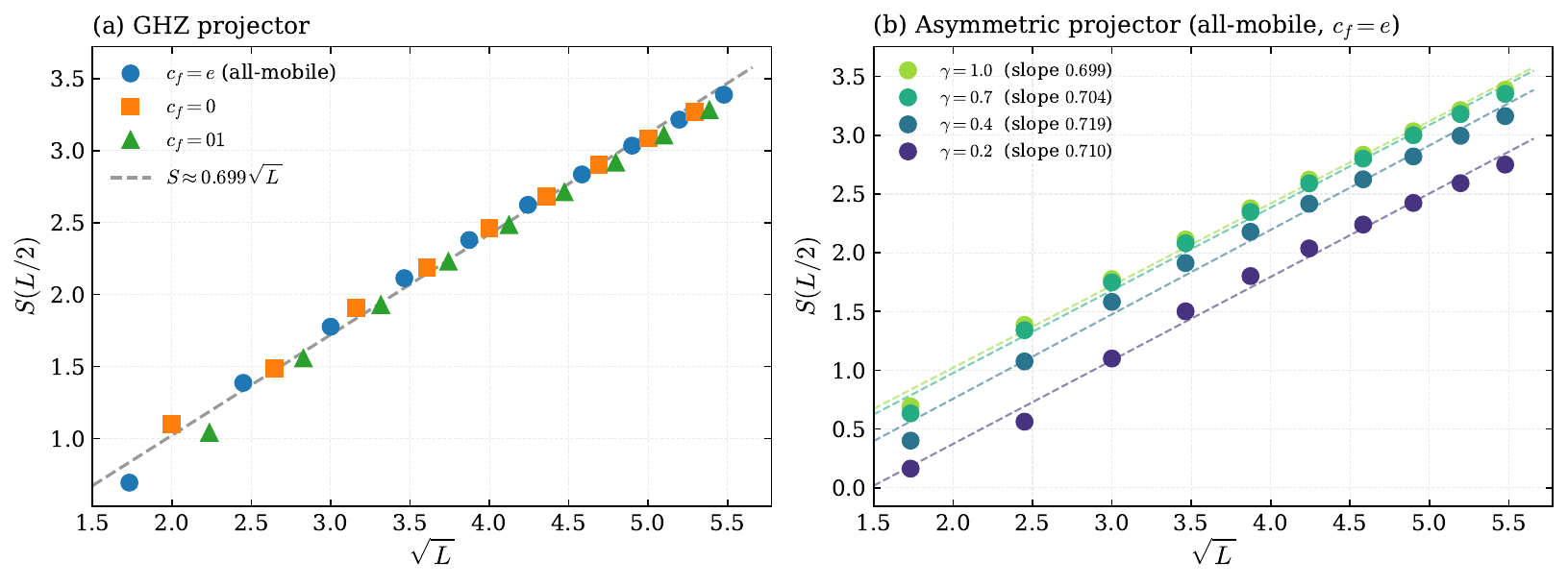}
\caption{Midpoint bipartite entanglement entropy $S(L/2)$ of the EFS versus $\sqrt{L}$. \textbf{(a)} EFS entanglement entropy scaling in the GHZ projector model with three sector families ($c_f \in \{e, 0, 01\}$). They all follow the same linear scaling $S \approx 0.699\sqrt{L}$. \textbf{(b)} Asymmetric projector with $\gamma$ varying from the GHZ point $\gamma = 1$ down to $\gamma = 0.2$, all in the all-mobile classical Krylov sector. The $\sqrt{L}$ scaling persists for every $\gamma$, with $\gamma$-dependent prefactor.}
\label{fig:efs-scaling}
\end{figure}

\begin{figure}[!h]
\centering
\includegraphics[width=0.95\linewidth]{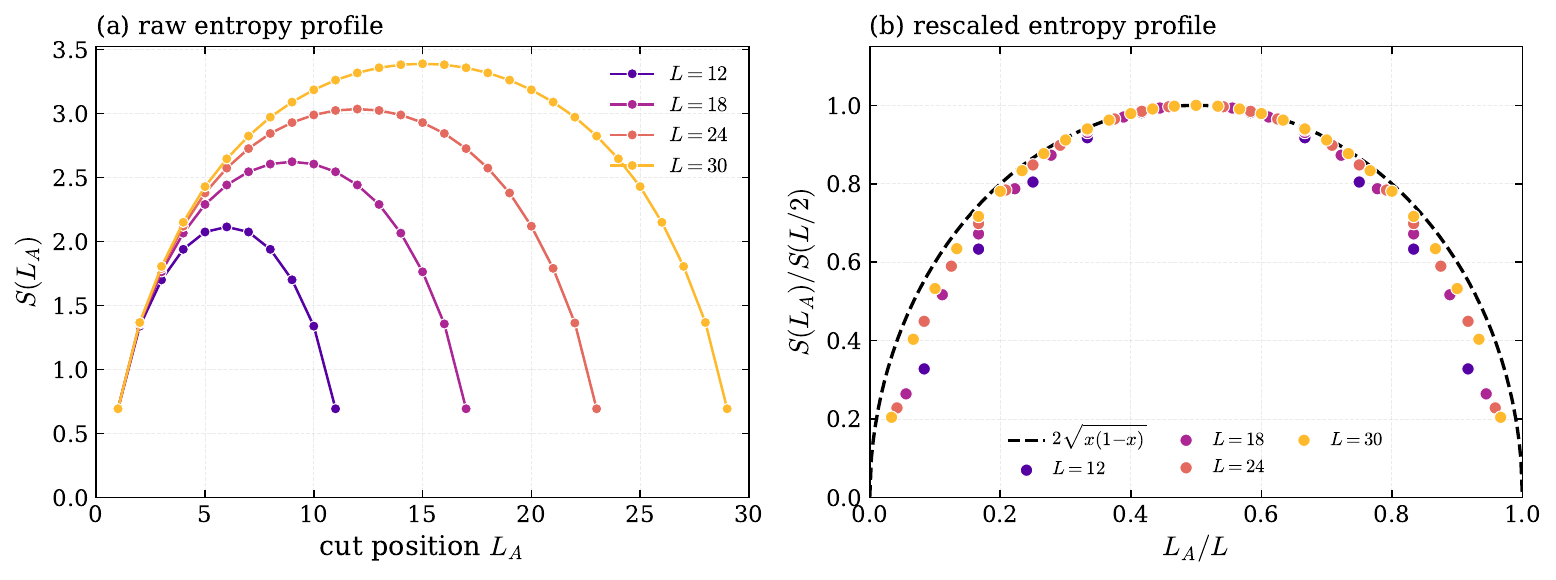}
\caption{EFS entanglement entropy versus cut position for the GHZ all-mobile sector at various $L$. \textbf{(a)} Raw arch profile $S(L_A)$, peaking at $L_A = L/2$ and vanishing at the boundaries. \textbf{(b)} Rescaled profile: rescaling $S(L_A)/S(L/2)$ vs $L_A/L$ collapses all curves onto a universal arch close to the Brownian-bridge envelope $2\sqrt{x(1-x)}$ shown in the dashed line, with small finite-$L$ deviations at the boundaries.}
\label{fig:efs-vs-LA}
\end{figure}
\end{widetext}

\bibliography{mybib.bib}

@article{hartExactMazurBounds2024,
  title = {Exact {{Mazur}} Bounds in the Pair-Flip Model and Beyond},
  author = {Hart, Oliver},
  year = 2024,
  month = jul,
  journal = {SciPost Phys. Core},
  volume = {7},
  number = {3},
  pages = {040},
  issn = {2666-9366},
  doi = {10.21468/SciPostPhysCore.7.3.040},
  urldate = {2025-11-21},
}

@article{Balasubramanian:2023wgn,
    author = "Balasubramanian, Shankar and Gopalakrishnan, Sarang and Khudorozhkov, Alexey and Lake, Ethan",
    title = "{Glassy Word Problems: Ultraslow Relaxation, Hilbert Space Jamming, and Computational Complexity}",
    eprint = "2312.04562",
    archivePrefix = "arXiv",
    primaryClass = "quant-ph",
    doi = "10.1103/PhysRevX.14.021034",
    journal = "Phys. Rev. X",
    volume = "14",
    number = "2",
    pages = "021034",
    year = "2024"
}

@article{Cornwell2008COUNTINGFP,
  title={COUNTING FUNDAMENTAL PATHS IN CERTAIN GARSIDE SEMIGROUPS},
  author={Christopher R. Cornwell and Stephen P. Humphries},
  journal={Journal of Knot Theory and Its Ramifications},
  year={2008},
  volume={17},
  pages={191-211},
  url={https://api.semanticscholar.org/CorpusID:122380191}
}

@article{Giraud2022,
  title = {Probing symmetries of quantum many-body systems through gap ratio statistics},
  author = {Giraud, Olivier and Mac{\'e}, Nicolas and Vernier, {\'E}ric and Alet, Fabien},
  year = {2022},
  journal = {Physical Review X},
  volume = {12},
  number = {1},
  pages = {011006},
  doi = {10.1103/PhysRevX.12.011006},
  eprint = {2008.11173},
  archivePrefix = {arXiv},
  primaryClass = {cond-mat.dis-nn}
}

@article{Khemani:2019vor,
    author = "Khemani, Vedika and Hermele, Michael and Nandkishore, Rahul M.",
    title = "{Localization from Hilbert space shattering: From theory to physical realizations}",
    eprint = "1910.01137",
    archivePrefix = "arXiv",
    primaryClass = "cond-mat.stat-mech",
    doi = "10.1103/PhysRevB.101.174204",
    journal = "Phys. Rev. B",
    volume = "101",
    number = "17",
    pages = "174204",
    year = "2020"
}

@article{sala2020ergodicity,
  title={{E}rgodicity breaking arising from {H}ilbert space fragmentation in dipole-conserving {H}amiltonians},
  author={Sala, Pablo and Rakovszky, Tibor and Verresen, Ruben and Knap, Michael and Pollmann, Frank},
  journal={Physical Review X},
  doi = "10.1103/PhysRevX.10.011047",
  volume={10},
  number={1},
  pages={011047},
  year={2020},
  publisher={APS}
}

@article{Moudgalya:2021xlu,
    author = "Moudgalya, Sanjay and Bernevig, B. Andrei and Regnault, Nicolas",
    title = "{Quantum many-body scars and Hilbert space fragmentation: a review of exact results}",
    eprint = "2109.00548",
    archivePrefix = "arXiv",
    primaryClass = "cond-mat.str-el",
    doi = "10.1088/1361-6633/ac73a0",
    journal = "Rept. Prog. Phys.",
    volume = "85",
    number = "8",
    pages = "086501",
    year = "2022"
}

@inbook{Moudgalya:2019vlp,
    author = "Moudgalya, Sanjay and Prem, Abhinav and Nandkishore, Rahul and Regnault, Nicolas and Bernevig, B. Andrei",
    title = "{Thermalization and Its Absence within Krylov Subspaces of a Constrained Hamiltonian}",
    eprint = "1910.14048",
    archivePrefix = "arXiv",
    primaryClass = "cond-mat.str-el",
    doi = "10.1142/9789811231711_0009",
    year = "2021"
}

@article{Moudgalya:2021ixk,
    author = "Moudgalya, Sanjay and Motrunich, Olexei I.",
    title = "{Hilbert Space Fragmentation and Commutant Algebras}",
    eprint = "2108.10324",
    archivePrefix = "arXiv",
    primaryClass = "cond-mat.stat-mech",
    doi = "10.1103/PhysRevX.12.011050",
    journal = "Phys. Rev. X",
    volume = "12",
    pages = "011050",
    year = "2022"
}

@article{Moudgalya:2022gtj,
    author = "Moudgalya, Sanjay and Motrunich, Olexei I.",
    title = "{From symmetries to commutant algebras in standard Hamiltonians}",
    eprint = "2209.03370",
    archivePrefix = "arXiv",
    primaryClass = "cond-mat.str-el",
    doi = "10.1016/j.aop.2023.169384",
    journal = "Annals Phys.",
    volume = "455",
    pages = "169384",
    year = "2023"
}

@article{Regnault:2022ocy,
    author = "Regnault, Nicolas and Liu, Shuo and Bernevig, B. Andrei",
    title = "{Integer characteristic polynomial factorization and Hilbert space fragmentation}",
    journal = "arXiv e-prints",
    eprint = "2210.08019",
    archivePrefix = "arXiv",
    primaryClass = "cond-mat.stat-mech",
    year = "2022"
}

@article{Iadecola:2025fyg,
    author = "Iadecola, Thomas",
    title = "{Symmetry Fragmentation}",
    eprint = "2510.06333",
    archivePrefix = "arXiv",
    primaryClass = "quant-ph",
    doi = "10.1103/n516-psj5",
    journal = "Phys. Rev. Lett.",
    volume = "136",
    number = "10",
    pages = "100401",
    year = "2026"
}

@article{Lian:2022nqj,
    author = "Lian, Biao",
    title = "{Quantum breakdown model: From many-body localization to chaos with scars}",
    eprint = "2208.10509",
    archivePrefix = "arXiv",
    primaryClass = "cond-mat.str-el",
    doi = "10.1103/PhysRevB.107.115171",
    journal = "Phys. Rev. B",
    volume = "107",
    number = "11",
    pages = "115171",
    year = "2023"
}

@article{Brighi:2022kca,
    author = "Brighi, Pietro and Ljubotina, Marko and Serbyn, Maksym",
    title = "{Hilbert space fragmentation and slow dynamics in particle-conserving quantum East models}",
    eprint = "2210.15607",
    archivePrefix = "arXiv",
    primaryClass = "quant-ph",
    doi = "10.21468/SciPostPhys.15.3.093",
    journal = "SciPost Phys.",
    volume = "15",
    number = "3",
    pages = "093",
    year = "2023"
}

@article{Chen:2024fkk,
    author = "Chen, Bo-Ting and Prem, Abhinav and Regnault, Nicolas and Lian, Biao",
    title = "{Quantum fragmentation in the extended quantum breakdown model}",
    eprint = "2401.16480",
    archivePrefix = "arXiv",
    primaryClass = "cond-mat.str-el",
    doi = "10.1103/PhysRevB.110.165109",
    journal = "Phys. Rev. B",
    volume = "110",
    number = "16",
    pages = "165109",
    year = "2024"
}

@article{Chen:2026aqj,
    author = "Chen, Bo-Ting and Wang, Yu-Ping and Lian, Biao",
    title = "{Bridging Commutant and Polynomial Methods for Hilbert Space Fragmentation}",
    journal = "arXiv e-prints",
    eprint = "2601.00294",
    archivePrefix = "arXiv",
    primaryClass = "cond-mat.stat-mech",
    year = "2026"
}

@article{Oganesyan2007,
    author = "Oganesyan, Vadim and Huse, David A.",
    title = "{Localization of interacting fermions at high temperature}",
    doi = "10.1103/PhysRevB.75.155111",
    journal = "Phys. Rev. B",
    volume = "75",
    number = "15",
    pages = "155111",
    year = "2007"
}

@article{Atas2013,
    author = "Atas, Y. Y. and Bogomolny, E. and Giraud, O. and Roux, G.",
    title = "{Distribution of the ratio of consecutive level spacings in random matrix ensembles}",
    doi = "10.1103/PhysRevLett.110.084101",
    journal = "Phys. Rev. Lett.",
    volume = "110",
    number = "8",
    pages = "084101",
    year = "2013",
    eprint = "1212.5611",
    archivePrefix = "arXiv",
    primaryClass = "cond-mat.stat-mech"
}

@article{huet1980confluent,
    author = "Huet, G{\'e}rard",
    title = "{Confluent Reductions: Abstract Properties and Applications to Term Rewriting Systems}",
    doi = "10.1145/322217.322230",
    journal = "J. ACM",
    volume = "27",
    number = "4",
    pages = "797--821",
    year = "1980"
}

@book{Stanley_Fomin_1999,
    author = "Stanley, Richard P.",
    title = "{Enumerative Combinatorics, Volume 2}",
    publisher = "Cambridge University Press",
    year = "1999",
    doi = "10.1017/CBO9780511609589"
}

@misc{oeisA047098,
    author = "{OEIS Foundation}",
    title = "{Entry A047098 in The On-Line Encyclopedia of Integer Sequences}",
    url = {https://oeis.org/A047098},
    year = "2024"
}

@misc{oeisA213028,
    author = "{OEIS Foundation}",
    title = "{Entry A213028 in The On-Line Encyclopedia of Integer Sequences}",
    url = {https://oeis.org/A213028},
    year = "2024"
}

@article{Rakovszky:2020apb,
    author = "Rakovszky, Tibor and Sala, Pablo and Verresen, Ruben and Knap, Michael and Pollmann, Frank",
    title = "{Statistical localization: From strong fragmentation to strong edge modes}",
    eprint = "1910.06341",
    archivePrefix = "arXiv",
    primaryClass = "cond-mat.str-el",
    doi = "10.1103/PhysRevB.101.125126",
    journal = "Phys. Rev. B",
    volume = "101",
    number = "12",
    pages = "125126",
    year = "2020"
}

@article{Moudgalya:2018fua,
    author = "Moudgalya, Sanjay and Rachel, Stephan and Bernevig, B. Andrei and Regnault, Nicolas",
    title = "{Exact excited states of nonintegrable models}",
    eprint = "1708.05021",
    archivePrefix = "arXiv",
    primaryClass = "cond-mat.str-el",
    doi = "10.1103/PhysRevB.98.235155",
    journal = "Phys. Rev. B",
    volume = "98",
    number = "23",
    pages = "235155",
    year = "2018"
}

@article{Morningstar:2020ror,
    author = "Morningstar, Alan and Khemani, Vedika and Huse, David A.",
    title = "{Kinetically constrained freezing transition in a dipole-conserving system}",
    eprint = "2004.00096",
    archivePrefix = "arXiv",
    primaryClass = "cond-mat.stat-mech",
    doi = "10.1103/PhysRevB.101.214205",
    journal = "Phys. Rev. B",
    volume = "101",
    number = "21",
    pages = "214205",
    year = "2020"
}

@article{Yang:2019mft,
    author = "Yang, Zhi-Cheng and Liu, Fangli and Gorshkov, Alexey V. and Iadecola, Thomas",
    title = "{Hilbert-Space Fragmentation from Strict Confinement}",
    eprint = "1912.04300",
    archivePrefix = "arXiv",
    primaryClass = "cond-mat.str-el",
    doi = "10.1103/PhysRevLett.124.207602",
    journal = "Phys. Rev. Lett.",
    volume = "124",
    number = "20",
    pages = "207602",
    year = "2020"
}

@article{Temperley:1971iq,
    author = "Temperley, H. N. V. and Lieb, E. H.",
    title = "{Relations between the `percolation' and `colouring' problem and other graph-theoretical problems associated with regular planar lattices: some exact results for the `percolation' problem}",
    doi = "10.1098/rspa.1971.0067",
    journal = "Proc. Roy. Soc. Lond. A",
    volume = "322",
    pages = "251--280",
    year = "1971"
}

@article{Jones:1983kv,
    author = "Jones, V. F. R.",
    title = "{Index for subfactors}",
    doi = "10.1007/BF01389127",
    journal = "Invent. Math.",
    volume = "72",
    pages = "1--25",
    year = "1983"
}

@article{Bernien:2017ubn,
    author = "Bernien, Hannes and others",
    title = "{Probing many-body dynamics on a 51-atom quantum simulator}",
    eprint = "1707.04344",
    archivePrefix = "arXiv",
    primaryClass = "quant-ph",
    doi = "10.1038/nature24622",
    journal = "Nature",
    volume = "551",
    pages = "579--584",
    year = "2017"
}

@article{Turner:2017fxc,
    author = "Turner, Christopher J. and Michailidis, Alexios A. and Abanin, Dmitry A. and Serbyn, Maksym and Papic, Zlatko",
    title = "{Weak ergodicity breaking from quantum many-body scars}",
    eprint = "1711.03528",
    archivePrefix = "arXiv",
    primaryClass = "quant-ph",
    doi = "10.1038/s41567-018-0137-5",
    journal = "Nature Phys.",
    volume = "14",
    pages = "745--749",
    year = "2018"
}

@article{Nandkishore:2014kca,
    author = "Nandkishore, Rahul and Huse, David A.",
    title = "{Many body localization and thermalization in quantum statistical mechanics}",
    eprint = "1404.0686",
    archivePrefix = "arXiv",
    primaryClass = "cond-mat.stat-mech",
    doi = "10.1146/annurev-conmatphys-031214-014726",
    journal = "Ann. Rev. Condensed Matter Phys.",
    volume = "6",
    pages = "15--38",
    year = "2015"
}

@article{Bohigas:1983er,
    author = "Bohigas, O. and Giannoni, M. J. and Schmit, C.",
    title = "{Characterization of chaotic quantum spectra and universality of level fluctuation laws}",
    doi = "10.1103/PhysRevLett.52.1",
    journal = "Phys. Rev. Lett.",
    volume = "52",
    pages = "1--4",
    year = "1984"
}

@article{Abanin:2018yrt,
    author = "Abanin, Dmitry A. and Altman, Ehud and Bloch, Immanuel and Serbyn, Maksym",
    title = "{Colloquium: Many-body localization, thermalization, and entanglement}",
    eprint = "1804.11065",
    archivePrefix = "arXiv",
    primaryClass = "cond-mat.dis-nn",
    doi = "10.1103/RevModPhys.91.021001",
    journal = "Rev. Mod. Phys.",
    volume = "91",
    number = "2",
    pages = "021001",
    year = "2019"
}

@article{Moudgalya:2022nll,
    author = "Moudgalya, Sanjay and Motrunich, Olexei I.",
    title = "{Exhaustive Characterization of Quantum Many-Body Scars Using Commutant Algebras}",
    eprint = "2209.03377",
    archivePrefix = "arXiv",
    primaryClass = "cond-mat.str-el",
    doi = "10.1103/PhysRevX.14.041069",
    journal = "Phys. Rev. X",
    volume = "14",
    number = "4",
    pages = "041069",
    year = "2024"
}

@article{Turner:2018gwa,
    author = "Turner, Christopher J. and Michailidis, Alexios A. and Abanin, Dmitry A. and Serbyn, Maksym and Papic, Zlatko",
    title = "{Quantum scarred eigenstates in a Rydberg atom chain: Entanglement, breakdown of thermalization, and stability to perturbations}",
    eprint = "1806.10933",
    archivePrefix = "arXiv",
    primaryClass = "cond-mat.quant-gas",
    doi = "10.1103/PhysRevB.98.155134",
    journal = "Phys. Rev. B",
    volume = "98",
    number = "15",
    pages = "155134",
    year = "2018"
}

@article{Berry:1977jf,
    author = "Berry, M. V. and Tabor, M.",
    title = "{Level clustering in the regular spectrum}",
    doi = "10.1098/rspa.1977.0140",
    journal = "Proc. Roy. Soc. Lond. A",
    volume = "356",
    pages = "375--394",
    year = "1977"
}

@book{Wigner1959,
  author    = {Eugene P. Wigner},
  title     = {Group Theory and Its Application to the Quantum Mechanics of Atomic Spectra},
  publisher = {Academic Press},
  address   = {New York},
  year      = {1959}
}

@article{RevModPhys.40.359,
  title = {Magnetic Groups and Their Corepresentations},
  author = {Bradley, C. J. and Davies, B. L.},
  journal = {Rev. Mod. Phys.},
  volume = {40},
  issue = {2},
  pages = {359--379},
  numpages = {0},
  year = {1968},
  month = {Apr},
  publisher = {American Physical Society},
  doi = {10.1103/RevModPhys.40.359},
  url = {https://link.aps.org/doi/10.1103/RevModPhys.40.359}
}

@misc{rumynin2021realrepresentationsc2gradedgroups,
      title={Real Representations of $C_2$-Graded Groups: The Antilinear Theory}, 
      author={Dmitriy Rumynin and James Taylor},
      year={2021},
      eprint={2006.09765},
      archivePrefix={arXiv},
      primaryClass={math.RT},
      doi={https://doi.org/10.1016/j.laa.2020.09.040},
      url={https://arxiv.org/abs/2006.09765}, 
}

@article{Kwan:2023kjp,
    author = "Kwan, Yves H. and Wilhelm, Patrick H. and Biswas, Sounak and Parameswaran, S. A.",
    title = "{Minimal Hubbard Models of Maximal Hilbert Space Fragmentation}",
    eprint = "2304.02669",
    archivePrefix = "arXiv",
    primaryClass = "cond-mat.stat-mech",
    doi = "10.1103/PhysRevLett.134.010411",
    journal = "Phys. Rev. Lett.",
    volume = "134",
    number = "1",
    pages = "010411",
    year = "2025"
}

@article{ganguli2025aspects,
  title={Aspects of {H}ilbert space fragmentation in the quantum {E}ast model: {F}ragmentation, subspace-restricted quantum scars, and effects of density-density interactions},
  author={Ganguli, Maitri and Aditya, Sreemayee and Sen, Diptiman},
  doi="10.1103/PhysRevB.111.045411",
  journal={Physical Review B},
  volume={111},
  number={4},
  pages={045411},
  year={2025},
  publisher={APS}
}

@article{PhysRevB.108.144308,
  title = {Freezing transition in the particle-conserving East model},
  author = {Wang, Cheng and Yang, Zhi-Cheng},
  journal = {Phys. Rev. B},
  volume = {108},
  issue = {14},
  pages = {144308},
  numpages = {7},
  year = {2023},
  month = {Oct},
  publisher = {American Physical Society},
  doi = {10.1103/PhysRevB.108.144308},
  url = {https://link.aps.org/doi/10.1103/PhysRevB.108.144308}
}

@misc{wang2025exponentiallyslowthermalization1d,
      title={Exponentially slow thermalization in 1D fragmented dynamics}, 
      author={Cheng Wang and Shankar Balasubramanian and Yiqiu Han and Ethan Lake and Xiao Chen and Zhi-Cheng Yang},
      year={2025},
      eprint={2501.13930},
      archivePrefix={arXiv},
      primaryClass={quant-ph},
      url={https://arxiv.org/abs/2501.13930}, 
}

@misc{SM,
  note = {See Supplemental Material for details.},
}

@article{moudgalya2020quantum,
  title={Quantum many-body scars in a {L}andau level on a thin torus},
  author={Moudgalya, Sanjay and Bernevig, B Andrei and Regnault, Nicolas},
  journal={Physical Review B},
  doi={10.1103/PhysRevB.102.195150},
  volume={102},
  number={19},
  pages={195150},
  year={2020},
  publisher={APS}
}

@article{ren2021quasisymmetry,
  title={Quasisymmetry {G}roups and {M}any-{B}ody {S}car {D}ynamics},
  author={Ren, Jie and Liang, Chenguang and Fang, Chen},
  doi={10.1103/PhysRevLett.126.120604},
  journal={Physical Review Letters},
  volume={126},
  number={12},
  pages={120604},
  year={2021},
  publisher={APS}
}

@article{shiraishi2017systematic,
  title={Systematic {C}onstruction of {C}ounterexamples to the {E}igenstate {T}hermalization {H}ypothesis},
  author={Shiraishi, Naoto and Mori, Takashi},
  doi={10.1103/PhysRevLett.119.030601},
  journal={Physical Review Letters},
  volume={119},
  number={3},
  pages={030601},
  year={2017},
  publisher={APS}
}

@article{pal2010many,
  title={Many-body localization phase transition},
  author={Pal, Arijeet and Huse, David A},
  doi="10.1103/PhysRevB.82.174411",
  journal={Physical Review B},
  volume={82},
  number={17},
  pages={174411},
  year={2010},
  publisher={APS}
}

@article{vosk2015theory,
  title={Theory of the many-body localization transition in one-dimensional systems},
  author={Vosk, Ronen and Huse, David A and Altman, Ehud},
  doi={10.1103/PhysRevX.5.031032},
  journal={Physical Review X},
  volume={5},
  number={3},
  pages={031032},
  year={2015},
  publisher={APS}
}

@article{nandkishore2015many,
  title={Many-body localization and thermalization in quantum statistical mechanics},
  author={Nandkishore, Rahul and Huse, David A},
  doi={10.1146/annurev-conmatphys-031214-014726},
  journal={Annu. Rev. Condens. Matter Phys.},
  volume={6},
  number={1},
  pages={15--38},
  year={2015},
  publisher={Annual Reviews}
}

@article{Hu_2024,
   title={From the quantum breakdown model to the lattice gauge theory},
   author={Hu, Yu-Min and Lian, Biao},
   doi={10.1007/s43673-024-00128-4},
   journal={AAPPS Bulletin},
   volume={34},
   number={1},
   pages={},
   year={2024},
   publisher={Springer Science and Business Media LLC}
}

@article{Hu_2025,
   title={Bosonic quantum breakdown Hubbard model},
   author={Hu, Yu-Min and Lian, Biao},
   doi={10.1103/1r4m-7psy},
   journal={Physical Review B},
   volume={112},
   number={10},
   pages={},
   year={2025},
   ISSN={2469-9969},
   url={http://dx.doi.org/10.1103/1r4m-7psy},
   publisher={American Physical Society (APS)},
   month=sep
}

@article{Liu2025QBM,
  title = {Two-dimensional quantum breakdown model with Krylov subspace many-body localization},
  author = {Liu, Xinyu and Lian, Biao},
  journal = {Phys. Rev. B},
  volume = {111},
  issue = {5},
  pages = {054302},
  numpages = {14},
  year = {2025},
  month = {Feb},
  publisher = {American Physical Society},
  doi = {10.1103/PhysRevB.111.054302},
  url = {https://link.aps.org/doi/10.1103/PhysRevB.111.054302}
}

@ARTICLE{hu2025glass,
       author = {{Hu}, Yu-Min and {Han}, Zhaoyu and {Lian}, Biao},
        title = "{Quantum Breakdown Condensate as a Disorder-Free Quantum Glass}",
      journal = {arXiv e-prints},
         year = 2025,
        month = dec,
archivePrefix = {arXiv},
       eprint = {2512.21847},
 primaryClass = {cond-mat.str-el},
       adsurl = {https://ui.adsabs.harvard.edu/abs/2025arXiv251221847H}
}

@article{PhysRevX.15.011068,
  title = {Highly Entangled Stationary States from Strong Symmetries},
  author = {Li, Yahui and Pollmann, Frank and Read, Nicholas and Sala, Pablo},
  journal = {Phys. Rev. X},
  volume = {15},
  issue = {1},
  pages = {011068},
  numpages = {37},
  year = {2025},
  month = {Mar},
  publisher = {American Physical Society},
  doi = {10.1103/PhysRevX.15.011068},
  url = {https://link.aps.org/doi/10.1103/PhysRevX.15.011068}
}

@article{PhysRevResearch.5.043239,
  title = {Hilbert space fragmentation in open quantum systems},
  author = {Li, Yahui and Sala, Pablo and Pollmann, Frank},
  journal = {Phys. Rev. Res.},
  volume = {5},
  issue = {4},
  pages = {043239},
  numpages = {17},
  year = {2023},
  month = {Dec},
  publisher = {American Physical Society},
  doi = {10.1103/PhysRevResearch.5.043239},
  url = {https://link.aps.org/doi/10.1103/PhysRevResearch.5.043239}
}

@article{wbzt-scvs,
  title = {Entanglement-cost hierarchies in quantum fragmented mixed states},
  author = {Sahu, Subhayan and Li, Yahui and Sala, Pablo},
  journal = {Phys. Rev. A},
  volume = {113},
  issue = {2},
  pages = {022406},
  numpages = {8},
  year = {2026},
  month = {Feb},
  publisher = {American Physical Society},
  doi = {10.1103/wbzt-scvs},
  url = {https://link.aps.org/doi/10.1103/wbzt-scvs}
}

@article{PhysRevB.109.064302,
  title = {Probing Hilbert space fragmentation and the block inverse participation ratio},
  author = {Frey, Philipp and Mikhail, David and Rachel, Stephan and Hackl, Lucas},
  journal = {Phys. Rev. B},
  volume = {109},
  issue = {6},
  pages = {064302},
  numpages = {13},
  year = {2024},
  month = {Feb},
  publisher = {American Physical Society},
  doi = {10.1103/PhysRevB.109.064302},
  url = {https://link.aps.org/doi/10.1103/PhysRevB.109.064302}
}

@article{PhysRevB.106.L220301,
  title = {Hilbert space fragmentation and interaction-induced localization in the extended Fermi-Hubbard model},
  author = {Frey, Philipp and Hackl, Lucas and Rachel, Stephan},
  journal = {Phys. Rev. B},
  volume = {106},
  issue = {22},
  pages = {L220301},
  numpages = {7},
  year = {2022},
  month = {Dec},
  publisher = {American Physical Society},
  doi = {10.1103/PhysRevB.106.L220301},
  url = {https://link.aps.org/doi/10.1103/PhysRevB.106.L220301}
}

@article{han2026quantumfragmentation,
    author = "Yiqiu Han and Oliver Hart and Alexey Khudorozhkov and Rahul Nandkishore",
    title = "Quantum Fragmentation",
    journal = "arXiv e-prints",
    eprint = "2604.06461",
    archivePrefix = "arXiv",
    primaryClass = "quant-ph",
    year = "2026"
}

@article{motzkin,
    author = "Movassagh, Ramis and Shor, Peter W.",
    title = "{Supercritical entanglement in local systems: Counterexample to the area law for quantum matter}",
    eprint = "1408.1657",
    archivePrefix = "arXiv",
    primaryClass = "quant-ph",
    doi = "10.1073/pnas.1605716113",
    journal = "Proc. Nat. Acad. Sci.",
    volume = "113",
    number = "47",
    pages = "13278",
    year = "2016"
}

@book{revuz2013continuous,
  title={Continuous martingales and Brownian motion},
  author={Revuz, Daniel and Yor, Marc},
  year={2013},
  publisher={Springer Science \& Business Media}
}

@article{gouezel2014local,
  title={Local limit theorem for symmetric random walks in Gromov-hyperbolic groups},
  author={Gou{\"e}zel, S{\'e}bastien},
  journal={Journal of the American Mathematical Society},
  volume={27},
  number={3},
  pages={893--928},
  year={2014}
}

@article{bjorklund2010central,
  title={Central limit theorems for Gromov hyperbolic groups},
  author={Bj{\"o}rklund, Michael},
  journal={Journal of theoretical probability},
  volume={23},
  number={3},
  pages={871--887},
  year={2010},
  publisher={Springer}
}

@article{lalley1993finite,
  title={Finite range random walk on free groups and homogeneous trees},
  author={Lalley, Steven P},
  journal={The Annals of Probability},
  pages={2087--2130},
  year={1993},
  publisher={JSTOR}
}
		

\clearpage
\widetext
\begin{center}
\textbf{\large Supplemental Material for\\\vspace{0.1cm}``Quantum Hilbert Space Fragmentation and Entangled Frozen States''}
\vspace{0.5cm}

Zihan Zhou, Tian-Hua Yang, and Bo-Ting Chen
\vspace{0.3cm}

\textit{Department of Physics, Princeton University, NJ 08544, USA}
\end{center}
\appendix

\tableofcontents

\section{Quantum Hilbert space fragmentation in the quantum breakdown model}
\label{appendix: quantum breakdown model}

The quantum breakdown model is a one-dimensional system that captures particle avalanche dynamics. The model comes in several variants, including fermionic, bosonic, and spin versions, all of which exhibit Hilbert space fragmentation.
In the hardcore bosonic formulation, each site $i$ hosts multiple flavors labeled by $\nu$, with the constraint that each flavor can be occupied by at most one boson. The Hamiltonian for a system of size $L$ with $N$ modes per site is given by
\begin{equation}
    H = 
    \sum_{i=1}^{L-1} \sum_{\nu=1}^{N} J_{i}^{\nu} (b_{i+1,1}^{\dagger} b_{i+1,2}^{\dagger} b_{i,\nu} 
    + \text{h.c.})
,\end{equation}
where $b_{i,\nu}$ and $b_{i,\nu}^{\dagger}$ and are hardcore boson creation and annihilation operators, satisfying the commutation relations
\begin{equation}
b_{i,\nu}b_{j,\nu'}^\dagger-(1-2\delta_{ij}\delta_{\nu\nu'})b_{j,\nu'}^\dagger b_{i,\nu}=\delta_{ij}\delta_{\nu\nu'}
,\end{equation}
and ``h.c.'' denotes the Hermitian conjugate.
Unlike conventional hopping models, the dynamics is kinetically constrained, leading to a fragmentation of the Hilbert space into many dynamically disconnected sectors.

A basis state can be visualized as a $N\times L$ grid of circles, where columns correspond to lattice sites $i=1, \cdots, L$, and rows correspond to the two modes $\nu=1,\cdots, N$. Filled circles ($\bullet$) denote occupied states, while empty circles ($\circ$) denote unoccupied states. For example, $
\boxed{
\begin{smallmatrix}
\bullet & \circ & \circ \\
\circ & \circ & \circ \\
\end{smallmatrix}}
$ represents a configuration of a $(L, N)=(3, 2)$ model in which only the $\nu=1$ mode at site $i=1$ is occupied.
The model can be described by glassy dynamics with local length $r=2$. In the basis $\left\{ \boxed{\begin{smallmatrix}
\bullet & \circ \\
\circ & \circ
\end{smallmatrix}} 
\,,\,
\boxed{\begin{smallmatrix}
\circ & \circ \\
\bullet & \circ
\end{smallmatrix}} 
\,,\,
\boxed{\begin{smallmatrix}
\circ & \bullet \\
\circ & \bullet
\end{smallmatrix}} 
\,,\,
\boxed{\begin{smallmatrix}
\bullet & \circ \\
\bullet & \circ
\end{smallmatrix}} 
\,,\,
\boxed{\begin{smallmatrix}
\circ & \bullet \\
\bullet & \bullet
\end{smallmatrix}} 
\,,\,
\boxed{\begin{smallmatrix}
\bullet & \bullet \\
\circ & \bullet
\end{smallmatrix}} 
\right\}$, the coupling matrix takes the form
\begin{equation}
    g^{w'w} = 
    \begin{pmatrix}
        0 & 0 & J^{1} & 0 & 0 & 0 \\
        0 & 0 & J^{2} & 0 & 0 & 0 \\
        J^{1} & J^{2} & 0 & 0 & 0 & 0 \\
        0 & 0 & 0 & 0 & J^{1} & J^{2} \\
        0 & 0 & 0 & J^{1} & 0 & 0 \\
        0 & 0 & 0 & J^{2} & 0 & 0 \\
    \end{pmatrix}_{w'w}
\end{equation}
where all other basis states at local length $r=2$ have been omitted, since their corresponding matrix elements vanish. The matrix $g^{w'w}$ is rank deficient for any choice of $J_{i}^{\nu}$, implying the existence of entangled frozen states and hence quantum fragmentation.

To illustrate quantum fragmentation more explicitly, we consider the $(L, N)=(3, 2)$ model. The system admits multiple classical Krylov subspaces; one such subspace (see Ref.~\cite{Chen:2026aqj} for a detailed analysis) is given by
\begin{equation}
\begin{aligned}
\mathcal{K}_{\text{cl}} &= 
\text{span}
\left(
    \boxed{\begin{smallmatrix}
    \bullet & \circ & \circ \\
    \bullet & \circ & \circ
    \end{smallmatrix}}
    \;,\;
    \boxed{\begin{smallmatrix}
    \circ & \bullet & \circ \\
    \bullet & \bullet & \circ
    \end{smallmatrix}}
    \;,\;
    \boxed{\begin{smallmatrix}
    \bullet & \bullet & \circ \\
    \circ & \bullet & \circ
    \end{smallmatrix}}
    \;,\;
    \boxed{\begin{smallmatrix}
    \circ & \circ & \bullet \\
    \bullet & \bullet & \bullet
    \end{smallmatrix}}
    \;,\;
    \boxed{\begin{smallmatrix}
    \circ & \bullet & \bullet \\
    \bullet & \circ & \bullet
    \end{smallmatrix}}
    \;,\;
    \boxed{\begin{smallmatrix}
    \bullet & \circ & \bullet \\
    \circ & \bullet & \bullet
    \end{smallmatrix}}
    \;,\;
    \boxed{\begin{smallmatrix}
    \bullet & \bullet & \bullet \\
    \circ & \circ & \bullet
    \end{smallmatrix}}
\right)
\end{aligned}
\end{equation}
This subspace contains three frozen states that are annihilated by the Hamiltonian, which span:
\begin{equation}
\begin{aligned}
\mathcal{K}_{\text{EFS}} =& 
\text{span}
\left(
J^{1}
\boxed{\begin{smallmatrix}
\bullet & \bullet & \bullet \\
\circ & \circ& \bullet
\end{smallmatrix}}
-
J^{2}
\boxed{\begin{smallmatrix}
\bullet & \circ &\bullet \\
\circ & \bullet & \bullet
\end{smallmatrix}}
\right)
\oplus
\text{span}
\left(
J^{1}
\boxed{\begin{smallmatrix}
\circ & \bullet & \bullet \\
\bullet & \circ& \bullet
\end{smallmatrix}}
-
J^{2}
\boxed{\begin{smallmatrix}
\circ & \circ &\bullet \\
\bullet & \bullet & \bullet
\end{smallmatrix}}
\right)
\\
\oplus &
\text{span}
\left(
2J^{1} J^{2}
\boxed{\begin{smallmatrix}
\bullet & \circ & \circ \\
\bullet & \circ & \circ
\end{smallmatrix}}
-
J^{1} J^{2}
\boxed{\begin{smallmatrix}
\circ & \circ & \bullet \\
\bullet & \bullet & \bullet
\end{smallmatrix}}
-
J^{1} J^{1}
\boxed{\begin{smallmatrix}
\circ & \bullet & \bullet \\
\bullet & \circ & \bullet
\end{smallmatrix}}
-
J^{2} J^{1}
\boxed{\begin{smallmatrix}
\bullet & \bullet & \bullet \\
\circ & \circ & \bullet
\end{smallmatrix}}
-
J^{2} J^{2}
\boxed{\begin{smallmatrix}
\bullet & \circ & \bullet \\
\circ & \bullet & \bullet
\end{smallmatrix}}
\right)
\end{aligned}
.\end{equation}
Projecting out $\mathcal{K}_{\text{EFS}}$, one obtains two distinct mobile quantum Krylov subspaces, each of dimension two:
\begin{equation}
\begin{aligned}
\mathcal{K}_q =& 
\text{span} \left(
J^{2} 
\boxed{\begin{smallmatrix}
\circ & \bullet & \circ \\
\bullet & \bullet & \circ
\end{smallmatrix}}
- 
J^{1}
\boxed{\begin{smallmatrix}
\bullet & \bullet & \circ \\
\circ & \bullet & \circ
\end{smallmatrix}}
\,,\,
J^{2} J^{1}
\boxed{\begin{smallmatrix}
\circ & \circ & \bullet \\
\bullet & \bullet & \bullet
\end{smallmatrix}}
+
J^{2} J^{2}
\boxed{\begin{smallmatrix}
\circ & \bullet & \bullet \\
\bullet & \circ & \bullet
\end{smallmatrix}}
-
J^{1} J^{1}
\boxed{\begin{smallmatrix}
\bullet & \circ & \bullet \\
\circ & \bullet & \bullet
\end{smallmatrix}}
-
J^{1} J^{2}
\boxed{\begin{smallmatrix}
\bullet & \bullet & \bullet \\
\circ & \circ & \bullet
\end{smallmatrix}}
\right)
\\
\otimes&
\text{span} \left(
J^{1} 
\boxed{\begin{smallmatrix}
\circ & \bullet & \circ \\
\bullet & \bullet & \circ
\end{smallmatrix}}
+
J^{2}
\boxed{\begin{smallmatrix}
\bullet & \bullet & \circ \\
\circ & \bullet & \circ
\end{smallmatrix}}
\,,\,
J^{1} J^{1}
\boxed{\begin{smallmatrix}
\circ & \circ & \bullet \\
\bullet & \bullet & \bullet
\end{smallmatrix}}
+
J^{1} J^{2}
\boxed{\begin{smallmatrix}
\circ & \bullet & \bullet \\
\bullet & \circ & \bullet
\end{smallmatrix}}
+
J^{2} J^{1}
\boxed{\begin{smallmatrix}
\bullet & \circ & \bullet \\
\circ & \bullet & \bullet
\end{smallmatrix}}
+
J^{2} J^{2}
\boxed{\begin{smallmatrix}
\bullet & \bullet & \bullet \\
\circ & \circ & \bullet
\end{smallmatrix}}
+
(J^{1} J^{1} + J^{2} J^{2})
\boxed{\begin{smallmatrix}
\bullet & \circ & \circ \\
\bullet & \circ & \circ
\end{smallmatrix}}
\right)
\end{aligned}
\end{equation}

\section{Quantum Hilbert space fragmentation in the East model}
\label{appendix: east model}

The East model is another model in which quantum Hilbert space fragmentation is reported~\cite{Brighi:2022kca,ganguli2025aspects,PhysRevB.108.144308}. We show that the mechanism of quantum fragmentation there is different from what is discussed in this Letter, and as a result, the quantum fragmentation is weaker than in the models we proposed.

We consider the range-2 particle-conserving East model, whose local Hamiltonian term can be expressed as
\begin{equation}
h_i = \left[t_1 n_{i-1} + t_2(1-n_{i-1})n_{i-2}\right]\left(c_i^\dagger c_{i+1}+h.c.\right).
\end{equation}
In the language of the semigroup word problem, this model has two equivalence classes,
\begin{align}
R_1:& 110 \sim 101, \\
R_2:& 1010\sim 1001.
\end{align}
Seeing $h_i$ as a whole, $R_1$ can act on the first three sites or the last three sites in its four-site support. Therefore, $h_i$ can be written as
\begin{align}
h_i & = t_1 (|1100\rangle \langle 1010| + h.c.) + t_1 (|1101\rangle \langle 1011| + h.c.)  + t_1 (|0110\rangle \langle 0101| + h.c.) + t_1 (|1110\rangle \langle 1101| + h.c.) \nonumber \\
& + t_2 (|1010\rangle \langle 1001| + h.c.).
\end{align}
This is block-diagonal with respect to the number of $1$'s and the position of the first $1$. In the first block,
\begin{equation}
\left. g^{w^\prime w} \right|_{1110,1101,1011} = \begin{pmatrix}
 & t_1 & \\ t_1 & & t_1 \\ & t_1 &
\end{pmatrix}.
\end{equation}
This leads to the entangled frozen state
\begin{equation}
|f_1\rangle = |1110\rangle - |1011\rangle.
\end{equation}
In the second block,
\begin{equation}
\left. g^{w^\prime w} \right|_{1100,1010,1001} = \begin{pmatrix}
 & t_1 & \\ t_1 & & t_2 \\ & t_2 &
\end{pmatrix}.
\end{equation}
This leads to the entangled frozen state
\begin{equation}
|f_2\rangle = t_2 |1100\rangle - t_1 | 1001\rangle.
\end{equation}
In the third block,
\begin{equation}
\left. g^{w^\prime w} \right|_{0110,0101} = \begin{pmatrix}
 & t_1 \\ t_1 & 
\end{pmatrix}.
\end{equation}
This gives no entangled frozen states.

The problem remains whether these entangled frozen state can properly ``concatenate'' on a system with $L>4$. In fact, we show this is not possible for $|f_1\rangle$. We have
\begin{align}
h_{i+1} |f_1\rangle \otimes |0\rangle & = -t_1 |10101\rangle, \\
h_{i+1} |f_1\rangle \otimes |1\rangle & = t_1 |11110\rangle.
\end{align}
Therefore, we cannot make the state frozen on the next four sites. In fact, for $|f_1\rangle \otimes |0\rangle$, we get $|11100\rangle$ and $|10110\rangle$. Since $0110$ lies in a sector with no entangled frozen state, no addition can make $|f_1\rangle\otimes |0\rangle$ frozen. Similarly, in $|f_1\rangle \otimes |1\rangle$, we get $|11101\rangle$, where $1101$ is a pattern that does not appear in any frozen state.

By contrast, $|f_2\rangle$ can be extended to larger systems. To begin with, one can always pad $|f_2\rangle$ with zeros to the right. More generally, one can always form states like
\begin{equation}
|f_2\rangle \otimes |0\rangle^{m_1} \otimes |f_2\rangle \otimes |0\rangle^{m_2}\otimes |f_2\rangle \dots,
\end{equation}
where $m_i\geq 2$. These are guaranteed to be entangled frozen states, albeit the entanglement is purely localized.

More generally, we can construct an entangled frozen state $|\text{EFS}\rangle = \sum_b c_b |b\rangle$, where $b$ are bit-strings, as follows:
\begin{itemize}
    \item Whenever $b$ contains a pattern $1100$ at four consecutive bits, there must be a $b^\prime$ equal to $b$ but with $1100$ replaced by $1001$, and that $t_1 c_b + t_2 c_{b^\prime}=0$.
    \item Any $b$ that contains a pattern $101$ must have $c_b=0$.
\end{itemize}
Note that the second conditions guarantees that the only update rules that can happen is $1100\to 1010$ and $1001\to 1010$, which combined with the first condition ensures that $|\text{EFS}\rangle$ is frozen.

We notice that the constraint $t_1 c_b + t_2 c_{b^\prime}=0$ can always be satisfied given a set of basis states connected by the $1100\leftrightarrow 1001$ rule. In fact, we can define $\mu$ as the ``dipole moment'' of $1$, defined as the sum of the positions of $1$ in the string. Then obviously, $\mu_{b^\prime}=\mu_{b}+2$. Therefore, we can assign the coefficients such that $c_b \propto (-t_1/t_2)^{\mu_b/2}$.

\section{Sub-volume Law of entangled frozen states in GHZ and Asymmetric Projector Model}
\label{app:efs-entanglement}

In this appendix, we demonstrate that the EFS constructed in the GHZ projector and the asymmetric model exhibit an interesting feature: their bipartite entanglement entropy is related to combinatoric properties of the Krylov subspace dimensions. In particular, the EFS in the largest Krylov subspaces will scale as $\sqrt{L}$, in stark contrast to the expectation that Krylov subspaces in HSF systems should be weakly-entangled. The mechanism bears similarity to the power-law entangled frustration-free ground state of the colored Motzkin spin chain~\cite{motzkin}.

\subsection{Setup: the EFS ansatz}
\label{subsec:efs-setup}

In the GHZ projector model, each classical mobile Krylov sector contains $k \geq 1$ mobile triplets and an $|\text{N3C}\rangle$ frozen tail of length $L - 3k$. The entangled frozen state (EFS) in such a sector is the signed uniform superposition
\begin{equation}
\label{eq:EFS-general}
|\mathrm{EFS}\rangle = \frac{1}{\sqrt{D(L)}} \sum_{w} (-1)^{N_1(w)} |w\rangle,
\end{equation}
where the sum runs over all length-$L$ strings $w$ in the sector, $D(L)$ is the classical Krylov dimension, and $N_1(w)$ counts the $1$'s in $w$. The sign $(-1)^{N_1}$ flips under any local $000 \leftrightarrow 111$ update, so each pair of related strings contributes $(|000\rangle - |111\rangle) \otimes |\text{rest}\rangle \propto |\text{GHZ}^-\rangle$ on the three-site window. Since $|\text{GHZ}^-\rangle \perp |\text{GHZ}\rangle$, every local projector annihilates $|\text{EFS}\rangle$.

\subsection{The group $\mathbb{Z}_3 * \mathbb{Z}_3$ and its Cayley graph}
\label{subsec:z3star}

The semigroup $\langle 0, 1 \mid 000 = 111 \rangle$ lifts to a group by declaring both triplets $000$ and $111$ to equal the identity $e$. The generator $0$ then has order three ($0 \cdot 0 \cdot 0 = e$) and generates a copy of $\mathbb{Z}_3$, likewise $1$ generates another $\mathbb{Z}_3$. Since there is no relation mixing the two generators, the resulting group is the free product
\begin{equation}
G = \mathbb{Z}_3 * \mathbb{Z}_3.
\end{equation}
Because of the relation $000 = 111 = e$, each element of $G$ can be
brought into a unique reduced form by repeatedly deleting any $000$
or $111$ substring from the word until no further deletion is
possible. The resulting string contains no three consecutive
identical characters, and therefore consists of alternating blocks of
$0$'s and $1$'s with each block of length $1$ or $2$. These reduced
words coincide exactly with the N3C frozen strings, establishing a
bijection between $G$ and the set of N3C frozen strings. For instance,
\begin{equation}
  0011100 \;\xrightarrow{\,111 \to e\,}\; 0000
  \;\xrightarrow{\,000 \to e\,}\; 0,
\end{equation}
so the reduced form of $0011100$ is $0 \in G$.

The
classical Krylov sector containing $w$ is therefore the entire
preimage $f^{-1}(f(w)) \cap \{0,1\}^L$. This sets up a bijection
between classical Krylov sectors of the length-$L$ chain and N3C
strings of length $L-3k$ ($k\geq 0$). The all-mobile sector corresponds to the identity $e \in G$
and exists when $3 \mid L$. The \textbf{Cayley graph} of $G$ is a graph where nodes are all the group elements of $G$, and directed edges correspond to the application of the two generating elements, $0$ and $1$. More specifically, an edge exists pointing from $g \in G$ to $h\in G$ if $h$ can be obtained by appending a generator to $g$, i.e., $h$ is equal to $g0$ or $g1$. The edges are colored corresponding to the generator applied. Since both $0$ and $1$ are generators of $\mathbb Z_3$, the Cayley graph would locally consist of triangles, as applying the same generator thrice returns to the original group element. The full graph can be spanned by repeated attaching such triangles, as demonstrated in Fig.~\ref{fig:cayley}. The identity element $e$ has two triangles attached to it, corresponding to two copies of the $\mathbb Z_3$. Each existing triangle has two triangles of the opposite color attached to it. If each triangle is abstracted to the point, the Cayley graph would have the structure of a binary tree.

The Cayley graph provides a visualization of words in the group: a word of length $L$ is exactly a path on the Cayley graph starting from $e$ and walking $L$ steps along the directed edges. The corresponding group element would be the node on which the walk ends. The depth of a word is defined as the minimum number of steps it would take to reach the corresponding node from the origin $e$; equivalently, it is the length of the N3C string.

\begin{figure}[t]
\centering
\resizebox{0.6\linewidth}{!}{%
\begin{tikzpicture}[
  vertex/.style={fill=black, circle, inner sep=0pt, minimum size=3pt},
  identity/.style={fill=white, draw=black, very thick, circle, inner sep=0pt, minimum size=6pt},
  edgeZ0/.style={draw=blue!65!black, thick, -Stealth},
  edgeZ1/.style={draw=red!65!black, thick, -Stealth},
]

\begin{scope}

\node[identity] (e) at (0,0) {};
\node[font=\normalsize, anchor=north] at (0, -0.15) {$e$};

\node[vertex] (r0) at (1.7, 0.7) {};
\node[vertex] (r00) at (1.7, -0.7) {};
\node[font=\normalsize, anchor=south] at (1.8, 0.8) {$0$};
\node[font=\normalsize, anchor=north] at (1.8, -0.9) {$00$};
\draw[edgeZ0] (e) -- (r0);
\draw[edgeZ0] (r00) -- (e);
\draw[edgeZ0] (r0) -- (r00);

\node[vertex] (l1) at (-1.7, 0.7) {};
\node[vertex] (l11) at (-1.7, -0.7) {};
\node[font=\normalsize, anchor=south] at (-1.8, 0.8) {$1$};
\node[font=\normalsize, anchor=north] at (-1.8, -0.9) {$11$};
\draw[edgeZ1] (e) -- (l1);
\draw[edgeZ1] (l11) -- (e);
\draw[edgeZ1] (l1) -- (l11);

\node[vertex] (r01) at (3.4, 1.6) {};
\node[vertex] (r011) at (3.4, 0.2) {};
\node[font=\normalsize, anchor=west] at (3.5, 1.6) {$01$};
\node[font=\normalsize, anchor=west] at (3.5, 0.2) {$011$};
\draw[edgeZ1] (r0) -- (r01);
\draw[edgeZ1] (r011) -- (r0);
\draw[edgeZ1] (r01) -- (r011);

\node[vertex] (r001) at (3.4, -0.2) {};
\node[vertex] (r0011) at (3.4, -1.6) {};
\node[font=\normalsize, anchor=west] at (3.5, -0.2) {$001$};
\node[font=\normalsize, anchor=west] at (3.5, -1.6) {$0011$};
\draw[edgeZ1] (r00) -- (r001);
\draw[edgeZ1] (r0011) -- (r00);
\draw[edgeZ1] (r001) -- (r0011);

\node[vertex] (l10) at (-3.4, 1.6) {};
\node[vertex] (l100) at (-3.4, 0.2) {};
\node[font=\normalsize, anchor=east] at (-3.5, 1.6) {$10$};
\node[font=\normalsize, anchor=east] at (-3.5, 0.2) {$100$};
\draw[edgeZ0] (l1) -- (l10);
\draw[edgeZ0] (l100) -- (l1);
\draw[edgeZ0] (l10) -- (l100);

\node[vertex] (l110) at (-3.4, -0.2) {};
\node[vertex] (l1100) at (-3.4, -1.6) {};
\node[font=\normalsize, anchor=east] at (-3.5, -0.2) {$110$};
\node[font=\normalsize, anchor=east] at (-3.5, -1.6) {$1100$};
\draw[edgeZ0] (l11) -- (l110);
\draw[edgeZ0] (l1100) -- (l11);
\draw[edgeZ0] (l110) -- (l1100);

\node[font=\normalsize, text=gray!60] at (4.5, 0.5) {$\cdots$};
\node[font=\normalsize, text=gray!60] at (4.5, -1.0) {$\cdots$};
\node[font=\normalsize, text=gray!60] at (-4.5, 0.5) {$\cdots$};
\node[font=\normalsize, text=gray!60] at (-4.5, -1.0) {$\cdots$};

\node[font=\large, anchor=west] at (-4.7, -2.6) {\textcolor{blue!65!black}{blue edges}: $\mathbb{Z}_3^{(0)}$\quad \textcolor{red!65!black}{red edges}: $\mathbb{Z}_3^{(1)}$};
\end{scope}

\end{tikzpicture}%
}
\caption{Cayley graph of $G = \mathbb{Z}_3 * \mathbb{Z}_3$ with generators $\{0, 1\}$: blue triangles are the $\mathbb{Z}_3^{(0)} = \langle 0 \mid 0^3 = e\rangle$ cosets and red triangles the $\mathbb{Z}_3^{(1)}$ cosets. Each non-identity vertex attaches a triangle of the opposite color. Collapsing each triangle to a point gives a binary tree.}
\label{fig:cayley}
\end{figure}
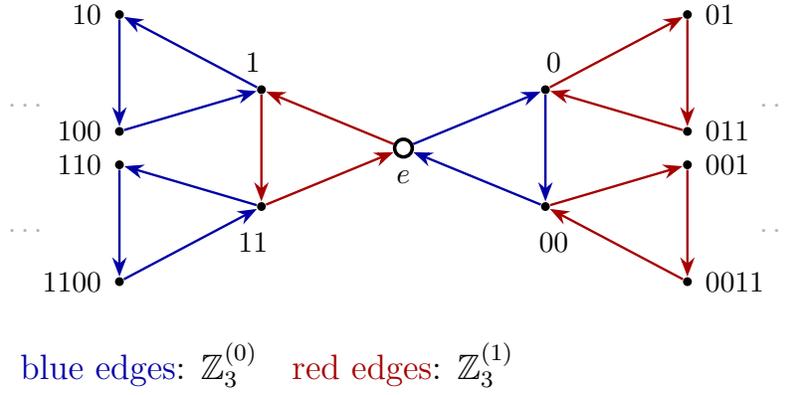

\subsection{Schmidt decomposition of EFS}
\label{subsec:schmidt-general}

We bisect the chain at any position $L_A \in \{1,\ldots,L-1\}$, writing
$w = w_A w_B$ with $w_A \in \{0,1\}^{L_A}$ and $w_B \in \{0,1\}^{L_B}$
where $L_B = L - L_A$. The EFS then takes the form
\begin{equation}
\label{eq:EFS-bipart}
  |\mathrm{EFS}\rangle 
  = \frac{1}{\sqrt{D(L)}}
    \sum_{(w_A, w_B)} (-1)^{N_1(w_A) + N_1(w_B)}\,
    |w_A\rangle \otimes |w_B\rangle,
\end{equation}
the sum restricted by the condition $f(w_A w_B) = c_f$ with $c_f$ the reduced N3C frozen strings labeling the classical Krylov subspace. Since
$f$ is a homomorphism, the constraint factorizes:
\begin{equation}
\label{eq:bipart-constraint}
  f(w_A) \cdot f(w_B) = c_f,
  \qquad \Leftrightarrow \qquad
  f(w_B) = f(w_A)^{-1} c_f.
\end{equation}
For each prefix reduction $f(w_A)$, the suffix reduction $f(w_B)$ is uniquely determined. We will shortly see that this single algebraic fact will lead to
an exact Schmidt decomposition. We illustrate the procedure of calculating $f(w_A)$ and $f(w_B)$ in Fig.~\ref{fig:efs-boundary}
\begin{figure}[h]
\centering
\begin{tikzpicture}[
  x=0.62cm, y=0.62cm, >=Stealth,
  site0/.style={draw, thick, minimum size=0.58cm, inner sep=0pt, font=\small, fill=blue!12},
  site1/.style={draw, thick, minimum size=0.58cm, inner sep=0pt, font=\small, fill=red!12},
  masked/.style={draw, thick, minimum size=0.58cm, inner sep=0pt, font=\small,
                 fill=gray!15, text=gray!50, draw=gray!40},
  brace/.style={decorate, decoration={brace, amplitude=4pt, raise=4pt}},
  cut/.style={ultra thick, orange ,dashed}
]
\node[anchor=east, font=\small] at (-0.5, 0) {$w =$};
\node[masked] (s0) at (0,0) {0};
\node[masked] (s1) at (1,0) {0};
\node[masked] (s2) at (2,0) {0};
\node[site0]  (s3) at (3,0) {0};
\node[site0]  (s4) at (4,0) {0};
\node[site1]  (s5) at (5,0) {1};
\node[site1]  (s6) at (6,0) {1};
\node[site1]  (s7) at (7,0) {1};
\node[site0]  (s8) at (8,0) {0};
\node[masked] (s9) at (9,0) {0};
\node[masked] (s10) at (10,0) {0};
\node[masked] (s11) at (11,0) {0};
\draw[gray!50, thick] (s0.south west) -- (s2.north east);
\draw[gray!50, thick] (s9.south west) -- (s11.north east);
\draw[cut] (5.5, -0.7) -- (5.5, 1.0);
\node[red!70!black, font=\scriptsize, anchor=south] at (5.5, 1.0) {cut at $L_A=6$};
\draw[brace] (s0.north west) -- node[above=8pt, font=\small\itshape] {$w_A$} (s5.north east);
\draw[brace] (s6.north west) -- node[above=8pt, font=\small\itshape] {$w_B$} (s11.north east);
\draw[->, thick, gray!60] (4, -0.55) -- node[right=2pt, font=\scriptsize, gray!60!black] {$f$} (4, -1.55);
\draw[->, thick, gray!60] (7, -0.55) -- node[right=2pt, font=\scriptsize, gray!60!black] {$f$} (7, -1.55);
\node[site0, ultra thick] at (3, -2.0) {0};
\node[site0, ultra thick] at (4, -2.0) {0};
\node[site1, ultra thick] at (5, -2.0) {1};
\node[site1, ultra thick] at (6, -2.0) {1};
\node[site1, ultra thick] at (7, -2.0) {1};
\node[site0, ultra thick] at (8, -2.0) {0};
\node[font=\small, anchor=east] at (3-0.4, -2.0) {};
\node[font=\small, anchor=west] at (5+0.4, -2.0) {};
\node[font=\small, anchor=east] at (6-0.4, -2.0) {};
\node[font=\small, anchor=west] at (8+0.4, -2.0) {};
\draw[cut] (5.5, -2.88) -- (5.5, -1.1);
\node[font=\small] at (5.5, -3.3)
  {$f(w_A)  f(w_B)=e$ };
\end{tikzpicture}
\caption{Bipartition of $w = 000001\,110000$ at the cut $L_A = 6$
in the all-mobile sector ($c_f = e$). After removing the characters $000$ by the reduction $f$, the remaining
sites give the boundary labels $f(w_A) = 001$ and $f(w_B) = 110$. }
\label{fig:efs-boundary}
\end{figure}
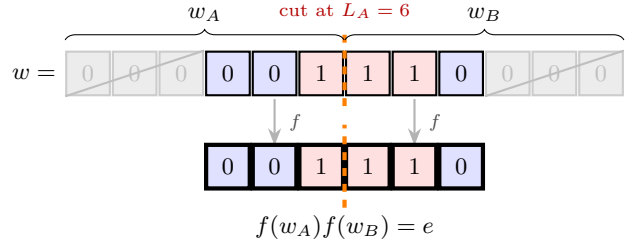

The group inverse has an explicit combinatorial description. Each
reduced word $g \in G$ has a unique block decomposition
\begin{equation}
  g = c_1^{k_1} c_2^{k_2} \cdots c_d^{k_d},
  \qquad c_i \in \{0,1\},\ c_i \neq c_{i+1},\ k_i \in \{1,2\},
\end{equation}
with $d$ called the depth of $g$. Reversing the block order and
swapping each length $k \leftrightarrow 3-k$ yields the inverse:
\begin{equation}
\label{eq:closure}
  g^{-1} = c_d^{3-k_d} c_{d-1}^{3-k_{d-1}} \cdots c_1^{3-k_1}.
\end{equation}
Since $3 - k_i \in \{1,2\}$ whenever $k_i \in \{1,2\}$, the right-hand
side is itself a reduced word, so the inverse of a frozen string is
again a frozen string. For instance, $(011)^{-1} = (0^1 1^2)^{-1} = 
1^{1} 0^{2} = 100$, and one checks directly that $011 \cdot 100 = 
011100 \to 0111\,00 \to 0\,00 \to e$. 

Let $c_A = f(w_A)$ act as the boundary label across the cut. It records the group element of the left half of the chain. Grouping terms in Eq.~\eqref{eq:EFS-bipart} by $c_A$ and using Eq.~\eqref{eq:closure} to identify the constraint on $w_B$,
\begin{equation}
\label{eq:EFS-Schmidt}
|\mathrm{EFS}\rangle = \frac{1}{\sqrt{D(L)}} \sum_{c_A} \left(\sum_{\substack{|w_A|=L_A\\ f(w_A)=c_A}} (-1)^{N_1(w_A)} |w_A\rangle\right) \otimes \left(\sum_{\substack{|w_B|=L_B\\ f(w_B)=c_A^{-1} c_f}} (-1)^{N_1(w_B)} |w_B\rangle\right).
\end{equation}
Each word has a unique N3C reduction, so the sets
$\{w_A \in \{0,1\}^{L_A} : f(w_A) = c_A\}$ partition $\{0,1\}^{L_A}$
as $c_A$ ranges over $G$. The left vectors corresponding to
different $c_A$ are therefore supported on disjoint basis subsets
and are mutually orthogonal. Similar arguments can be used to show the orthogonal of the right vectors.  Equation~\eqref{eq:EFS-Schmidt} is therefore an exact
Schmidt decomposition.

Normalizing the left and right factors, the Schmidt weights are
\begin{equation}
\label{eq:Schmidt-weights}
p_{c_A} = \frac{D_{c_A}(L_A)\, D_{c_{A}^{-1} c_f}(L_B)}{D_{c_f}(L)},
\end{equation}
where $D_c(\ell)$ is the number of length-$\ell$ words whose reduction equals $c$. The bipartite entanglement entropy is exactly the Shannon entropy $H(p)$ of the distribution $\{p_{c_A}\}$:
\begin{equation}
\label{eq:S-is-H}
S(L_A) = H(p) = -\sum_{c_A} p_{c_A} \log p_{c_A}.
\end{equation}

As a concrete example, we consider $L = 9$ with the cut $L_A = 4$,
$L_B = 5$. The all-mobile sector has $D(9) = 38$ states. There are
ten possible boundary labels $c_A = f(w_A) \in G$, i.e. ten distinct
N3C reductions of length-$4$ prefixes that arise from sector
elements. For each label $c_A$, Eq.~\eqref{eq:closure} fixes the suffix reduction
$f(w_B) = c_A^{-1}e = c_A^{-1}$, and the Schmidt weight is
$p_{c_A} = D_{c_A}(L_A)\, D_{c_A^{-1}}(L_B) / D(L)$:

\begin{center}
\begin{tabular}{cccccc}
\toprule
$c_A = f(w_A)$ & $c_A^{-1}$ & $D_{c_A}(4)$ & $D_{c_A^{-1}}(5)$ & $D_{c_A}\, D_{c_A^{-1}}$ & $p_{c_A}$ \\
\midrule
$0$    & $00$    & 3 & 4 & 12 & $12/38$ \\
$1$    & $11$    & 3 & 4 & 12 & $12/38$ \\
$0011$ & $10$    & 1 & 4 &  4 & $\phantom{0}4/38$ \\
$1100$ & $01$    & 1 & 4 &  4 & $\phantom{0}4/38$ \\
$0010$ & $00110$ & 1 & 1 &  1 & $\phantom{0}1/38$ \\
$0100$ & $01100$ & 1 & 1 &  1 & $\phantom{0}1/38$ \\
$0110$ & $00100$ & 1 & 1 &  1 & $\phantom{0}1/38$ \\
$1001$ & $11011$ & 1 & 1 &  1 & $\phantom{0}1/38$ \\
$1011$ & $10011$ & 1 & 1 &  1 & $\phantom{0}1/38$ \\
$1101$ & $11001$ & 1 & 1 &  1 & $\phantom{0}1/38$ \\
\bottomrule
\end{tabular}
\end{center}
The probabilities sum to $1$
as required, and the Shannon entropy is
\begin{equation}
  S \;=\; -\sum_c p_c \log_2 p_c
  \;\approx\; 2.563 ,
\end{equation}
in agreement with direct singular-value decompositiosn of the EFS
amplitude matrix. The explicit Schmidt decomposition is given by
\begin{align}
\sqrt{D(9)}\,\ket{\mathrm{EFS}}
  \;=\;
  &\bigl(\ket{0000} - \ket{0111} - \ket{1110}\bigr)
   \otimes
   \bigl(\ket{00000} - \ket{00111} - \ket{01110} - \ket{11100}\bigr)
   \nonumber\\
  +\;
  &\bigl(-\ket{0001} - \ket{1000} + \ket{1111}\bigr)
   \otimes
   \bigl(\ket{00011} + \ket{10001} + \ket{11000} - \ket{11111}\bigr)
   \nonumber\\
  +\;
  &\ket{0011}
   \otimes
   \bigl(-\ket{00010} - \ket{10000} + \ket{10111} + \ket{11110}\bigr)
   \nonumber\\
  +\;
  &\ket{1100}
   \otimes
   \bigl(-\ket{00001} - \ket{01000} + \ket{01111} + \ket{11101}\bigr)
   \nonumber\\
  +\;
  &\ket{0010}\otimes \bigl(-\ket{00110}\bigr)
   \;+\;
   \ket{0100}\otimes\bigl(-\ket{01100}\bigr)
   \;+\;
   \ket{0110}\otimes\bigl(-\ket{00100}\bigr)
   \nonumber\\
  +\;
  &\ket{1001}\otimes\ket{11011}
   \;+\;
   \ket{1011}\otimes\bigl(-\ket{10011}\bigr)
   \;+\;
   \ket{1101}\otimes\bigl(-\ket{11001}\bigr).
\end{align}

\subsection{The $\sqrt{L}$ scaling}
\label{subsec:sqrtL}

The identity $S(L_A) = H(p)$ reduces the entanglement calculation to
a purely combinatorial question: the distribution of
$p_{c_A} = D_{c_A}(L_A)\,D_{c_A^{-1} c_f}(L_B)/D_{c_f}(L)$ over group elements
$c_A \in G$. We will show that, for the all-mobile sector $c_f = e$
and the bipartition $L_A = L/2$, this distribution gives a sub-volume
entanglement scaling $S(L_A) \sim \sqrt{L}$.

The count $D_c(\ell)$ is the dimension of the classical Krylov subspace. For the triplet-flip model, Krylov dimension only depends on $\ell$ and the number of triplets. Equivalently speaking, the count $D_c(\ell)$ depends only on the depth of $c$ in the Cayley graph, not on the specific element. Writing $|c| = \ell - 3k$, so that $k$ is the number of mobile triplets in a length-$\ell$ string reducing to $c$, one has
\begin{equation}
\label{eq:Dk-closed}
D_k(\ell) \equiv D_c(\ell)\big|_{|c| = \ell - 3k} = \binom{\ell}{k} - \sum_{j=0}^{k-1}\binom{\ell}{j},
\end{equation}
which will be derived in Appendix~\ref{sec:D-cl-tripflip}. Substituting Eq.~\eqref{eq:Dk-closed} into the Shannon entropy Eq.~\eqref{eq:S-is-H} and applying Stirling's formula yields $S \sim \sqrt{L}$ by a saddle-point analysis. The detailed calculation will be carried out in Appendix~\ref{eq:subsec-efs-entanglement}.

It is illuminating to recognize that the Schmidt weight $p_{c_A}$ has a direct probabilistic interpretation. The numerator $D_{c_A}(L_A) D_{c_A^{-1} c_f}(L_B)$ counts pairs $(w_A, w_B)$ such that $f(w_A) = c_A$ and $f(w_B) = c_A^{-1} c_f$, i.e., concatenated strings $w = w_A w_B$ that pass through the group element $c_A$ at time $t = L_A$ on their way from $e$ at $t=0$ to $c_f$ at $t=L$. The denominator $D_{c_f}(L)$ counts \emph{all} such length-$L$ strings ending at $c_f$. Therefore
\begin{equation}
\label{eq:p-as-bridge}
p_{c_A} = \frac{\#\{w \in \{0,1\}^L : f(w) = c_f,\; f(w_{[1:L_A]}) = c_A\}}{\#\{w \in \{0,1\}^L : f(w) = c_f\}} = \Pr\!\left( X_{L_A} = c_A \,\Big|\, X_0 = e,\; X_L = c_f \right),
\end{equation}
where $X_t \equiv f(w_{[1:t]})$ is the position of the walker on the Cayley graph after $t$ steps. The Schmidt weights are therefore exactly the time-$L_A$ marginal of a uniform random bridge~\cite{revuz2013continuous} on the Cayley graph of $G = \mathbb{Z}_3 * \mathbb{Z}_3$, conditioned to start at $e$ and end at $c_f$ after $L$ steps. The entanglement entropy
\begin{equation}
S(L_A) = H\bigl(\Pr(X_{L_A} = \cdot \mid \text{bridge from $e$ to $c_f$})\bigr)
\end{equation}
is therefore the Shannon entropy of the bridge-walk position at time $L_A$. This identification is the conceptual bridge between the combinatorial closed form above and the random-walk argument that follows.

We now extract a quantitative relation between the entropy $S(L_A)$ and the average distance $\langle |c_A| \rangle$ travelled by the bridge walk. The Cayley graph of $G = \mathbb{Z}_3 * \mathbb{Z}_3$ is tree-like as shown in Fig.~\ref{fig:cayley}, and we organize each vertex $c_A$ by its position on this tree using two coordinates:
\begin{itemize}
\item The radial coordinate $r \equiv |c_A|$ is the graph distance from the identity $e$ to $c_A$, equivalently the length of the reduced-word representation of $c_A$.
\item The angular coordinate $\theta(c_A)$ distinguishes between the $N_r$ different vertices that all sit at the same depth $r$. For instance, at depth $r=3$ there are $N_3 = 6$ reduced words $\{001, 010, 011, 100, 101, 110\}$: they are all the same radial distance from $e$, but $\theta$ specifies which of these six length-$3$ words we have. In general, since a reduced word $c_A = c_1^{k_1} c_2^{k_2} \cdots c_d^{k_d}$ is uniquely determined by its block structure, we can take $\theta$ to be the sequence of pairs $(c_1, k_1), \ldots, (c_d, k_d)$.
\end{itemize}
Since $r = |c_A|$ is a deterministic function of $c_A$, the Shannon entropy of the bridge distribution $p_{c_A}$ obeys the standard chain rule $H(c_A) = H(|c_A|) + H(c_A \mid |c_A|)$:
\begin{equation}
\label{eq:S-decomp}
S(L_A) = H_{\rm radial} + H_{\rm angular},
\end{equation}
where the radial and angular pieces are
\begin{align}
H_{\rm radial} &\equiv H(|c_A|) = -\sum_{r \geq 0} P_r \log P_r,\\
H_{\rm angular} &\equiv \sum_{r \geq 0} P_r\, H\bigl(c_A \mid |c_A| = r\bigr),
\end{align}
with $P_r \equiv \Pr(|c_A| = r)$ the marginal of the radial distance.

We first discuss the radial piece. Numerically, the bridge-walk depth $|X_t|$ at $L=24$ has midpoint mean $\langle |X_{L/2}|\rangle \approx 0.74 \sqrt{L}$ with fluctuations of the same order as shown in Fig.~\ref{fig:bridge-walks}. Intuitively, the radial process $|X_t| \equiv |f(w_{[1:t]})|$ is a one-dimensional random walk on $\mathbb{Z}_{\geq 0}$ with bounded increments, i.e. each character appended changes $|X_t|$ by $+1$ or $-2$, and the strings in our sector correspond to paths conditioned on $X_0 = 0$ and $X_L = |c_f|$. This is precisely a one-dimensional bridge walk, whose mid-walk profile follows the classical Brownian-bridge envelope. The rigorous version of this statement \cite{lalley1993finite,bjorklund2010central,gouezel2014local} yields
\begin{equation}
\label{eq:bridge-CLT}
\langle |c_A| \rangle \sim \sigma \sqrt{\tfrac{2}{\pi}\,\tfrac{L_A(L-L_A)}{L}}.
\end{equation}
An independent rigorous derivation of this mean via saddle-point analysis on the closed-form $D_k(\ell)$ in Eq.~\eqref{eq:Dk-closed} is given in Appendix~\ref{eq:subsec-efs-entanglement}. Hence $P_r$ is supported on only $O(\sqrt{L})$ consecutive values of $r$, and so $H_{\rm radial} \leq \log(\#\text{support}) = O(\log L)$.

We now discuss the angular piece. Conditional on $|c_A| = r$, the bridge measure on the Cayley graph is rotationally homogeneous, so all $N_r$ elements at depth $r$ are equally weighted. The count $N_r$ is the number of length-$r$ N3C strings, which from the main text Eq.~\eqref{eq:dim-frozen} equals $2F_{r+1}$ and grows exponentially as $\phi^r$ with $\phi = (1+\sqrt{5})/2$. Therefore
\begin{equation}
H(c_A \mid |c_A| = r) \approx \log N_r \sim r \log\phi,
\qquad
H_{\rm angular} \approx \langle |c_A| \rangle \log\phi.
\end{equation}

Substituting into Eq.~\eqref{eq:S-decomp},
\begin{equation}
\label{eq:S-vs-cA}
S(L_A) \approx \langle |c_A| \rangle \log\phi + O(\log L)
\end{equation}
to leading order in $L$. The bipartite entanglement entropy is, up to a logarithmic correction, proportional to the average reduced-word length $\langle |c_A| \rangle$. Combining with Eq.~\eqref{eq:bridge-CLT} gives the announced $\sqrt{L}$ scaling.

\begin{figure}[h]
\centering
\includegraphics[width=0.65\linewidth]{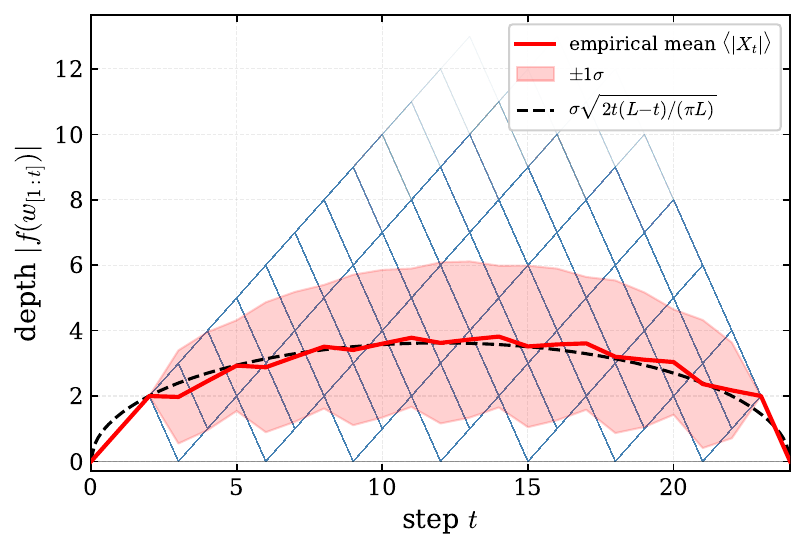}
\caption{Numerical verification of Eq.~\eqref{eq:bridge-CLT}. Light blue curves: $5000$ uniform bridge walks in the all-mobile sector at $L = 24$, drawn from the $199316$ length-$24$ strings reducing to $e$. Red curve: empirical mean depth $\langle |X_t| \rangle$ over the $5000$ samples. Pink band: $\pm 1\sigma$ envelope. Dashed black: Brownian-bridge profile $\sigma\sqrt{2t(L-t)/(\pi L)}$ with the empirically fitted $\sigma \approx 1.85$.}
\label{fig:bridge-walks}
\end{figure}

\subsection{Extensions}
\label{subsec:efs-general}

For the asymmetric projector with coupling ratio $\gamma = a/b$, the EFS ansatz generalizes by weighting each string by $\gamma^{N_1(w)/3}$:
\begin{equation}
|\mathrm{EFS}^\gamma\rangle \propto \sum_w (-\gamma^{1/3})^{N_1(w)} |w\rangle.
\end{equation}
The Schmidt decomposition takes the same form as Eq.~\eqref{eq:EFS-Schmidt} with the signs replaced by $(-\gamma^{1/3})^{N_1}$. Defining the weighted sector size
\begin{equation}
D_c^\gamma(\ell) = \sum_{\substack{|w|=\ell\\ f(w) = c}} \gamma^{(2/3) N_1(w)},
\end{equation}
the Schmidt weights become $p_c = D_c^\gamma(L_A) D_{c^{-1} c_f}^\gamma(L_B) / D_{c_f}^\gamma(L)$. The Cayley-graph bridge-walk structure is unchanged with only the step weights on the walk are modified. So the $\sqrt{L}$ scaling persists.

\section{Classical Krylov structure of the triplet flip model}
\label{app: proof of dimension}

\subsection{Krylov structure}

The triplet-flip model is a generalization of the pair-flip model. It is generated by local Hamiltonian terms of the form $|aaa\rangle \langle bbb|$. Note that in the literature, the pair-flip model is usually accompanied by on-site potential terms $n_i^a$. The presence of such terms lifts the degeneracy of the computational basis, thus kills any quantum fragmentation. However, as far as classical fragmentation is concerned, the presence of the on-site potentials are irrelevant. We will consider the model without on-site potential terms.

We consider the triplet-flip model on a chain with local Hilbert space dimension $q$, corresponding to the semigroup $\tilde{G}_1 = \langle 0,1,\ldots,q{-}1 \,|\, a^3 = b^3 \text{ for all digits } a,b\rangle$, which at $q=2$ reduces to $\langle 0,1\,|\,000=111\rangle$ studied in the main text.

Let $X=a^3$ for any digit $a$; this is well-defined since all $a^3$ are equivalent. It is easy to show that $X$ centralizes the semigroup, as
\begin{equation}
Xa = a^4 = aX
\end{equation}
for any digit $a$. For any word, we can recursively extract occurrences of three consecutive identical digits and replace them with $X$ until this cannot be done further. This proves that any word can be written as $X^kw$, where $w$ is a word with no three consecutive characters (N3C). It can be proven (see Theorem~\ref{thm:triplet-flip-uniqueness}) that this decomposition is unique. We call $k$ the number of mobile triplets, and $w$ the frozen string. Each Krylov subspace thus corresponds to a $(k,w)$ pair. It is obvious that Krylov subspaces with the same $k$ and total length $L$ have the same structure. For a given system size $L$, we call $D_k(L)$ the dimension of the Krylov subspaces with $k$ mobile triplets, and $d_k(L)$ the number of such subspaces. Obviously, such quantities are defined only when $3k \leq L$.

\subsection{Krylov degeneracy} \label{subsec:krylov-degeneracy-triplet-flip}

Let the number of N3C strings of length $L$ be $d_L$. We can find a recurrence relation for $d_L$ by breaking down $d_L = d^{xx}_L+d^{xy}_L$, where $d_L^{xx}$ means N3C strings that end with two identical characters, and $d_L^{xy}$ ends with distinct characters (assume that $L\geq 2$). Then, we get two recurrence relations: (1) $d_{L+1}^{xx}=d_L^{xy}$, since we cannot append another $x$ to a string that already ends with $xx$, hence a $L+1$ string that ends with two identical characters can only be obtained by appending an $y$ to a $L$ string ending with $xy$; (2) $d_{L+1}^{xy} = (q-1)d_L$, since we can append any character that is distinct with the last character of the $L$ string. Combining these two, we have
\begin{equation}
d_{L+2} = (q-1)\left[d_{L+1}+d_L\right].
\end{equation}
With the initial conditions $d_1=q$ and $d_2=q^2$, this solves to
\begin{equation}
d_L = \frac{q}{2\sqrt{q-1}} \left[\frac{1}{\sqrt{q-1}}\left(\lambda_+^L+\lambda_-^L\right)+\frac{1}{\sqrt{q+3}}\left(\lambda_+^L - \lambda_-^L\right)\right],
\end{equation}
where
\begin{equation}
\lambda_\pm = \frac{q-1\pm\sqrt{(q+3)(q-1)}}{2}.
\end{equation}
For $q=2$, this gives $\{2,4,6,10,16,26,\dots\}$; for $q=3$, this gives $\{3,9,24,66,180,492,\dots\}$. At large $L$, this scales as
\begin{equation}
d_L \sim \left(\frac{q-1+\sqrt{(q+3)(q-1)}}{2}\right)^{L}.
\end{equation}
At $q=2$, the exponent is $\phi = \frac{1+\sqrt{5}}{2}$; at $q=3$, the exponent is $1+\sqrt{3}$.

\subsection{Krylov dimension} \label{sec:D-cl-tripflip}

The Krylov dimensions $D_k(L)$ should satisfy the recurrence relation
\begin{equation}
D_k(L+1) = D_k(L)+(q-1)D_{k-1}(L). \label{eq:recurrence-relation}
\end{equation}
To see this, consider a string $s$ of length $L+1$ that reduces to the normal form $X^k w$, where $w$ is non-empty. Let $w=av$, where $a$ is its first character, and $v$ is a frozen string of length $L-3k$. Similarly, let $s=bt$, where $b$ is the first character, and $t$ is of length $L$. Then,
\begin{enumerate}
    \item If $a=b$, then $t$ reduces to normal form $X^k v$;
    \item If $a\neq b$, then $t$ reduces to normal form $X^{k-1}bbw$.
\end{enumerate}
This can be proven by assuming that $t$ has normal form $X^{m} u$, then matching $bt=bX^m u=X^m bu$ to $X^k w$. Eq.~\eqref{eq:recurrence-relation} is then a direct consequence of this: the string $s$ in our subspace $X^k w$ is either $a$ plus a string in the subspace $X^k v$, or any $b\neq a$ (there are $q-1$ choices of $b$) plus a string in the subspace $X^{k-1}bbw$.

Based on Eq.~\eqref{eq:recurrence-relation}, all $D_k(L)$ can be determined from initial conditions, which are $D_k(3k)$. This corresponds to Krylov subspaces where the frozen string is empty, or words that can be constructed by starting from the empty string and recursively inserting three consecutive identical characters. We call such strings all-mobile strings. The formula is
\begin{equation}
D_k(3k) = \frac{q}{k} \sum_{j=0}^{k-1}(q-1)^j(k-j) \begin{pmatrix}3k \\ j\end{pmatrix}. \label{eq:num-of-all-mobile-strings}
\end{equation}
This number is documented in the online encyclopedia of integer sequences (OEIS) as series A213028~\cite{oeisA213028}. For $q=2$ specifically, the above expression coincides with the series A047098 ($a(n)$ in the notation on the website) in OEIS \cite{oeisA047098}. Interestingly, this formula also arises from the counting of the number of words generated by $(aba)^n$ under the 3-strand braid semi-group dynamics $B_3^+=\langle a,b|aba=bab\rangle$ \cite{Cornwell2008COUNTINGFP}. Compared with our semi-group dynamics, we notice that there exists an bijection between the elements in $B_3^+$ and $\tilde G_1$. Define $\Phi:\{0,1\}^L \rightarrow \{a,b\}^L.$:
\begin{itemize}
    \item if $j$ is odd: $0 \mapsto a, 1 \mapsto b$ .
    \item if $j$ is even: $0 \mapsto b, 1 \mapsto a$ .
\end{itemize}
Under this map, $0^{3n}$ maps to $(aba)^n$. The equivalence relation $000=111$ maps to $aba=bab$. Therefore, Eq.~\eqref{eq:num-of-all-mobile-strings} (at $q=2$) gives the dimension of the Krylov generated by $|000\cdots 0\rangle$. For a proof of the general case of this formula, see Theorem~\ref{thm:all-mobile-string-number}.

Combining Eq.~\eqref{eq:recurrence-relation} and Eq.~\eqref{eq:num-of-all-mobile-strings}, we can arrive at the following expression,
\begin{equation}
D_k(L) = (q-1)^k \begin{pmatrix}L \\ k\end{pmatrix} - (2q-3)\sum_{j=0}^{k-1}(q-1)^j\begin{pmatrix}L \\ j\end{pmatrix}. \label{eq:general-DkL-triplet-flip}
\end{equation}
Taking $q=2$, this reduces to the formula for $D_k(L)$ utilized in the main text. For the first few $k$, we have
\begin{align}
D_0(L) & = 1, \\
D_1(L) & = (q-1)L-(2q-3), \\
D_2(L) & = \frac{(q-1)^2}{2}L^2 -\frac{(q-1)(5q-7)}{2}L - (2q-3).
\end{align}

\section{Classical Krylov structure of the cyclic qutrit model}

The cyclic qutrit model is defined by the semigroup $\tilde{G}_3 = \langle 0,1,2 \,|\, 012=120=201,\; 021=102=210\rangle$ studied in the main text. The first equivalence class consists of even permutations of $(0,1,2)$, and the second of odd permutations. Throughout this section, we write $a,b,c$ for three distinct digits taken in cyclic order, i.e., $(a,b,c)$ is any cyclic permutation of $(0,1,2)$; in this notation the two equivalence classes read $abc=012=120=201$, and $acb=021=102=210$. We describe a complete classification of the Krylov subspaces of this model on a chain of length $L$.

\subsection{Frozen states}

We begin the discussion with frozen states, or one-dimensional Krylov subspaces. A frozen state would be a string such that no three consecutive characters form a permutation of $\{0,1,2\}$. The number of frozen states can be obtained by a recurrence relation, similar to that in the triplet flip model. Consider $d_L = d_L^0+d_L^++d_L^-$, with strings broken down into three classes based on their last two digits: $0$ means the two digits are the same, $+$ means the last two digits are in $\{01,12,20\}$, and $-$ means $\{10,21,02\}$. With an argument essentially identical to that in Section.~\ref{subsec:krylov-degeneracy-triplet-flip}, we get
\begin{align}
d_{L+1}^0 & = d_L; \\
d_{L+1}^+ & = d_L^0 + d_L^-; \\
d_{L+1}^- & = d_L^0 + d_L^+. \\
\end{align}
In total, this produces
\begin{equation}
d_{L+1} = 2d_{L}+d_{L-1}.
\end{equation}
With $d_1=3$ and $d_2=9$, the solution is
\begin{equation}
d_L = \frac{3}{2}\left[(\sqrt 2+1)^L + (1-\sqrt 2)^L\right].
\end{equation}
The sequence is $\{3,9,21,51,123,297,\dots\}$.

\subsection{Conserved quantities and symmetries}
\label{sec:cyclic-conserved-quantities}

Before discussing the detailed Krylov space structure of the model, we identify some conserved quantities in the system.

Let us denote $n_i^a$ as the number of $a$ symbol on the $i$-th site; that is, $n_i^a=1$ if the $i$-th site is an $a$, and $n_i^a=0$ otherwise. Then, $n_i^0+n_i^1+n_i^2=1$. We can see that the dynamics conserves
\begin{equation}
N^a = \sum_i n_i^a,
\end{equation}
for $a=0,1,2$. Moreover, it conserves
\begin{equation}
N_\text{ord} = \sum_{i<j} (n_i^0 n_j^1+n_i^1 n_j^2 + n_i^2 n_j^0).
\end{equation}
It is easy to verify that as strings of three symbols, $abc$ triplets have $N_\text{ord}=2$, and $acb$ triplets all have $N_\text{ord}=1$. It is convenient to define
\begin{equation}
D = (N^0N^1+N^1N^2+N^2N^0)-2N_\text{ord} = \sum_{i<j} \left(n_i^1 n_j^0 - n_i^0 n_j^1 + n_i^2 n_j^1 - n_i^1 n_j^2 + n_i^0 n_j^2 - n_i^2 n_j^0\right).
\end{equation}
Then, $abc$ triplets would have $D=-1$, and $acb$ triplets have $D=+1$. $D$ also behaves nicely under concatenation: consider $w=uv$, where $u$ and $v$ are two substrings, then
\begin{equation}
D_w = D_u + D_v + \left(N^1_uN^0_v - N_u^0N_v^1 + N^2_uN^1_v - N_u^1N_v^2 + N^0_uN^2_v - N_u^2N_v^0\right). \label{eq:D-concatenation}
\end{equation}
Specifically, if $N^0_u=N_u^1=N_u^2$, then $D_w=D_u+D_v$.

The rewriting rules of the system exhibit a $D_3$ symmetry, which is the semidirect product of $\mathbb Z_3: 0\to 1 \to 2 \to 0$ and $\mathbb Z_2 :0 \leftrightarrow 1$. The $\mathbb Z_3$ symmetry permutes the numbers $(N^0,N^1,N^2)$ while leaving $D$ invariant. The transposition swaps $N^0$ and $N^1$ and renders $D\to -D$. This symmetry implies that if $w$ generates a Krylov subspace, then $w^\prime$ obtained by acting a symmetry operation on $w$ should generate a subspace degenerate to (if not the same as) that of $w$.

\subsection{Triplets and quasi-Krylovs}

We call $X=abc$ an $X$ (even) triplet, and $Y=acb$ a $Y$ (odd) triplet; explicitly, the $X$ triplets are $012,120,201$ and the $Y$ triplets are $021,102,210$. It can be easily shown that $X$ and $Y$ centralize the semigroup, i.e., they commute with any string. For example, $Xa=(abc)a=a(bca)=aX$. Notably, if we only apply one pair of the update rules, for example, consider the semigroup $\langle 0,1,2\,|\, 012=120=201\rangle$ with only the even class, then the system would be isomorphic to the $q=3$ triplet flip model, via the mapping $s_i^\prime = s_i-i \mod 3$. Therefore, the way $X$ triplets behave in the string is exactly the same as the triplets in the $q=3$ triplet flip model if the update rules associated with $Y$ triplets are not invoked. This also shows that the cyclic qutrit model is strictly less fragmented than the triplet flip model.

Due to this structure, we can first exclusively invoke the $X$ update rules to reduce a string to $w=X^ku$, where $u$ is a frozen string with respect to the $X$ update rules. Then, we can apply the $Y$ update rules to $u$, such that $u=Y^{k^\prime} v$. Doing this iteratively, we can reduce every string to $w = X^{m_X} Y^{m_Y} w^\prime$, where $w^\prime$ is a frozen string that does not contain any $X$ or $Y$ triplet patterns. We call this a normal form of the strings in this model. The set of all the strings with a given normal form $(m_X,m_Y,w)$, or equivalently, all the strings that can be obtained by inserting $m_X$ $X$ triplets and $m_Y$ $Y$ triplets into $w^\prime$, is referred to a quasi-Krylov of this model.

The term ``quasi-Krylov'' is used because they are often not actual Krylov subspaces due to the fact that one string can have more than one normal forms. For example, the string $abcba$ can be represented by two different normal forms, $Xba$ and $Yab$. More generally, the following equivalence relations exist between normal forms:
\begin{align}
Yab & =Xba, \label{eq:fusion-1} \\
YYa & =Xbaac,  \label{eq:fusion-2}  \\
YYY & = Xbaaccb .  \label{eq:fusion-3}
\end{align}
Therefore, a Krylov subspace would in general be a union of several quasi-Krylovs connected by the rewrite rules above.

\subsection{Single-triplet sectors}

We separate non-frozen Krylov subspaces into two broad classes, called the single-triplet sector and multi-triplet sector. The single-triplet sector are defined as Krylov subspaces that consist of one or several quasi-Krylovs such that each of them have $m_X+m_Y=1$. We will soon see that single-triplet Krylov subspaces are qualitatively different from multi-triplet ones. In this subsection, we discuss the single-triplet sector.

\paragraph{Single quasi-Krylov. --- } The simplest single-triplet Krylov is a quasi-Krylov that is also a Krylov subspace, i.e., that does not tunnel into other quasi-Krylovs. Without loss of generality, assume the quasi-Krylov is characterized by $(m_X=1,m_Y=0,w)$. Then, we require that none of the rules Eq.~\eqref{eq:fusion-1}, \eqref{eq:fusion-2}, or \eqref{eq:fusion-3} can apply to the string $Xw$. This happens when $w$ does not contain any $ba$ patterns, which we will call $Y$-activation patterns. Therefore, such a Krylov would be entirely consists of strings where $X$ is inserted at different positions in $w$, similar to the triplet flip model. The dimension of such Krylov subspaces would be equal to the $D_k(L)$ in the triplet flip model with $k=1$ and $q=3$, which is $2L-3$.

The Krylov degeneracy is simply equal to the number of frozen strings that contain no $Y$-activation patterns. In such strings, any two consecutive characters are either the same, or an $X$-activation pattern, $ab$. A simple recurrence relation can be derived: if a string ends with $aa$, then we can append either $a$ or $b$; if a string ends $ab$, then we can only append $b$. This shows that the number of frozen strings with no $Y$-activation patterns satisfies a Fibonacci recurrence relation. Specifically, the number of length-$(L-3)$ frozen strings with no $Y$-activation patterns is $3F_{L-2}$, where $F$ is the Fibonacci sequence. Further considering that there are symmetric sectors $Yw$, where $w$ contains no $X$-activation patterns, we get the total degeneracy of this kind of subspace is $6F_{L-2}$.

Generally, a single-triplet Krylov would contain more than one quasi-Krylovs. In this case, the rule Eq.~\eqref{eq:fusion-1} could apply, but not Eq.~\eqref{eq:fusion-2} or Eq.~\eqref{eq:fusion-3}, since that contradicts our definition of single-triplet. Therefore, a single-triplet Krylov is a set of $Xw$ and $Yw^\prime$ connected via application of Eq.~\eqref{eq:fusion-1}, subject to the constraint that it cannot connect to a Krylov that creates a second triplet within the frozen string. We propose two families of $w$ that are guaranteed to remain single-triplet under this dynamics.

\paragraph{First family. --- } The first family consists of $Xw$ or $Yw$ where $w$ contains only two kinds of characters, $\{a,b\}$, which implies that $N^c=1$. In this case, Eq.~\eqref{eq:fusion-2} or Eq.~\eqref{eq:fusion-3} are automatically prohibited, since their right-hand sides require $N^c \geq 2$.  Therefore, the only process that can happen is $Xba\leftrightarrow Yab$ and its equivalents. It can be proven that in such case, all the strings sharing the same invariants $(N^0,N^1,N^2,D)$ lie in the same Krylov subspace (see Theorem~\ref{thm:single-triplet-family-1-connectivity}). Therefore, given a set of invariants $(N^0, N^1, N^2, D)$, if $\mathrm{min}(N^0,N^1,N^2)=1$, then all the non-frozen strings with the same invariants form a Krylov subspace.

\paragraph{Second family. --- } The second family is a generalization of the single quasi-Krylov subspace. We introduce it with an example. Consider the Krylov subspace generated by the string $Xb^2ab^2 c$. Repeatedly applying rule Eq.~\eqref{eq:fusion-1}, this can be transformed into $\{Ybab^3 c, Xbab^2 cb,Yab^3cb,Xab^2 cb^2\}$. The union of these five quasi-Krylovs forms a Krylov subspace. To see the general construction, notice that in all the strings above, the suffix can be obtained by acting transpositions that replace an $X$-activation pattern by an $Y$-activation pattern on $ab^4 c$. For example, $b^2 ab^2 c$ can be obtained by $ab^4 c \to bab^3 c \to b^2 ab^2 c$. We call a string moment-$\mu$ variant of $ab^4 c$ if it can be obtained from acting $\mu$ such transpositions on $ab^4 c$. Observing that $D$ is conserved throughout the process, and $D(Xw)=D(w)-1$, $D(Yw)=D(w)+1$, and that if $w$ is a moment-$\mu$ string, then $D(w)=D(ab^4 c)-2\mu$; we see that for all strings of the form $Xw$ in a Krylov, the $w$ part must have the same $\mu$. In fact, we see that the three suffixes of $X$ in this Krylov, $\{b^2ab^2 c, bab^2 cb, ab^2 cb^2\}$, are exactly the three strings with moment $\mu=2$, and the two suffixes of $Y$ are exactly the two strings with moment $\mu=1$.

Based on the same string $ab^4c$, we can also construct a Krylov containing $X$ plus moment-1 strings and $Y$ plus moment-0 strings, which is $\{Xbab^3 c,Yab^4c,Xab^3cb\}$. We can also construct $X$ plus moment-0 strings, in which case there is no value $\mu$ value for $Y$, hence the Krylov subspace consists of a single quasi-Krylov $Xab^4 c$. However, we cannot further increase $\mu$ from $\mu=2$, since $b^3 abc$ would be a valid moment-3 string, yet it contains an $abc=X$ pattern, such that $Xb^3abc=X^2b^3$, violating the assumption that the Krylov subspace is single-triplet.

More generally, take any suffix that does not contain $Y$-activation patterns. Without loss of generality, assume that it begins with $a$, then it should have the form $w_0=a^{m_1} b^{m_2} c^{m_3} a^{m_4} \dots x ^{m_{K}}$, where the symbols appear in the sequence $abcabc\dots$, hence $x$ is a symbol $a,b,c$ depending on $K \mod 3$, and the exponents are all positive integers. Then the union of quasi-Krylovs $Xw$ where $w$ are moment-$\mu$ variants of $w_0$ and $Yw$ where $w$ are moment-$(\mu-1)$ variants of $w_0$ would form a Krylov subspace. To avoid any triplet patterns, we require that $m_2,\dots,m_{K-1}\geq \mu+2$.

We can show that the number of second-family strings is exponential in system length. It suffices to show that the number of base strings $a^{m_1} b^{m_2} c^{m_3} a^{m_4} \dots x ^{m_{K}}$ is exponential. If $\mu=0$, the base strings are exactly the frozen strings with no $Y$-activation patterns. We have shown that it satisfies the recurrence relation
\begin{equation}
d_{\mu=0}(L+2) = d_{\mu=0}(L+1) + d_{\mu=0}(L).
\end{equation}
Generally, for non-zero $\mu$, we can show that
\begin{equation}
d_{\mu}(L+\mu+2) = d_{\mu}(L+\mu+1) + d_{\mu}(L). \label{eq:general-mu-recurrence}
\end{equation}
In fact, consider a string of length $L+\mu+1$. We can always append a character that is the same as the last character of the current string to form a valid string of length $L+\mu+2$. If we want to append a different character, however, the constraint $m_i \geq \mu+2$ requires that the last $\mu+2$ characters of the current string all be the same. This leads to Eq.~\eqref{eq:general-mu-recurrence}. This implies that $d_{\mu}(L) \sim (\lambda_\mu)^L$, with
\begin{equation}
(\lambda_\mu)^{\mu+2} = (\lambda_\mu)^{\mu+1}+1.
\end{equation}
Therefore, the number of Krylov subspaces with any given $\mu$ scales exponentially with system size. For the first case distinct from the single quasi-Krylov subspaces, $\mu=1$, the exponent is the supergolden ratio $\psi \approx 1.47$.

\subsection{Multi-triplet sectors}

The case where strings can contain more than one triplet appears to be more complex. However, we find the exact opposite: Krylov subspaces in the multi-triplet sectors are fully characterized by the invariants we identified in Section.~\ref{sec:cyclic-conserved-quantities}. Formally, all the strings with the same $(N^0,N^1, N^2, D)$ that are not frozen or belong to a single-triplet subspace would lie in the same Krylov subspace. We have confirmed this numerically for all system sizes up to $L=21$; an argument based on conserved quantities and a reduction to canonical form is given in Sec.~\ref{sec:multi-triplet-argument}. This implies that the size of multi-triplet subspaces can scale exponentially, while the number of them scales at most polynomially with respect to $L$.

\subsection{Summary}\label{app:cyclic_table}

We summarize the full Krylov subspace structure of the classical cyclic qutrit model in Table~\ref{tab:cyclic-qutrit-all-krylovs}.

\begin{table}[!h]
    \centering
    \begin{tabular}{|c|c|c|c|}\hline
         Type of Krylov&  Example & Number & Symmetry\\ \hline
         Frozen& 
     $aabbcc$&$\sim(\sqrt 2+1)^L$ &None\\\hline
 Single-triplet, F1& $abcabab$&$\mathrm{poly}(L)$ &$\mathbb Z_2$ if $D=0$ and $N^a=N^b$\\\hline
 Quasi-Krylov (F2, $\mu=0$) & $abcaabbcc$&$6F_{L-2} \sim \phi^L$ &None\\\hline
 Single-triplet, F2& $abcbabbc$&$\sim (\lambda_\mu)^L$ & None\\\hline
 Multi-triplet& $abcacb$ &$\mathrm{poly}(L)$ & \makecell{$\mathbb Z_3$ if $N^0=N^1=N^2$;  $\mathbb Z_2$ if $D=0$ and $N^a=N^b$;  $D_3$ if both}\\ \hline\end{tabular}
    \caption{All kinds of Krylov subspaces for the classical cyclic qutrit model. $\phi=\frac{\sqrt 5+1}{2}$ is the golden ratio.}
    \label{tab:cyclic-qutrit-all-krylovs}
\end{table}

As a specific example, we show all the Krylov sectors of this model at $L=9$ in Table.~\ref{tab:L9}.

\begin{table*}[!h]
\caption{Classical and quantum Krylov subspaces for the cyclic qutrit model at $L=9$.
$N_{\mathrm{frozen}} = 4179$.
The first row lists product-frozen sectors.
Each non-frozen, non-$\mathbb{Z}_3$-stable sector has two isospectral partners
generated by the digit shift $\hat{X}^{\otimes L}$.
$\mathbb{Z}_3$-stable sectors ($\star$) have $\hat{X}^{\otimes L}$ acting within $\Kq$,
splitting it by charge.}
\label{tab:L9}
\begin{tabular}{llcccccc}
\toprule
root states of classical Krylov subspaces
   &Type& $D_{\mathrm{cl}}$ & \# EFS & $D_q$
  & generator & Deg. & \# sectors \\
\midrule
$\ket{{\rm Frozen}_{\rm prod}}$
   &Frozen& 1 & N/A & 1 & N/A & 1 & 4179 \\[3pt]
\texttt{000000012}, \texttt{000000021}, \ldots, \texttt{002122222}
   &Quasi-Krylov& 15 & 8 & 7 & $\hat{X}^{\otimes L}$ & 3 & 26 \\[2pt]
\texttt{000000122}, \texttt{000000211}, \ldots, \texttt{001222222}
   &Single-triplet& 28 & 14 & 14 & $\hat{X}^{\otimes L}$ & 3 & 16 \\[2pt]
\texttt{000012110}, \texttt{000012202}, \ldots, \texttt{001222221}
   &Single-triplet& 41 & 20 & 21 & $\hat{X}^{\otimes L}$ & 3 & 12 \\[2pt]
\texttt{000012220}, \texttt{000021110}, \texttt{001211000}, \texttt{001211101}
   &Single-triplet& 54 & 26 & 28 & $\hat{X}^{\otimes L}$ & 3 & 4 \\[2pt]
\texttt{000121100}, \texttt{000121110}, \texttt{000122001}, \texttt{000122020},
   &\multirow{2}{*}{Single-triplet}& \multirow{2}{*}{67} & \multirow{2}{*}{32} & \multirow{2}{*}{35}
  & \multirow{2}{*}{$\hat{X}^{\otimes L}$} & \multirow{2}{*}{3} & \multirow{2}{*}{8} \\
\texttt{000122200}, \texttt{000122202}, \texttt{000211002}, \texttt{000211100}
   && & & & & & \\[2pt]
\texttt{001211010}
   &Single-triplet& 78 & 36 & 42 & $\hat{X}^{\otimes L}$ & 3 & 1 \\[2pt]
\texttt{000122220}, \texttt{000211110}
   &Single-triplet& 80 & 38 & 42 & $\hat{X}^{\otimes L}$ & 3 & 2 \\[4pt]
\texttt{000001212}, \texttt{000002121}
   &Multi-triplet& 123 & 49 & 74 & $\hat{X}^{\otimes L}$ & 3 & 2 \\[2pt]
\texttt{000012121}, \texttt{000012221}, \texttt{000121211}, \texttt{000122212}
   &Multi-triplet& 136 & 55 & 81 & $\hat{X}^{\otimes L}$ & 3 & 4 \\[2pt]
\texttt{000001221}
   &Multi-triplet& 156 & 50 & 106 & $\hat{X}^{\otimes L}$ & 3 & 1 \\[2pt]
\texttt{000012211}, \texttt{000012212}, \texttt{000121220}, \texttt{000122122}
   &Multi-triplet& 226 & 74 & 152 & $\hat{X}^{\otimes L}$ & 3 & 4 \\[2pt]
\texttt{000012122}, \texttt{000021211}
   &Multi-triplet& 271 & 87 & 184 & $\hat{X}^{\otimes L}$ & 3 & 2 \\
\midrule
$\star$\; \texttt{000121212}, \texttt{000122211}
   &Multi-triplet& 252 & 92 & $54{+}53{+}53$
  & N/A & 1 & 2 \\[2pt]
$\star$\; \texttt{000121221}, \texttt{000122121}
   &Multi-triplet& 432 & 120 & $106{+}103{+}103$
  & N/A & 1 & 2 \\
\bottomrule
\end{tabular}
\end{table*}

\section{Mathematical proofs}

This section compiles the mathematical proofs related to the structure of the Krylov subspaces in the triplet flip model and the cyclic qutrit model.

\subsection{Triplet flip model: Uniqueness of decomposition} 
\label{subsec:math-triplet-flip}

\begin{theorem}
In the triplet flip model, every string can be uniquely written as $X^kw$, where $w$ is a frozen string. \label{thm:triplet-flip-uniqueness}
\end{theorem}
To prove this, we introduce the \textit{critical pair theorem}~\cite{huet1980confluent}. The theorem can be summarized as follows:
\begin{theorem}
Consider a string rewriting system $(\Sigma,R)$, where $\Sigma$ is a finite alphabet and $R$ is a finite set of local rewriting rules of the form $\ell \to r$, where $\ell$ and $r$ are words over $\Sigma$. A normal form is defined as a word that cannot be further rewritten under $R$. Any word over $\Sigma$ can be reduced to a unique normal form under $R$ if the following conditions are satisfied:
\begin{itemize}
    \item \textbf{Termination:} Applying the rules of $R$ starting from any word, a normal form must be reached in a finite number of steps.
    \item \textbf{Joinability of critical pairs:} Consider two (not necessarily distinct) rewriting rules $\ell_1\to r_1$ and $\ell_2\to r_2$. An overlap critical pair occurs when $\ell_1=ab, \ell _2=bc$, with $b$ non-empty; the critical pair is $(r_1c,ar_2)$, two strings reduced from the string $abc$. An inclusion critical pair occurs when $\ell_2=b$, $\ell_1=abc$, and the pair is $(r_1,ar_2c)$. We require all critical pairs to reduce to a common normal form (i.e. joinable).
\end{itemize}
\end{theorem}
\begin{proof}[Proof of Theorem~\ref{thm:triplet-flip-uniqueness}]
We apply this theorem to the alphabet $\Sigma=\{0,1,\ldots,q{-}1\}\cup\{X\}$ and rewriting rules $R=\{a^3 \to X,\; aX\to Xa\}$ for every digit $a$. This rule set is obviously terminating. To show that critical pairs are joinable, we find that eligible critical pairs include $a^4 \to (Xa,aX) \to Xa$, $a^5 \to (Xa^2,a^2X)\to Xa^2$, and $a^3X \to (XX,a^2Xa) \to XX$, all of which are joinable. This proves the uniqueness of the normal form.
\end{proof}

\subsection{Triplet flip model: Size of all-mobile Krylov subspaces} 

We will offer a proof of Eq.~\eqref{eq:num-of-all-mobile-strings}, the formula for the number of all-mobile strings, by giving an explicit constructive characterization of all-mobile strings. We then verify that this implies Eq.~\eqref{eq:general-DkL-triplet-flip}.

Given a character $a$, we define \textit{$a$-encompassed strings} as all-mobile strings such that it does not contain any suffix that is an $a$-leading all-mobile string. That is, $w$ is $a$-encompassed if it cannot be written as $w=w_1w_2$, where both $w_1$ and $w_2$ are all-mobile and $w_2$ begins with $a$. In ``$w_1$ is all-mobile'', we allow $w_1$ to be empty. This convention carries on to this entire section.
\begin{theorem}
An all-mobile string starting with the character $a$ can be uniquely represented as $aw_1 aw_2\dots aw_{3n}$, where each $w_i$ is an $a$-encompassed string . \label{thm:all-mobile-string-number}
\end{theorem}
For example,
\begin{itemize}
    \item $w=aaa$ has $n=1$ and all $w_i$ empty.
    \item $w=abbbaccca$ has $w_1=bbb$ and $w_2=ccc$ being $a$-encompassed strings.
    \item $w=abbcccbaa$ has $w_1=bbcccb$ being an $a$-encompassed string.
\end{itemize}
To prove this, we will first establish a lemma.
\begin{lemma}
For any all-mobile string starting with $a$, there is a unique shortest prefix of the string that has the form $aw_1aw_2a$ with $w_1$ and $w_2$ being all-mobile strings. For this shortest prefix, the strings $w_1$ and $w_2$ will be $a$-encompassed strings.\label{thm:lemma-all-mobile}
\end{lemma}
\begin{proof}[Proof of Lemma.~\ref{thm:lemma-all-mobile}]
We first establish that there is a unique shortest prefix $aw_1aw_2a$ where $w_1$ and $w_2$ are all-mobile, then proceed to show that $w_1$ and $w_2$ are $a$-encompassed.

Since a finite string has a finite number of prefixes, the existence of a shortest prefix follows directly from the existence of at least one prefix of the given form. This existence is guaranteed since the string must be emptied by recursive removals of triplets, and the leading $a$ must be removed as part of an $aaa$ at some step; these three $a$'s would be the $a$'s in the pattern $aw_1aw_2a$, and $w_1$ and $w_2$ would just be the characters that were stacked between them at beginning, which must be all-mobile by construction. Hereby, we have established that a shortest prefix $aw_1aw_2a$ where $w_1$ and $w_2$ are all-mobile exists.

We further prove that the representation $a w_1 a w_2 a$ is unique for this prefix, since if $a w_1 a w_2 a = a u_1 a u_2 a$, assume $|w_1|<|u_1|$, then we must have $u_1=w_1 a v$, from which we can write $w=aw_1 avau_2 a$. Since $w_1$ and $u_2$ are all-mobile, we can remove them, and the remaining string $aavaa$ must still be all-mobile. This means that $v$ must have the normal form $X^maa$, which then follows that $v$ must be prefixed by $sa$, where $s$ is an all-mobile string. But then $w$ is prefixed by $aw_1 asa$, which is a shorter prefix of the desired form than $aw_1aw_2a$. This establishes the uniqueness of representation.

Finally, we show that $w_1$ cannot have a suffix that is a mobile string which begins with $a$. Suppose that is not true, then that suffix can be written as $a w_{11} aw_{12} aw_{13}$, where $w_{11}$, $w_{12}$, $w_{13}$ are all-mobile. Therefore, $w_1 = w_{10}a w_{11} aw_{12} aw_{13}$. But then $a w_{10} aw_{11} a$ would be a prefix of our desired pattern but of a shorter length. The same reasoning goes for $w_2$.
\end{proof}

\begin{proof}[Proof of Thm.~\ref{thm:all-mobile-string-number}]
A string that begins with $a$ must have a prefix $aw_1aw_2a$. Now consider the rest of the string after removing $aw_1 a w_2 a$. If it contains any suffix that is an all-mobile string that begins with $a$, take the longest of such suffix. The part not contained in this suffix would be an $a$-encompassed mobile string, which we call $w_3$. Hereby, we have decomposed our target string into $a w_1 a w_2 aw_3$, where $w_i$ are all $a$-encompassed strings, plus a shorter $a$-leading mobile string. 
\end{proof}

We then show a key property about encompassed strings.
\begin{theorem}
An $a$-encompassed string can be canonically represented as $bw_1bw_2bw_3$, where $b\neq a$, $w_1$ and $w_2$ are $b$-encompassed, and $w_3$ is $a$-encompassed. \label{thm:encompassed-string-ternary-tree}
\end{theorem}
\begin{proof}
It follows from the definition that an $a$-encompassed string cannot begin with $a$, therefore its first character $b\neq a$. Applying Lemma~\ref{thm:lemma-all-mobile}, it must have the form $bw_1 bw_2 bw_3$, where $w_1$ and $w_2$ are $b$-encompassed. That $w_3$ is $a$-encompassed follows from the fact that the whole string is $a$-encompassed, since any suffix of $w_3$ is an suffix of the whole string. Since the construction from Lemma~\ref{thm:lemma-all-mobile} is unique, the entire decomposition is canonical.
\end{proof}

Combining Thm.~\ref{thm:all-mobile-string-number} and Thm.~\ref{thm:encompassed-string-ternary-tree}, we obtain the following characterization of all-mobile strings.
\begin{theorem}
A bijection exists between all-mobile strings of length $3k$ over an alphabet of $q$ characters and colored ternary forests (A ternary tree is a tree where all internal (non-leaf) nodes have exactly three children. A ternary forest is an ordered collection of ternary trees) with $k$ internal nodes, where each internal node is colored with one of the $q$ characters, according to the following rules:
\begin{itemize}
  \item All roots nodes of the trees in the forests must have the same color;
  \item An internal node that is the first or second child of its parent must have different color from its parent;
  \item An internal node that is the third child of its parent must have different color from its ``effective parent'', defined recursively as follows:
  \begin{itemize}
    \item The effective parent of a root node is itself;
    \item The effective parent of an internal node that is the first or second child of its parent is its parent;
    \item The effective parent of an internal node that is the third child of its parent is the same as the effective parent of its parent.
  \end{itemize}
\end{itemize}
\label{thm:mobile-string-to-tree}
\end{theorem}
\begin{proof}
The coloring rules of the trees are exactly tailored to match the structures of all-mobile strings established in Thm.~\ref{thm:all-mobile-string-number} and Thm.~\ref{thm:encompassed-string-ternary-tree}. We define the mapping between forests and strings as follows:
\begin{itemize}
  \item Leaf nodes correspond to empty strings;
  \item A tree with the root node colored $a$ and three leaf nodes corresponds to the string $aaa$;
  \item A tree where the root node has color $a$, and the sub-trees on the three children corresponding to strings $w_1$, $w_2$, $w_3$, respectively, corresponds to the string $aw_1 aw_2 aw_3$;
  \item If the trees in a forest correspond to strings $w_1$, $w_2$, etc., then the forest corresponds to the string $w_1 w_2 w_3 \dots$.
\end{itemize}
By construction, the string corresponding to any such forest must be all-mobile. It suffices to show that each all-mobile string can map to a tree. To this end, we use Thm.~\ref{thm:all-mobile-string-number}, and decompose the all-mobile string into $aw_1 aw_2 aw_3 \dots a w_{3t}$. Then, we can construct $t$ trees, with root nodes all colored by $a$. Each $w_i$ can further be decomposed using its leading character. Notice how the coloring rules play their role here. The different color for first and second child rule is guaranteed since when we decompose $w=bu_1 bu_2 bu_3$, $u_1$ and $u_2$ cannot begin with $b$. The effective parent rule is satisfied since if $w$ if $c$-encompassed, then so is $u_3$. As we can iteratively use Thm.~\ref{thm:encompassed-string-ternary-tree}, we can always reduce the strings $w_i$ to a tree structure, and the argument above guarantees that the coloring rules are satisfied.
\end{proof}

Finally, we prove Eq.~\eqref{eq:num-of-all-mobile-strings}. The proof uses the language of \emph{generating functions}: given a sequence $\{a_n\}_{n\geq 0}$, we encode it as the formal power series $A(x)=\sum_n a_n x^n$ and write $[x^n]\,A(x)=a_n$ for its $n$-th coefficient. A recursion among the $a_n$ translates into an algebraic equation for $A(x)$, from which closed-form coefficients can be extracted. The key tool we will need is the following:

\begin{lemma}[Lagrange inversion]
\label{lem:lagrange}
If a formal power series $w(x)$ satisfies $w(x) = x\,\phi(w(x))$ with $\phi(0)\neq 0$, then for any analytic $f$ and $n\geq 1$,
\begin{equation}
[x^n]\,f(w(x)) = \frac{1}{n}\,[w^{n-1}]\,f'(w)\,\phi(w)^n\,.
\end{equation}
For a proof, see \eg\ Ref.~\cite{Stanley_Fomin_1999}.
\end{lemma}

\begin{proof}[Proof of Eq.~\eqref{eq:num-of-all-mobile-strings}]
By Theorem~\ref{thm:mobile-string-to-tree}, $D_k(3k)$ equals $q$ times the number of valid colored ternary forests with $k$ internal nodes and a fixed root color. We now count these forests.

Let $T_m$ denote the number of encompassed strings with $m$ internal nodes (equivalently, the number of valid colored sub-trees with $m$ internal nodes in the bijection of Theorem~\ref{thm:mobile-string-to-tree}), and define $T(x)=\sum_{m\geq 0} T_m\,x^m$. By convention $T_0=1$, corresponding to the empty string (a leaf). For $m\geq 1$, Theorem~\ref{thm:encompassed-string-ternary-tree} decomposes each encompassed string as $bw_1 bw_2 bw_3$, where $b$ can be any of $(q-1)$ colors distinct from the forbidden one, and $w_1,w_2,w_3$ are themselves encompassed strings with a total of $m-1$ internal nodes. Summing over all ways to distribute these nodes among the three sub-strings, the recursion reads
\begin{equation}
T_m = (q-1)\sum_{m_1+m_2+m_3=m-1}T_{m_1}\,T_{m_2}\,T_{m_3}\,,
\end{equation}
which translates into
\begin{equation}
T(x) = 1+(q-1)\,x\,T(x)^3. \label{eq:T-functional}
\end{equation}

Now we count forests. By Theorem~\ref{thm:all-mobile-string-number}, an all-mobile string starting with character $a$ decomposes as $aw_1\,aw_2\,aw_3\cdots aw_{3n}$, which groups into $n$ trees: $(aw_1\,aw_2\,aw_3)(aw_4\,aw_5\,aw_6)\cdots$. Let $F_k$ be the number of such forests with $k$ internal nodes and a fixed root color, so that $D_k(3k) = q\,F_k$. A single tree has $k_i\geq 1$ internal nodes: one from the root $a$ and $k_i-1$ distributed among its three encompassed sub-trees $w_1,w_2,w_3$. Summing over all such distributions, the number of single-tree configurations with $k_i$ internal nodes is $\sum_{m_1+m_2+m_3=k_i-1}T_{m_1}T_{m_2}T_{m_3} = [x^{k_i-1}]\,T(x)^3$, or equivalently $[x^{k_i}]\,x\,T(x)^3$. A forest of $n$ trees partitions $k$ nodes among the trees, so
\begin{equation}
F_k = \sum_{n\geq 0}\sum_{k_1+\cdots+k_n=k}\prod_{i=1}^{n}[x^{k_i}]\,xT^3 = [x^k]\,\frac{1}{1-xT(x)^3}\,.
\end{equation}
Therefore,
\begin{equation}
\sum_{k\geq 0}D_k(3k)\,x^k = \frac{q}{1-x\,T(x)^3}\,.
\end{equation}
Expanding the geometric series,
\begin{equation}
D_k(3k)=q\sum_{n=0}^{k}[x^{k-n}]\,T(x)^{3n}\,. \label{eq:Dk-geometric}
\end{equation}
To apply Lemma~\ref{lem:lagrange}, set $w(x)=T(x)-1$ so that Eq.~\eqref{eq:T-functional} becomes $w(x) = x\,(q-1)(1+w(x))^3$, which is in the standard form $w(x)=x\,\phi(w(x))$ with $\phi(w)=(q-1)(1+w)^3$. Taking $f(w)=(1+w)^{3n}=T^{3n}$, the lemma gives
\begin{equation}
[x^m]\,T(x)^{3n} = \frac{1}{m}\,[w^{m-1}]\,3n(1+w)^{3n-1}\cdot(q-1)^m(1+w)^{3m} = \frac{3n}{m}(q-1)^m\binom{3m+3n-1}{m-1}\,. \label{eq:lagrange-applied}
\end{equation}
Using the identity $\frac{1}{m}\binom{3m+3n-1}{m-1}=\frac{1}{3m+3n}\binom{3m+3n}{m}$, this simplifies to $[x^m]\,T^{3n}=\frac{n}{m+n}\binom{3m+3n}{m}(q-1)^m$.
Substituting into Eq.~\eqref{eq:Dk-geometric},
\begin{equation}
D_k(3k) = q\sum_{n=0}^{k}\frac{n}{k}\binom{3k}{k-n}(q-1)^{k-n} = \frac{q}{k}\sum_{j=0}^{k-1}(k-j)(q-1)^j\binom{3k}{j}\,,
\end{equation}
where the $j=k$ term vanishes since $k-j=0$.
\end{proof}

Finally, we verify that this implies the formula Eq.~\eqref{eq:general-DkL-triplet-flip} for general $D_k(L)$.

\begin{proof}[Proof of Eq.~\eqref{eq:general-DkL-triplet-flip}]
One can verify that Eq.~\eqref{eq:general-DkL-triplet-flip} satisfies $D_0(L)=1$. The result is to verify Eq.~\eqref{eq:recurrence-relation} and Eq.~\eqref{eq:num-of-all-mobile-strings}.

Let $\Delta_L D_k(L) = D_k(L+1)-D_k(L)$. Using the identity $\binom{L+1}{m} - \binom{L}{m} = \binom{L}{m-1}$, 
\begin{align}
    \Delta_L D_k(L) &= (q-1)^k \binom{L}{k-1} - (2q-3) \sum_{j=1}^{k-1} (q-1)^j \binom{L}{j-1} \nonumber \\
    &= (q-1) \left[ (q-1)^{k-1} \binom{L}{k-1} -(2q-3) \sum_{j=0}^{k-2} (q-1)^j \binom{L}{j} \right] \nonumber \\
    &= (q-1) D_{k-1}(L).
\end{align}
This exactly matches the required Eq.~\eqref{eq:recurrence-relation}.

Next, we verify that $D_k(L=3k)$ reduces to the form of  Eq.~\eqref{eq:num-of-all-mobile-strings}. Notice that $2q-3=2(q-1)-1$, we can write
\begin{align}
D_k(L=3k) & = \sum_{j=0}^{k}(q-1)^j\begin{pmatrix}3k \\ j\end{pmatrix} - 2\sum_{j=0}^{k-1}(q-1)^{j+1}\begin{pmatrix}3k \\ j\end{pmatrix}. \nonumber \\
& = 1 + \sum_{j=1}^k (q-1)^j \left[\begin{pmatrix}3k \\ j\end{pmatrix} - 2 \begin{pmatrix}3k \\ j-1\end{pmatrix}\right]
\end{align}
At the same time, Eq.~\eqref{eq:num-of-all-mobile-strings} can be written as
\begin{align}
D_k(3k) & = \frac{(q-1)+1}{k} \sum_{j=0}^{k-1}(q-1)^j(k-j) \begin{pmatrix}3k \\ j\end{pmatrix} \nonumber \\
& =1 + \sum_{j=1}^{k}(q-1)^j \left[\left(1-\frac{j}{k}\right) \begin{pmatrix}3k \\ j\end{pmatrix} + \left(1-\frac{j-1}{k}\right) \begin{pmatrix}3k \\ j-1\end{pmatrix}\right].
\end{align}
By direct comparison,
\begin{align}
& \left[\left(1-\frac{j}{k}\right) \begin{pmatrix}3k \\ j\end{pmatrix} + \left(1-\frac{j-1}{k}\right) \begin{pmatrix}3k \\ j-1\end{pmatrix}\right] - \left[\begin{pmatrix}3k \\ j\end{pmatrix} - 2 \begin{pmatrix}3k \\ j-1\end{pmatrix}\right] \nonumber \\
= & \left(3-\frac{j-1}{k}\right)\begin{pmatrix}3k \\ j-1\end{pmatrix} - \frac{j}{k} \begin{pmatrix}3k \\ j\end{pmatrix} \nonumber \\
= & \frac{1}{k} \left[(3k-j+1) \frac{(3k)!}{(j-1)!(3k-j+1)!} - j \frac{(3k)!}{j!(3k-j)!}\right] \nonumber \\ = & 0.
\end{align}
Therefore, $D_k(L=3k)$ reduces to Eq.~\eqref{eq:num-of-all-mobile-strings}.
\end{proof}

\subsection{Cyclic qutrit model: Connectivity of Single-triplet Krylovs}

\begin{theorem}
If $N^a,N^b\geq 1$ and $N^c=1$, then all the strings that are not frozen with the same $(N^0,N^1,N^2,D)$ lie in the same orbit. \label{thm:single-triplet-family-1-connectivity}
\end{theorem}

Since we consider non-frozen strings, and that $N^c=1$, any such string must be representable as $Xw$ or $Yw$, where $w$ contains only $a$ and $b$.
Let us consider starting with a string $Xw$. Then any string $Xw^\prime$ that lies in its Krylov must be attainable by repeatedly going through $Xw \to Y \tilde w \to X \tilde{\tilde w}$, where $\tilde w$ is $w$ with an $ba$ replaced by $ab$, and $\tilde{\tilde w}$ is $\tilde w$ with an $ab$ replaced by $ba$. Therefore, $\tilde{\tilde{w}}$ would always be equal to choosing an $ab$ pattern and a $ba$ pattern (they have to be disjoint) and then swapping them. We call this the non-local dipole-conserving update rule. For convenience, we will define the dipole moment of $a$ in $w$ as $\mu(w)=\sum_i a_i^w$, in which $a^w_1,\dots,a^w_m$ are the positions of $a$ characters in a string $w$, sorted in ascending order. For example, if $w=abbab$, then $m=2$, $a^w_1=1$ and $a^w_2=4$. The dipole moment is related to $D$ as follows. Since $w$ contains only $a$ and $b$, the general expression for $D$ reduces to $D_w = \sum_{i<j}(n_i^b n_j^a - n_i^a n_j^b)$, which counts $(ba)$ pairs minus $(ab)$ pairs. For each $a$ at position $a^w_k$, there are $a^w_k - k$ copies of $b$ preceding it, so
\begin{equation}
D_w = 2\sum_{k=1}^{m}(a^w_k - k) - m(\ell - m) = 2\mu(w) - m(\ell+1)\,,
\end{equation}
where $\ell = |w|$ is the length of the string $w$. By Eq.~\eqref{eq:D-concatenation}, $D_{Xw} = D_w - 1$ (similarly $D_{Yw} = D_w + 1$). Writing $\ell = L - 3$ and $m = N^a - 1$, we obtain
\begin{equation}
\mu(w) = \frac{D \pm 1 + (N^a - 1)(L-2)}{2}\,,
\end{equation}
where the sign is $+1$ for $Xw$ and $-1$ for $Yw$. Therefore, the dipole moment $\mu$ is essentially $D$ shifted by another constant.

\begin{theorem}
Two strings $w_1$ and $w_2$ consisting of characters $a$ and $b$ can be cast into each other via repeated application of the non-local dipole-conserving update rule if and only if they are of the same length, have the same number of $a$ characters, and the same dipole moment for $a$.
\end{theorem}
\begin{proof}
``Only if'' is obvious. To show ``if'', we show that we can design a greedy algorithm to convert string $w_2$ into $w_1$ if they have the same conserved quantities. Let both strings be of length $L$, have $m$ characters of $a$, and the positions of $a$ characters in a string $w$ be $a^w_1,\dots,a^w_m$. Let us define the distance between two strings as $d(w_1,w_2)=\sum_i |a_i^{w_1}-a_i^{w_2}|$. We show that whenever $d(w_1,w_2)>0$, we can find an update on $w_2$ such that the resulting string $w_2^\prime$ has $d(w_1,w_2^\prime)<d(w_1,w_2)$.

To this end, define $J=\{j|a_j^{w_1}>a_j^{w_2}\}$ and $K=\{k|a_k^{w_1}<a_k^{w_2}\}$. Then, suppose we choose $j _\ast\in J$ and $k_\ast\in K$, and we swap $ab$ at positions $(a_{j_\ast}^{w_2},a_{j_\ast}^{w_2}+1)$ with $ba$ at positions $(a_{k_\ast}^{w_2}-1,a_{k_\ast}^{w_2})$ for string $w_2$, then the resulting string would have a strictly smaller distance with $w_1$. We then show such $j_\ast$ and $k_\ast$ must exist, and $|a_{j_\ast}^{w_2}+1-a_{k_\ast}^{w_2}| \geq 2$.

In fact, suppose that $j_\ast \in J$ and $j_\ast+1\notin J$, then $a_{j_\ast}^{w_2} < a_{j_\ast}^{w_1} < a_{{j_\ast}+1}^{w_1} \leq a_{{j_\ast}+1}^{w_2}$, meaning that $a_{{j_\ast}+1}^{w_2} - a_{j_\ast} ^{w_2} > 1$, so the character in $w_2$ at position $a_{j_\ast}^{w_2}+1$ must be $b$. If $m\in J$, then $a_m^{w_2}+1\leq L$, and it must contain a $b$, so $m$ is also a valid choice of $j_\ast$. Therefore, there exists at least one $j_\ast\in J$ such that $(a_{j_\ast}^{w_2},a_{j_\ast}^{w_2}+1)$ is $ab$.  Similarly, if $k_\ast\in K$ and $k_\ast-1\notin K$, or $k_\ast = 1 \in K$, then $(a_{k_\ast}^{w_2}-1,a_{k_\ast}^{w_2})$ must be $ba$.

We then show that there exists at least one choice of $j_\ast$ and $k_\ast$ such that $|a_{j_\ast}^{w_2}+1-a_{k_\ast}^{w_2}|\geq 2$. Notice that $J$ and $K$ are distinct, $j_\ast \neq k_\ast$, so $a_{j_\ast}^{w_2}+1-a_{k_\ast}^{w_2}\neq 1$. Also, since we have shown $a_{k_\ast}^{w_2}-1$ is a $b$, $a_{j_\ast}^{w_2}-(a_{k_\ast}^{w_2}-1)\neq 0$. Also, if $j<k$, then $a_{j}^{w_2}<a_j^{w_1}<a_k^{w_1}<a_k^{w_2}$, which means that $a_j^{w_2} < a_k^{w_2}-2$, therefore $a_{j_\ast}^{w_2}+1-a_{k_\ast}^{w_2}\neq -1$. This establishes the desired result.
\end{proof}

This proof can be naturally generalized to show that the moment-$\mu$ modifications of a base string lie in the same Krylov in the second family of strings.

\subsection{Cyclic qutrit model: Multi-triplet sectors}
\label{sec:multi-triplet-argument}

In this subsection, we argue that all non-frozen, non-single-triplet strings with the same $(N^0,N^1,N^2,D)$ lie in the same Krylov subspace. Since the rewriting rules preserve $(N^0,N^1,N^2,D)$, strings with different invariants are certainly disconnected. The non-trivial direction is to show that strings sharing the same invariants \emph{are} connected. We do this by constructing a reduced representation of the dynamics and sketching a reduction to canonical form.

For any string, take one of its canonical form $X^{m_X} Y^{m_Y} w$. We can always treat $w$ as a suffix akin to the construction in the second family single-triplet sectors: let it be obtained from a base string $a^{m_1} b^{m_2} \dots$ by $\mu $ transpositions. We will lift the constraint $m_i \geq \mu+2$, and allow the $m_i$ sequence to extend in both directions. In this way, any string can be represented by a state $(m_X,m_Y,\mu,(m_i))$, where $m_X$ and $m_Y$ count the number of $X$ and $Y$ triplets, $(m_i)_{i\in\mathbb{Z}}$ is a finitely supported sequence of non-negative integers encoding the frozen part, and $\mu$ is the inversion number of the frozen string relative to a perfectly ordered configuration. All entries must remain non-negative. The invariants $(N^0,N^1,N^2,D)$ are determined by these variables: the total mass $M = m_X+m_Y+\frac{1}{3}\sum_i m_i$ and the mod-3 distribution of $(m_i)$ encode $N^0,N^1,N^2$, while $D$ is a function of $m_X,m_Y,\mu$ and the $(m_i)$.

The dynamics in this representation consists of four processes (assuming the base string is ordered as $abcabc\dots$; the opposite case is analogous):
\begin{enumerate}
    \item $m_X \to m_X-1$, $m_Y\to m_Y+1$, $\mu \to \mu-1$, $(m_i)$ unchanged;
    \item $m_X \to m_X+1$, $m_Y\to m_Y-1$, $\mu \to \mu+1$, $(m_i)$ unchanged (the inverse of 1);
    \item $m_X \to m_X-1$, $(m_{i-1},m_i,m_{i+1})\to (m_{i-1}+1,m_i+1,m_{i+1}+1)$, $\mu \to \mu+m_i$ (expanding a triplet into the frozen string);
    \item $m_X \to m_X+1$, $(m_{i-1},m_i,m_{i+1})\to (m_{i-1}-1,m_i-1,m_{i+1}-1)$, $\mu \to \mu+1-m_i$ (absorbing a triplet from the frozen string; inverse of 3).
\end{enumerate}
Processes~1 and~2 interconvert $X$ and $Y$ triplets. Processes~3 and~4 transfer mass between the triplet reservoir and the frozen string. In the single-triplet sector, process~4 is forbidden by the constraint $m_i\geq\mu+2$; relaxing this constraint is precisely what allows multi-triplet dynamics. Notice that all these processes can only happen when $m_X+m_Y \geq 1$, since for frozen strings there is no mobility at all.

We now sketch a reduction showing that any multi-triplet state can be brought to a canonical form depending only on $(N^0,N^1,N^2,D)$. First, if $m_Y > 0$, we can repeatedly apply process~2 until $m_Y=0$. This is always legal since $m_Y\geq 1$ is the only prerequisite; $\mu$ increases and remains non-negative. By our assumption, one must have $m_X \geq 2$.

We will show that by repeatedly applying processes 3 and 4, we can reduce the support of $(m_i)$ to contain only $(m_1,m_2,m_3)$, at most. This suffices to prove our statement, as the fragmentation in the second-family single-triplet sector is exactly due to the fact that we can have arbitrarily long $(m_i)$.

In fact, consider applying process 3 on sites $(m_i,m_{i+1},m_{i+2})$, then applying process 4 on sites $(m_{i+1},m_{i+2},m_{i+3})$. This would result in $m_i\mapsto m_i+1$, $m_{i+3}\mapsto m_{i+3}-1$. This process will result in $\mu \mapsto \mu + m_{i+1} - m_{i+2}$, and is only possible when the final $\mu$ is positive. If $m_{i+1} \geq m_{i+2}$, this is always possible. Therefore, whenever $m_{i+1} \geq m_{i+2}$, we can always repeatedly do this to send $(m_i,m_{i+1},m_{i+2},m_{i+3}) \mapsto (m_i+m_{i+3},m_{i+1},m_{i+2},0)$. Doing this on sites $i+1$ towards $i+4$, and using $m_{i+2} \geq 0$, this will send $(m_i,m_{i+1},m_{i+2},m_{i+3},m_{i+4}) \mapsto (m_i+m_{i+3},m_{i+1}+m_{i+4},m_{i+2},0,0)$. Now two consecutive zeros mean that $i+5$ and $i+2$ are essentially the same block. Therefore, we can merge them, ultimately achieving
\begin{equation}
(m_i,m_{i+1},m_{i+2},m_{i+3},m_{i+4},m_{i+5}) \mapsto (m_i+m_{i+3},m_{i+1}+m_{i+4},m_{i+2}+m_{i+5}).
\end{equation}
The $\mu$ value will not decrease throughout this process. A symmetric operation can be done if $m_{i+1} \leq m_{i+2}$. Repeatedly applying this, we can reduce the support of $(m_i)$ to contain only $(m_1,m_2,m_3)$.

\subsection{The \texorpdfstring{$m$}{m}GOE gap-ratio distribution}
\label{app:mgoe}
 
When $m$ independent GOE blocks are superposed without resolving block labels, the gap-ratio distribution $P_m(r)$ interpolates between GOE ($m=1$) and Poisson ($m\to\infty$). We present the full derivation following Ref.~\cite{Giraud2022}.

\begin{theorem} Let $X_1,\dots,X_m$ be independent stationary unfolded spectra on $\mathbb{R}$ with intensities $\mu_i > 0$, $\sum_{i=1}^m \mu_i = 1$. Assume all blocks share the same local statistics up to rescaling by their densities. Let $p_2(u,v)$ denote the joint density of the left and right consecutive gaps around a typical level of a single unfolded block, and let $p_1(L) \equiv \int_0^\infty p_2(L,v)\,dv$ be the one-gap marginal density obtained by integrating out the second spacing. Define
\begin{equation}
f(s) \equiv \int_s^\infty p_1(L)\,dL, \qquad
g(s) \equiv \int_s^\infty (L-s)\,p_1(L)\,dL,
\end{equation}
and
\begin{equation}
h(s,t) \equiv \int_s^\infty du \int_t^\infty dv\, p_2(u,v).
\end{equation}
For the mixed spectrum $X = \bigcup_{i=1}^m X_i$, let $H_m(x,y)$ be the probability that, around a typical mixed-spectrum level $E$, there is no other mixed-spectrum level in the interval $(E-x, E+y)$. Then
\begin{equation}
\label{eq:Hm}
H_m(x,y) = \sum_{i=1}^m \mu_i\, h(\mu_i x,\, \mu_i y) \prod_{j \neq i} g\!\bigl(\mu_j(x{+}y)\bigr).
\end{equation}
\end{theorem}

\begin{proof}
Starting from the joint gap density $p_2(s,t)$ of the left and right spacings around a typical level, the quantities entering the main formula are obtained by successive marginalizations, illustrated in Fig.~\ref{fig:fgh}.
\begin{figure}[t]
\centering
\begin{tikzpicture}[
  x=0.85cm, y=1cm,
  >=Stealth,
  level/.style={fill=black, circle, inner sep=0pt, minimum size=4.5pt},
  ghost/.style={draw=black!50, circle, inner sep=0pt, minimum size=4.5pt, thick, densely dashed},
  refpt/.style={fill=black, shape=diamond, inner sep=0pt, minimum size=5pt},
  brace/.style={decorate, decoration={brace, amplitude=3.5pt, raise=3pt}},
  fixedzone/.style={fill=blue!10},
  marginzone/.style={pattern=north east lines, pattern color=red!25},
  annot/.style={font=\scriptsize},
]
 
\def\rowsep{2.5}
\def\LBL{-2.8}
\def\FRM{10.5}
 
\begin{scope}[yshift=0cm]
  \node[font=\small\bfseries, anchor=east] at (\LBL, 0) {$p_2(s,t)$};
  \draw[thick, black!25] (0.3, 0) -- (8.2, 0);
  \fill[fixedzone] (1,-0.2) rectangle (3.5,0.2);
  \fill[fixedzone] (3.5,-0.2) rectangle (6.5,0.2);
  \node[level, label={[annot,below=3pt]prev}] at (1,0) {};
  \node[level, label={[annot,below=3pt]ref level}] at (3.5,0) {};
  \node[level, label={[annot,below=3pt]next}] at (6.5,0) {};
  \draw[brace] (1, 0.3) -- node[above=6pt, annot] {$s$} (3.5, 0.3);
  \draw[brace] (3.5, 0.3) -- node[above=6pt, annot] {$t$} (6.5, 0.3);
  \node[anchor=west, annot, text width=4.2cm] at (\FRM, 0)
    {joint density of left and right gaps around a level};
\end{scope}
 
\begin{scope}[yshift=-\rowsep cm]
  \node[font=\small\bfseries, anchor=east] at (\LBL, 0) {$p_1(s)$};
  \draw[thick, black!25] (0.3, 0) -- (8.2, 0);
  \fill[fixedzone] (1,-0.2) rectangle (3.5,0.2);
  \fill[marginzone] (3.5,-0.2) rectangle (7.0,0.2);
  \node[level, label={[annot,below=3pt]prev}] at (1,0) {};
  \node[level, label={[annot,below=3pt]ref level}] at (3.5,0) {};
  \node[ghost, label={[annot,below=3pt,red!60!black]next}] at (7.0,0) {};
  \draw[brace] (1, 0.3) -- node[above=6pt, annot] {$s$} (3.5, 0.3);
  \draw[brace] (3.5, 0.3) -- node[above=6pt, annot, red!60!black] {$v$} (7.0, 0.3);
  \draw[->, red!60!black, thick] (7.3, 0) -- (8.0, 0)
    node[right, font=\tiny, red!60!black] {$\int_0^\infty\! dv$};
  \node[anchor=west, annot] at (\FRM, 0)
    {$\displaystyle = \int_0^\infty p_2(s,v)\,dv$};
\end{scope}

\begin{scope}[yshift=-2*\rowsep cm]
  \node[font=\small\bfseries, anchor=east] at (\LBL, 0) {$f(s)$};
  \draw[thick, black!25] (0.3, 0) -- (8.2, 0);
  \fill[fixedzone] (1,-0.2) rectangle (4.0,0.2);
  \node[above, font=\tiny, blue!50!black] at (2.5, 0.2) {empty};
  \fill[marginzone] (4.0,-0.2) rectangle (6.5,0.2);
  \node[level, label={[annot,below=3pt]level}] at (1,0) {};
  \node[ghost, label={[annot,below=3pt,red!60!black]neighbor}] at (6.5,0) {};
  \draw[brace] (1, 0.3) -- node[above=6pt, annot] {$s$} (4.0, 0.3);
  \draw[brace] (4.0, 0.3) -- node[above=6pt, annot, red!60!black] {$a$} (6.5, 0.3);
  \draw[->, red!60!black, thick] (6.8, 0) -- (8.0, 0)
    node[right, font=\tiny, red!60!black] {$\int_0^\infty\! da$};
  \node[anchor=west, annot] at (\FRM, 0)
    {$\displaystyle = \int_0^\infty p_1(s{+}a)\,da$};
\end{scope}

\begin{scope}[yshift=-3*\rowsep cm]
  \node[font=\small\bfseries, anchor=east] at (\LBL, 0) {$g(s)$};
  \draw[thick, black!25] (-0.5, 0) -- (8.5, 0);
  \fill[marginzone] (0.2,-0.2) rectangle (2.2,0.2);
  \fill[fixedzone] (2.2,-0.2) rectangle (5.7,0.2);
  \node[above, font=\tiny, blue!50!black] at (3.95, 0.2) {empty};
  \fill[marginzone] (5.7,-0.2) rectangle (7.7,0.2);
  \node[ghost, label={[annot,below=3pt,red!60!black]left level}] at (0.2,0) {};
  \node[refpt, label={[annot,below=3pt]generic pt}] at (2.2,0) {};
  \node[ghost, label={[annot,below=3pt,red!60!black]right level}] at (7.7,0) {};
  \draw[brace] (0.2, 0.3) -- node[above=6pt, annot, red!60!black] {$a$} (2.2, 0.3);
  \draw[brace] (2.2, 0.3) -- node[above=6pt, annot] {$s$} (5.7, 0.3);
  \draw[brace] (5.7, 0.3) -- node[above=6pt, annot, red!60!black] {$b$} (7.7, 0.3);
  \draw[->, red!60!black, thick] (-0.1, 0) -- (-0.8, 0)
    node[left, font=\tiny, red!60!black] {$\int_0^\infty\! da$};
  \draw[->, red!60!black, thick] (8.0, 0) -- (8.7, 0)
    node[right, font=\tiny, red!60!black] {$\int_0^\infty\! db$};
  \node[anchor=west, annot] at (\FRM, 0)
    {$\displaystyle = \!\int_0^\infty\!\! da\!\int_0^\infty\!\! db\; p_1(s{+}a{+}b)$};
\end{scope}

\begin{scope}[yshift=-4*\rowsep cm]
  \node[font=\small\bfseries, anchor=east] at (\LBL, 0) {$h(s,t)$};
  \draw[thick, black!25] (-0.5, 0) -- (8.5, 0);
  \fill[marginzone] (0.2,-0.2) rectangle (1.5,0.2);
  \fill[fixedzone] (1.5,-0.2) rectangle (3.7,0.2);
  \node[above, font=\tiny, blue!50!black] at (2.6, 0.2) {empty};
  \fill[fixedzone] (3.7,-0.2) rectangle (6.0,0.2);
  \node[above, font=\tiny, blue!50!black] at (4.85, 0.2) {empty};
  \fill[marginzone] (6.0,-0.2) rectangle (7.5,0.2);
  \node[ghost, label={[annot,below=3pt,red!60!black]prev}] at (0.2,0) {};
  \node[level, label={[annot,below=3pt]ref level}] at (3.7,0) {};
  \node[ghost, label={[annot,below=3pt,red!60!black]next}] at (7.5,0) {};
  \draw[brace] (0.2, 0.3) -- node[above=6pt, annot, red!60!black] {$a$} (1.5, 0.3);
  \draw[brace] (1.5, 0.3) -- node[above=6pt, annot] {$s$} (3.7, 0.3);
  \draw[brace] (3.7, 0.3) -- node[above=6pt, annot] {$t$} (6.0, 0.3);
  \draw[brace] (6.0, 0.3) -- node[above=6pt, annot, red!60!black] {$b$} (7.5, 0.3);
  \draw[->, red!60!black, thick] (-0.1, 0) -- (-0.8, 0)
    node[left, font=\tiny, red!60!black] {$\int_0^\infty\! da$};
  \draw[->, red!60!black, thick] (7.8, 0) -- (8.7, 0)
    node[right, font=\tiny, red!60!black] {$\int_0^\infty\! db$};
  \node[anchor=west, annot] at (\FRM, 0)
    {$\displaystyle = \!\int_0^\infty\!\! da\!\int_0^\infty\!\! db\; p_2(s{+}a,\,t{+}b)$};
\end{scope}

\begin{scope}[yshift=-4*\rowsep cm - 1.6cm]
  \node[level] at (0.5,0) {};
  \node[anchor=west] at (0.8, 0) {fixed level};
  \node[refpt] at (3.5,0) {};
  \node[anchor=west] at (3.8, 0) {generic point};
  \node[ghost] at (9.0,0) {};
  \node[anchor=west, red!60!black] at (9.2, 0) {neighbor (integrated out)};
  \fill[fixedzone] (0.5,-0.55) rectangle (0.9,-0.25);
  \draw (0.5,-0.55) rectangle (0.9,-0.25);
  \node[anchor=west] at (1.1, -0.4) {empty};
  \fill[marginzone] (5.5,-0.55) rectangle (5.9,-0.25);
  \draw (5.5,-0.55) rectangle (5.9,-0.25);
  \node[anchor=west, red!60!black] at (6.1, -0.4) {marginalized};
\end{scope}
 
\end{tikzpicture}
\caption{Marginalization hierarchy from the joint gap density $p_2(s,t)$ to the functions entering the mixed-spectrum formula in Eq.~\eqref{eq:Hm}.}
\label{fig:fgh}
\end{figure}
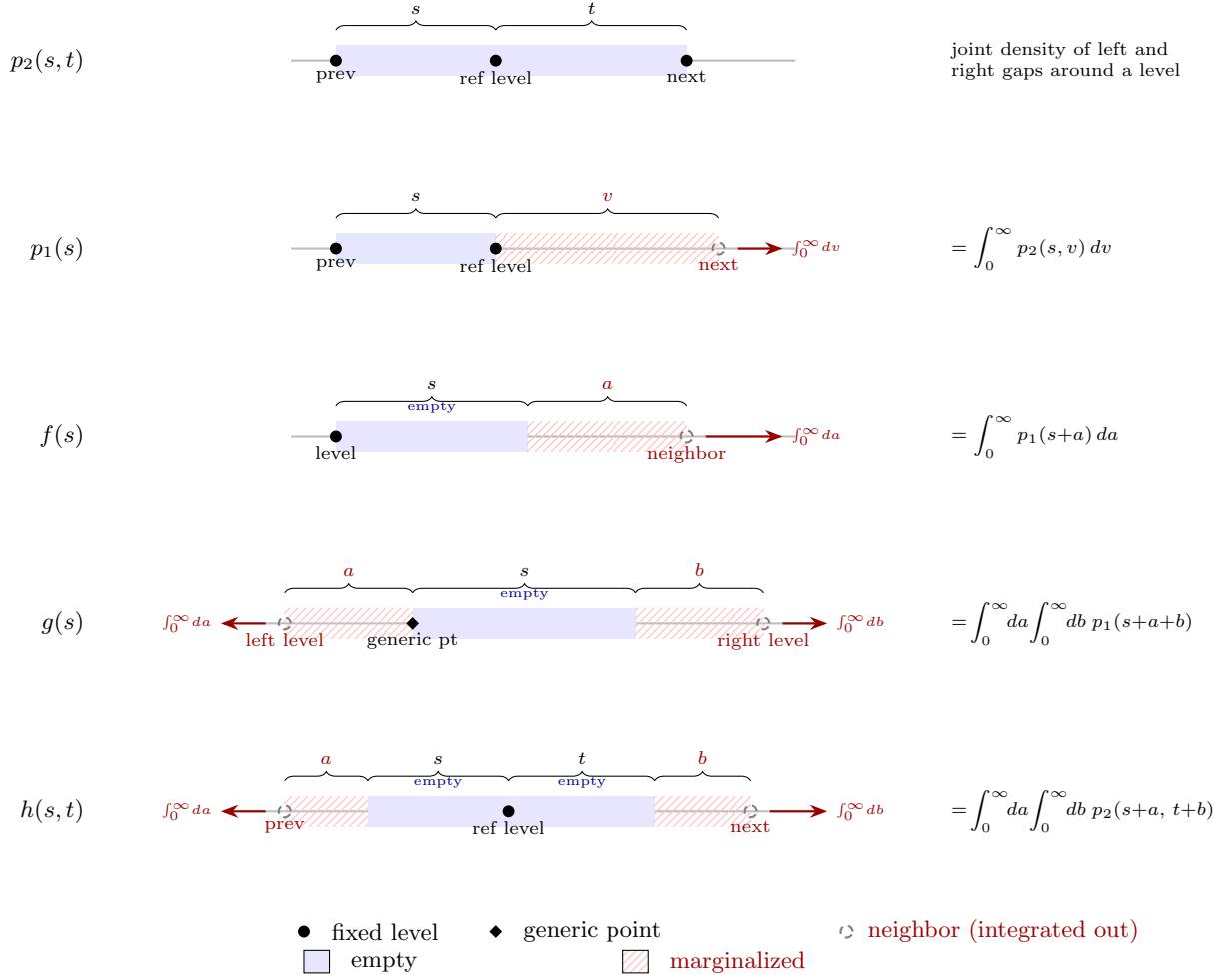
The function 
\begin{equation}
p_1(s) = \int_0^\infty p_2(s,v)\,dv
\end{equation}
describes the marginalized gap distribution where the right gap is integrated out. The can be thought as a effective single gap distribution as long as one does not care about the existence of the other gap. The function 
\begin{equation}
f(s) = \int_0^\infty p_1(s{+}a)\,da
\end{equation}
calculates the probability when the gap exceeds the given value $s$, with the excess $a$ integrated out. At this stage, the three-level picture has collapsed to a single gap and the reference point is the left energy level. Now, we are going to define the function
\begin{equation}
    g(s) = \int_0^\infty da \int_0^\infty db\, p_1(s{+}a{+}b)
\end{equation}
which calculates the probability of seeing a gap with length $s$ starting from a generic reference point not coinciding with the energy level. Therefore both margins $a$ and $b$ to the enclosing levels are integrated out. As we will see later, this is the key function for blocks that do not contain the reference energy level. Finally, we introduce the function
\begin{equation}
h(s,t) = \int_0^\infty da \int_0^\infty db\, p_2(s{+}a,\,t{+}b) ~,
\end{equation}
that calculates the probability starting from the reference energy level and having no same-block neighbor within $s$ to the left and $t$ to the right. This is the key function for the block that does contain the reference energy.

Let $X = \bigcup_{i=1}^m X_i$. Choose a typical level $E$ of the mixed spectrum, and let $M \in \{1,\dots,m\}$ be the label of the block containing this level. Since block $i$ has intensity $\mu_i$, the probability that a typical mixed-spectrum level comes from block $i$ is $\Pr(M=i) = \mu_i$. Fix $i$ and condition on $M=i$. The reference energy $E$ is a level of block~$i$ but a generic point inside a gap for every other block $j\neq i$. Therefore, for block $i$, the probability of no other $i$-level in $(E{-}x,\, E{+}y)$ is $h(\mu_i x,\, \mu_i y)$. While for each block $j \neq i$, the probability of no $j$-level in the whole interval $(E{-}x,\, E{+}y)$ of length $x+y$ is $g(\mu_j(x{+}y))$. We then immediately arrive at
\begin{equation}
    H_m(x,y) = \sum_{i=1}^m \mu_i\, h(\mu_i x,\, \mu_i y)\prod_{j \neq i} g\!\bigl(\mu_j(x{+}y)\bigr).
\end{equation}
\end{proof}

The joint density of two consecutive spacings is $P_m(x,y) = \partial_x \partial_y H_m(x,y)$. The gap ratio $r = \min(x,y)/\max(x,y) \in [0,1]$ is obtained by marginalizing over the overall scale:
\begin{equation}
\label{eq:Pm_r}
P_m(r) = 2\int_0^\infty dx\, x\, P_m(x,\, rx).
\end{equation}
For GOE distribution, $P_m$ then can be straightforwardly computed from 
\begin{equation}
    p_2(s,t) = \frac{27}{8\pi}\,st(s{+}t)\,\exp\!\left[-\frac{9}{4\pi}(s^2 + st + t^2)\right], \quad s,t > 0.
\end{equation}

\subsection{Saddle-point derivation of the EFS entanglement scaling}
\label{eq:subsec-efs-entanglement}

This appendix gives an independent rigorous derivation of $S \sim \sqrt{L}$ for the EFS in the all-mobile sector at the symmetric bisection $L_A = L/2$. Specializing the Schmidt weights of Eq.~\eqref{eq:Schmidt-weights} to $c_f = e$ and $L_A = L_B = L/2$,
\begin{equation}
\label{eq:saddle-S-start}
S = -\sum_{c \in G} \frac{D_c(L/2)\,D_{c^{-1}}(L/2)}{D(L)}\,\log \frac{D_c(L/2)\,D_{c^{-1}}(L/2)}{D(L)},
\end{equation}
where $D_c(\ell)$ counts length-$\ell$ binary strings reducing to $c \in G$.

\paragraph*{Step 1: Reduce to a sum over $(k_1, k_2)$.}
$D_c(\ell)$ depends only on $|c|$, equivalently the number of triplets $k = (\ell - |c|)/3$, with the closed form Eq.~\eqref{eq:Dk-closed}: $D_k(\ell) = \binom{\ell}{k} - \sum_{j=0}^{k-1}\binom{\ell}{j}$. Group the sum by $k_1 \equiv (L/2 - |c|)/3$ and $k_2 \equiv (L/2 - |c^{-1}|)/3$. The number of $c \in G$ with given $|c| = \ell_1$ and $|c^{-1}| = \ell_2$ is the number of reduced words with prescribed singleton/doubleton block counts:
\begin{equation}
|c| = \#\text{singletons} + 2\,\#\text{doubletons}, \qquad |c^{-1}| = 2\,\#\text{singletons} + \#\text{doubletons},
\end{equation}
so the count must contain $(2\ell_2 - \ell_1)/3$ singletons and $(2\ell_1 - \ell_2)/3$ doubletons, giving (with a factor $2$ for the choice of starting character)
\begin{equation}
\#\{c : |c| = \ell_1,\, |c^{-1}| = \ell_2\} = 2\binom{(\ell_1+\ell_2)/3}{(2\ell_1-\ell_2)/3}.
\end{equation}
Substituting $\ell_i = L/2 - 3k_i$,
\begin{equation}
S = -\sum_{k_1, k_2} 2\binom{L/3 - k_1 - k_2}{L/6 - 2k_1 + k_2}\, p_{k_1, k_2} \log p_{k_1, k_2}, \qquad
p_{k_1, k_2} = \frac{D_{k_1}(L/2)\,D_{k_2}(L/2)}{D_{L/3}(L)}.
\end{equation}

\paragraph*{Step 2: Stirling expansion near the saddle.}
Change variables to $\Delta_i \equiv L/6 - k_i$, so $\Delta_i = 0$ corresponds to the maximum-$k$ (most-mobile) sector. Stirling's formula gives, for $N \gg 1$,
\begin{equation}
\label{eq:N-delta-exponential}
\binom{N}{N/3 - \Delta} \approx \binom{N}{N/3}\, 2^{-\Delta}\, e^{-9\Delta^2/(4N)},
\end{equation}
up to $O(1)$ multiplicative corrections. Setting $\Delta_1 = x\Delta L$, $\Delta_2 = (1-x)\Delta L$ and applying Eq.~\eqref{eq:N-delta-exponential} to both numerator and denominator yields
\begin{equation}
p_{k_1, k_2} \approx \frac{1}{\sqrt{L}}\, \exp\left\{-L\left[(\log 2)\Delta + \tfrac{9}{2}\bigl(x^2 + (1-x)^2\bigr)\Delta^2\right]\right\}.
\end{equation}
The combinatorial multiplicity is also exponentially suppressed. The combinatorial factor becomes:
\begin{equation}
\binom{\Delta_1 + \Delta_2}{2\Delta_1 - \Delta_2} \approx \exp\left\{-L\bigl[(3x-1)\log(3x-1) + (2-3x)\log(2-3x)\bigr]\,\Delta\right\}.
\end{equation}
Combining,
\begin{equation}
\binom{\Delta_1 + \Delta_2}{2\Delta_1 - \Delta_2}\, p_{k_1, k_2} \approx \frac{1}{\sqrt{L}}\, e^{-L\, F(x, \Delta)},
\end{equation}
with action
\begin{equation}
F(x, \Delta) = \bigl[\log 2 + (3x-1)\log(3x-1) + (2-3x)\log(2-3x)\bigr]\,\Delta + \tfrac{9}{2}\bigl(x^2 + (1-x)^2\bigr)\Delta^2.
\end{equation}

\paragraph*{Step 3: Saddle point.}
The action $F(x, \Delta)$ has a strictly positive coefficient for $\Delta$ (linear) and $\Delta^2$ (quadratic) on its valid domain, so $F \geq 0$ with equality only at $\Delta = 0$. The saddle point is therefore at $\Delta = 0$, $x = 1/2$, where the integrand peaks. Around this point the quadratic action gives
\begin{equation}
\langle \Delta^m \rangle \sim L^{-m/2},
\end{equation}
and since $-\log p_{k_1, k_2} \sim L \cdot F(x, \Delta) \sim L \cdot \Delta$ near the saddle, we obtain
\begin{equation}
S = \langle -\log p \rangle \;\sim\; L \cdot \langle \Delta \rangle \;\sim\; L \cdot L^{-1/2} \;=\; \sqrt{L}.
\end{equation}
This rigorously establishes the $\sqrt{L}$ scaling derived heuristically in Sec.~\ref{subsec:sqrtL} via the bridge-walk argument.
\section{Numerical results}

In this section, we present the numerical algorithm we used to obtain the Krylov subspace structure, as well as the numerical results at small systems sizes.

\subsection{Algorithm for finding classical Krylov subspaces}

Finding the classical Krylov subspaces of a model is equivalence to finding the orbits of semigroup words under the equivalence relations. We exemplify this with the cyclic qutrit model.

\emph{Setup. ---}
Let $\Sigma = \{0,1,2\}$ and consider the semigroup acting on $\Sigma^L$ generated by local rewrite rules arising from cyclic permutation equivalences $012 = 120 = 201$ and $021 = 102 = 210$. A rewrite replaces any contiguous 3-letter window matching a non-canonical pattern with its canonical form ($012$ for even permutations, $021$ for odd). Two words $w, w' \in \Sigma^L$ lie in the same \emph{orbit} if $w$ can be transformed into $w'$ by a finite sequence of such rewrites. The goal is to compute the orbit size distribution $\{(s, c_s)\}$, where $c_s$ counts the number of orbits of size $s$.

\emph{Index Arithmetic. ---}
Each word $w = w_0 w_1 \cdots w_{L-1} \in \Sigma^L$ is identified with the integer
\[
\mathrm{idx}(w) = \sum_{i=0}^{L-1} w_i \cdot 3^{L-1-i} \in \{0, \ldots, 3^L - 1\}.
\]
A rewrite at position $i$ replaces $(w_i, w_{i+1}, w_{i+2})$ with $(w_i + d_0,\, w_{i+1} + d_1,\, w_{i+2} + d_2)$, where the delta vector $(d_0, d_1, d_2)$ depends on the pattern:
\[
\begin{aligned}
120 \to 012 &: \quad (-1,\, -1,\, +2), \\
201 \to 012 &: \quad (-2,\, +1,\, +1), \\
102 \to 021 &: \quad (-1,\, +2,\, -1), \\
210 \to 021 &: \quad (-2,\, +1,\, +1).
\end{aligned}
\]
The index of the rewritten word is then
\[
\mathrm{idx}(w') = \mathrm{idx}(w) + \Delta_i(c), \qquad \Delta_i(c) = d_0 \cdot 3^{L-1-i} + d_1 \cdot 3^{L-2-i} + d_2 \cdot 3^{L-3-i},
\]
where $c = 9w_i + 3w_{i+1} + w_{i+2}$ encodes the 3-digit window. The table $\Delta_i(c)$ is precomputed for all positions $i \in \{0, \ldots, L-3\}$ and codes $c \in \{0, \ldots, 26\}$.

\emph{Algorithm. ---} 
Computing orbits amounts to partitioning $\Sigma^L$ into equivalence classes under the transitive closure of the rewrite relation. We use the \emph{Union-Find} (disjoint-set) data structure, which maintains a partition of $\{0, \ldots, n-1\}$ under two operations:
\begin{itemize}
\item $\textsc{Find}(x)$: return a canonical \emph{representative} of the set containing $x$.
\item $\textsc{Unite}(x, y)$: merge the sets containing $x$ and $y$ into a single set.
\end{itemize}
Two elements $x, y$ belong to the same set if and only if $\textsc{Find}(x) = \textsc{Find}(y)$.

\paragraph{Representation.}
Each element $x$ stores a \emph{parent} pointer $\pi(x)$. A \emph{root} is an element with $\pi(x) = x$; it serves as the representative of its set. The parent pointers form a forest: following the chain $x \to \pi(x) \to \pi(\pi(x)) \to \cdots$ from any element eventually reaches the root of its tree.

\paragraph{Find with path compression.}
$\textsc{Find}(x)$ follows the parent chain to the root $r$, then sets $\pi(y) \gets r$ for every node $y$ visited along the path. This flattens the tree so that subsequent queries on the same nodes run in $O(1)$.

\paragraph{Unite by rank.}
Each root $r$ maintains a \emph{rank} $\rho(r)$, initially $0$, which upper-bounds the height of the tree rooted at $r$. To merge the sets of $x$ and $y$: find their roots $r_x, r_y$; if $r_x = r_y$, do nothing; otherwise, attach the lower-rank root under the higher-rank one. If ranks are equal, choose one as the new root and increment its rank. This keeps trees shallow.

With this structure, we use the following algorithm to find the orbits.

\begin{algorithm}[H]
\caption{Orbit size distribution}
\begin{algorithmic}[1]
\Require String length $L \geq 3$
\State Initialize Union-Find $U$ on $\{0, \ldots, 3^L - 1\}$
\State Precompute $\Delta_i(c)$ for $0 \leq i \leq L-3$, $0 \leq c \leq 26$
\For{each $w \in \Sigma^L$ (enumerated as a base-3 odometer)}
    \For{$i = 0$ \textbf{to} $L - 3$}
        \State $c \gets 9w_i + 3w_{i+1} + w_{i+2}$
        \If{$\Delta_i(c) \neq 0$}
            \State $U.\textsc{Unite}(\mathrm{idx}(w),\; \mathrm{idx}(w) + \Delta_i(c))$
        \EndIf
    \EndFor
\EndFor
\end{algorithmic}
\end{algorithm}

\subsection{Algorithm for finding quantum Krylov subspaces}

As we have established, quantum fragmentation happens when classical Krylov subspaces become reducible and further reduce into smaller entangled subspaces. Therefore, one the classical Krylov structure is obtained, we can infer the quantum fragmentation structure by looking at the reduced Hamiltonian within a classical Krylov subspace. This allows us to scale up our computation.

We employ the integer characteristic polynomial factorization (ICPF) method to detect the general Krylov subspace structure. A description of the method can be found in Refs.~\cite{Regnault:2022ocy,Chen:2026aqj}.


\subsection{Krylov subspace structure of the triplet flip model}

The following tables show the Krylov subspace structure of the triplet flip model for $q=2$ and $q=3$, at system sizes $L=4$ towards $L=9$. The tables are organized by $D_\text{cl}$, sizes of classical Krylov spaces. $d_\text{cl}$ is the number of such subspaces. The table then presents how such classical subspace splits into smaller entangled subspaces. $3^{\otimes 2}$ means $2$ subspaces of dimension $3$.

\twocolumngrid

\begin{table}[!h]
\centering
\caption{$q=2$ triplet flip model, $L = 4$}
\begin{tabular}{cc|cc}
\hline
$D_\text{cl}$ & $d_\text{cl}$ & Non-symmetric & $\mathbb{Z}_2$ symmetric \\
\hline
$1$ & $10$ & N/A & N/A \\
$3$ & $2$ & $1 \oplus 2$ & $1 \oplus 2$ \\
\hline
\end{tabular}
\end{table}

\begin{table}[!h]
\centering
\caption{$q=2$ triplet flip model, $L = 5$}
\begin{tabular}{cc|cc}
\hline
$D_\text{cl}$ & $d_\text{cl}$ & Non-symmetric & $\mathbb{Z}_2$ symmetric \\
\hline
$1$ & $16$ & N/A & N/A \\
$4$ & $4$ & $1 \oplus 3$ & $1 \oplus 3$ \\
\hline
\end{tabular}
\end{table}

\begin{table}[!h]
\centering
\caption{$q=2$ triplet flip model, $L = 6$}
\begin{tabular}{cc|cc}
\hline
$D_\text{cl}$ & $d_\text{cl}$ & Non-symmetric & $\mathbb{Z}_2$ symmetric \\
\hline
$1$ & $26$ & N/A & N/A \\
$5$ & $6$ & $1 \oplus 4$ & $1 \oplus 4$ \\
$8$ & $1$ & $1 \oplus \boldsymbol{7}$ & $1 \oplus \boldsymbol{3} \oplus \boldsymbol{4}$ \\
\hline
\end{tabular}
\end{table}

\begin{table}[!h]
\centering
\caption{$q=2$ triplet flip model, $L = 7$}
\begin{tabular}{cc|cc}
\hline
$D_\text{cl}$ & $d_\text{cl}$ & Non-symmetric & $\mathbb{Z}_2$ symmetric \\
\hline
$1$ & $42$ & N/A & N/A \\
$6$ & $10$ & $1 \oplus 5$ & $1 \oplus 5$ \\
$13$ & $2$ & $1 \oplus 12$ & $1 \oplus 12$ \\
\hline
\end{tabular}
\end{table}

\begin{table}[!h]
\centering
\caption{$q=2$ triplet flip model, $L = 8$}
\begin{tabular}{cc|cc}
\hline
$D_\text{cl}$ & $d_\text{cl}$ & Non-symmetric & $\mathbb{Z}_2$ symmetric \\
\hline
$1$ & $68$ & N/A & N/A \\
$7$ & $16$ & $1 \oplus 6$ & $1 \oplus 6$ \\
$19$ & $4$ & $1 \oplus 18$ & $1 \oplus 18$ \\
\hline
\end{tabular}
\end{table}

\begin{table}[!h]
\centering
\caption{$q=2$ triplet flip model, $L = 9$}
\begin{tabular}{cc|cc}
\hline
$D_\text{cl}$ & $d_\text{cl}$ & Non-symmetric & $\mathbb{Z}_2$ symmetric \\
\hline
$1$ & $110$ & N/A & N/A \\
$8$ & $26$ & $1 \oplus 7$ & $1 \oplus 7$ \\
$26$ & $6$ & $1 \oplus 25$ & $1 \oplus 25$ \\
$38$ & $1$ & $1 \oplus \boldsymbol{37}$ & $1 \oplus \boldsymbol{18} \oplus \boldsymbol{19}$ \\
\hline
\end{tabular}
\end{table}

\begin{table}[!h]
\centering
\caption{$q=3$ triplet flip model, $L = 4$}
\begin{tabular}{cc|cc}
\hline
$D_\text{cl}$ & $d_\text{cl}$ & Non-symmetric & $S_3$ symmetric \\
\hline
$1$ & $66$ & N/A & N/A \\
$5$ & $3$ & $1^{\otimes 3} \oplus 2$ & $1^{\otimes 3} \oplus 2$ \\
\hline
\end{tabular}
\end{table}

\begin{table}[!h]
\centering
\caption{$q=3$ triplet flip model, $L = 5$}
\begin{tabular}{cc|cc}
\hline
$D_\text{cl}$ & $d_\text{cl}$ & Non-symmetric & $S_3$ symmetric \\
\hline
$1$ & $180$ & N/A & N/A \\
$7$ & $9$ & $1^{\otimes 4} \oplus 3$ & $1^{\otimes 4} \oplus 3$ \\
\hline
\end{tabular}
\end{table}

\begin{table}[!h]
\centering
\caption{$q=3$ triplet flip model, $L = 6$}
\begin{tabular}{cc|cc}
\hline
$D_\text{cl}$ & $d_\text{cl}$ & Non-symmetric & $S_3$ symmetric \\
\hline
$1$ & $492$ & N/A & N/A \\
$9$ & $24$ & $1^{\otimes 5} \oplus 4$ & $1^{\otimes 5} \oplus 4$ \\
$21$ & $1$ & $1^{\otimes 10} \oplus \boldsymbol{11}$ & $1^{\otimes 10} \oplus \boldsymbol{3} \oplus \boldsymbol{4^{\otimes 2}}$ \\
\hline
\end{tabular}
\end{table}

\begin{table}[!h]
\centering
\caption{$q=3$ triplet flip model, $L = 7$}
\begin{tabular}{cc|cc}
\hline
$D_\text{cl}$ & $d_\text{cl}$ & Non-symmetric & $S_3$ symmetric \\
\hline
$1$ & $1344$ & N/A & N/A \\
$11$ & $66$ & $1^{\otimes 6} \oplus 5$ & $1^{\otimes 6} \oplus 5$ \\
$39$ & $3$ & $1^{\otimes 17} \oplus \boldsymbol{22}$ & $1^{\otimes 17} \oplus \boldsymbol{5^{\otimes 2}} \oplus \boldsymbol{12}$ \\
\hline
\end{tabular}
\end{table}

\begin{table}[!h]
\centering
\caption{$q=3$ triplet flip model, $L = 8$. $61_A$ indicates frozen strings of the form $aa$, $61_B$ has frozen string of the form $ab$.}
\begin{tabular}{cc|cc}
\hline
$D_\text{cl}$ & $d_\text{cl}$ & Non-symmetric & $S_3$ symmetric \\
\hline
$1$ & $3672$ & N/A & N/A \\
$13$ & $180$ & $1^{\otimes 7} \oplus 6$ & $1^{\otimes 7} \oplus 6$ \\
$61$\textsubscript{A} & $3$ & $1^{\otimes 25} \oplus \boldsymbol{36}$ & $1^{\otimes 25} \oplus \boldsymbol{6^{\otimes 3}} \oplus \boldsymbol{18}$ \\
$61$\textsubscript{B} & $6$ & $1^{\otimes 25} \oplus \boldsymbol{36}$ & $1^{\otimes 25} \oplus \boldsymbol{6^{\otimes 2}} \oplus \boldsymbol{24}$ \\
\hline
\end{tabular}
\end{table}

\begin{table}[!h]
\centering
\caption{$q=3$ triplet flip model, $L = 9$}
\begin{tabular}{cc|cc}
\hline
$D_\text{cl}$ & $d_\text{cl}$ & Non-symmetric & $S_3$ symmetric \\
\hline
$1$ & $10032$ & N/A & N/A \\
$15$ & $492$ & $1^{\otimes 8} \oplus 7$ & $1^{\otimes 8} \oplus 7$ \\
$87$ & $24$ & $1^{\otimes 34} \oplus \boldsymbol{53}$ & $1^{\otimes 34} \oplus \boldsymbol{7^{\otimes 3}} \oplus \boldsymbol{32}$ \\
$183$ & $1$ & $1^{\otimes 65} \oplus \boldsymbol{118}$ & $1^{\otimes 65} \oplus \boldsymbol{7^{\otimes 7}} \oplus \boldsymbol{19} \oplus \boldsymbol{25^{\otimes 2}}$ \\
\hline
\end{tabular}
\end{table}

\clearpage
\onecolumngrid

\subsection{Krylov subspace structure of the cyclic qutrit model}

The following tables show the Krylov subspace structure of the cyclic qutrit model at system sizes $L=4$ towards $L=9$. Meaning of table entries are consistent with the previous section. The $\mathbb Z_3$ symmetric column indicates the projector model with $\alpha \neq \beta$, and $D_3$ symmetric means $\alpha=\beta$.

\twocolumngrid

\begin{table}[!h]
\centering
\caption{Cyclic qutrit model, $L = 4$}
\begin{tabular}{cc|cc}
\hline
$D_\text{cl}$& $d_\text{cl}$& $\mathbb{Z}_3$ symmetric & $D_3$ symmetric \\
\hline
$1$ & $51$ & N/A & N/A \\
$5$ & $6$ & $1^{\otimes 3} \oplus 2$ & $1^{\otimes 3} \oplus 2$ \\
\hline
\end{tabular}
\end{table}

\begin{table}[!h]
\centering
\caption{Cyclic qutrit model, $L = 5$}
\begin{tabular}{cc|cc}
\hline
$D_\text{cl}$ & $d_\text{cl}$ & $\mathbb{Z}_3$ symmetric & $D_3$ symmetric \\
\hline
$1$ & $123$ & N/A & N/A \\
$7$ & $12$ & $1^{\otimes 4} \oplus 3$ & $1^{\otimes 4} \oplus 3$ \\
$12$ & $3$ & $1^{\otimes 6} \oplus \boldsymbol{6}$ & $1^{\otimes 6} \oplus \boldsymbol{3^{\otimes 2}}$ \\
\hline
\end{tabular}
\end{table}

\begin{table}[!h]
\centering
\caption{Cyclic qutrit model, $L = 6$}
\begin{tabular}{cc|cc}
\hline
$D_\text{cl}$ & $d_\text{cl}$ & $\mathbb{Z}_3$ symmetric & $D_3$ symmetric \\
\hline
$1$ & $297$ & N/A & N/A \\
$9$ & $18$ & $1^{\otimes 5} \oplus 4$ & $1^{\otimes 5} \oplus 4$ \\
$16$ & $12$ & $1^{\otimes 8} \oplus 8$ & $1^{\otimes 8} \oplus 8$ \\
$21$ & $2$ & $1^{\otimes 10} \oplus 3 \oplus 4^{\otimes 2}$ & $1^{\otimes 10} \oplus 3 \oplus 4^{\otimes 2}$ \\
$36$ & $1$ & $1^{\otimes 14} \oplus \boldsymbol{6} \oplus 8^{\otimes 2}$ & $1^{\otimes 14} \oplus \boldsymbol{3^{\otimes 2}} \oplus 8^{\otimes 2}$ \\
\hline
\end{tabular}
\end{table}

\begin{table}[!h]
\centering
\caption{Cyclic qutrit model, $L = 7$}
\begin{tabular}{cc|cc}
\hline
$D_\text{cl}$ & $d_\text{cl}$ & $\mathbb{Z}_3$ symmetric & $D_3$ symmetric \\
\hline
$1$ & $717$ & N/A & N/A \\
$11$ & $30$ & $1^{\otimes 6} \oplus 5$ & $1^{\otimes 6} \oplus 5$ \\
$20$ & $24$ & $1^{\otimes 10} \oplus 10$ & $1^{\otimes 10} \oplus 10$ \\
$29$ & $6$ & $1^{\otimes 14} \oplus 15$ & $1^{\otimes 14} \oplus 15$ \\
$47$ & $6$ & $1^{\otimes 20} \oplus 27$ & $1^{\otimes 20} \oplus 27$ \\
$68$ & $3$ & $1^{\otimes 24} \oplus \boldsymbol{44}$ & $1^{\otimes 24} \oplus \boldsymbol{22^{\otimes 2}}$ \\
\hline
\end{tabular}
\end{table}

\begin{table}[!h]
\centering
\caption{Cyclic qutrit model, $L = 8$}
\begin{tabular}{cc|cc}
\hline
$D_\text{cl}$ & $d_\text{cl}$ & $\mathbb{Z}_3$ symmetric & $D_3$ symmetric \\
\hline
$1$ & $1731$ & N/A & N/A \\
$13$ & $48$ & $1^{\otimes 7} \oplus 6$ & $1^{\otimes 7} \oplus 6$ \\
$24$ & $36$ & $1^{\otimes 12} \oplus 12$ & $1^{\otimes 12} \oplus 12$ \\
$35$ & $18$ & $1^{\otimes 17} \oplus 18$ & $1^{\otimes 17} \oplus 18$ \\
$46$ & $12$ & $1^{\otimes 22} \oplus 24$ & $1^{\otimes 22} \oplus 24$ \\
$81$ & $12$ & $1^{\otimes 33} \oplus 48$ & $1^{\otimes 33} \oplus 48$ \\
$108$ & $3$ & $1^{\otimes 36} \oplus \boldsymbol{72}$ & $1^{\otimes 36} \oplus \boldsymbol{36^{\otimes 2}}$ \\
$144$ & $6$ & $1^{\otimes 48} \oplus 96$ & $1^{\otimes 48} \oplus 96$ \\
\hline
\end{tabular}
\end{table}

\begin{table}[!h]
\centering
\caption{Cyclic qutrit model, $L = 9$}
\begin{tabular}{cc|cc}
\hline
$D_\text{cl}$ & $d_\text{cl}$ & $\mathbb{Z}_3$ symmetric & $D_3$ symmetric \\
\hline
$1$ & $4179$ & N/A & N/A \\
$15$ & $78$ & $1^{\otimes 8} \oplus 7$ & $1^{\otimes 8} \oplus 7$ \\
$28$ & $48$ & $1^{\otimes 14} \oplus 14$ & $1^{\otimes 14} \oplus 14$ \\
$41$ & $36$ & $1^{\otimes 20} \oplus 21$ & $1^{\otimes 20} \oplus 21$ \\
$54$ & $12$ & $1^{\otimes 26} \oplus 28$ & $1^{\otimes 26} \oplus 28$ \\
$67$ & $24$ & $1^{\otimes 32} \oplus 35$ & $1^{\otimes 32} \oplus 35$ \\
$78$ & $3$ & $1^{\otimes 36} \oplus \boldsymbol{42}$ & $1^{\otimes 36} \oplus \boldsymbol{21^{\otimes 2}}$ \\
$80$ & $6$ & $1^{\otimes 38} \oplus 42$ & $1^{\otimes 38} \oplus 42$ \\
$123$ & $6$ & $1^{\otimes 49} \oplus 74$ & $1^{\otimes 49} \oplus 74$ \\
$136$ & $12$ & $1^{\otimes 55} \oplus 81$ & $1^{\otimes 55} \oplus 81$ \\
$156$ & $3$ & $1^{\otimes 50} \oplus \boldsymbol{106}$ & $1^{\otimes 50} \oplus \boldsymbol{53^{\otimes 2}}$ \\
$226$ & $12$ & $1^{\otimes 74} \oplus 152$ & $1^{\otimes 74} \oplus 152$ \\
$252$ & $2$ & $1^{\otimes 92} \oplus 53^{\otimes 2} \oplus 54$ & $1^{\otimes 92} \oplus 53^{\otimes 2} \oplus 54$ \\
$271$ & $6$ & $1^{\otimes 87} \oplus 184$ & $1^{\otimes 87} \oplus 184$ \\
$432$ & $2$ & $1^{\otimes 120} \oplus 103^{\otimes 2} \oplus 106$ & $1^{\otimes 120} \oplus 103^{\otimes 2} \oplus 106$ \\
\hline
\end{tabular}
\end{table}

\clearpage


\end{document}